\documentclass[11pt]{article}
\usepackage{authblk}
\title{Landauer's principle for trajectories of repeated interaction systems}

\author{Eric P. Hanson}
\affil{%
Dept. Applied Math and Theoretical Physics\\
Centre for Mathematical Sciences\\
University of Cambridge\\
Cambridge, CB3 0WA, United Kingdom
}

\author{Alain Joye}
\affil{%
Univ. Grenoble Alpes, CNRS, Institut Fourier\\
F-38\,000 Grenoble, France}

\author{Yan Pautrat}
\affil{%
Laboratoire de Math\'ematiques d'Orsay\\
Univ. Paris-Sud, CNRS, Universit\'e
Paris-Saclay\\91405 Orsay, France
}

\author{Renaud Raqu{\'e}pas}
\affil{%
Dept. of Mathematics and Statistics, McGill University\\
1005--805 rue Sherbrooke O., Montr\'eal (Qu\'ebec) ~H3A 0B9}

\AtBeginDocument{%
   \def\MR#1{}
}

\usepackage[margin=1in]{geometry}

\usepackage{amssymb,amsmath,amsthm,amsopn}
\usepackage{braket}
\usepackage{mathrsfs}
\usepackage{bbm}

\usepackage{xcolor}
\usepackage[colorlinks]{hyperref}

\usepackage{tikz}
\usepackage{pgfplots}
\usepackage{graphicx}
\usepackage{tikzscale}
   \usepackage{placeins}
   \usepackage{float}

\pgfplotsset{lineplot/.style = {no markers,
            grid = major,
            grid style = {opacity=1},
            ylabel near ticks,
            axis line style={opacity=0},
            tick pos = left,
            tick align = outside}}

\pgfplotsset{halfwidth/.style = {width = 75mm,height = 75mm,}}

\usepackage{mathtools}
\usepackage{cite}
\usepackage{enumitem}

\setlist[description]{leftmargin=*}

\usepackage{thmtools, thm-restate}
\usepackage{xparse}
\usepackage{everysel}
\allowdisplaybreaks

\newtheorem{lemma}{Lemma}[section]
\newtheorem{theorem}[lemma]{Theorem}
\newtheorem{proposition}[lemma]{Proposition}
\newtheorem{corollary}[lemma]{Corollary}
\theoremstyle{definition}
\newtheorem{definition}[lemma]{Definition}
\newtheorem{remark}[lemma]{Remark}
\newtheorem{remarks}[lemma]{Remarks}
\newtheorem{example}[lemma]{Example}

\makeatletter
\newcommand{\labitem}[2]{%
	\def\@itemlabel{\textbf{#1}}
	\item
	\def\@currentlabel{#1}\label{#2}}
\makeatother
\usepackage{multirow}

\def\block(#1,#2)#3{\multicolumn{#2}{c}{\multirow{#1}{*}{$ #3 $}}}
\newcommand{\numset}[1]{\mathbb{#1}}
	\newcommand{\C}{\numset{C}}
	\newcommand{\R}{\numset{R}}
	
	\newcommand{\Z}{\numset{Z}}
	\newcommand{\N}{\numset{N}}
	\def\cc{\numset{C}}
	\def\rr{\numset{R}}

	\def\nn{\numset{N}}

\newcommand{\one}{\mathrm{Id}}
\def\id{\mathrm{Id}}
\newcommand{\ind}{\mathbbm{1}}
\def\i{{\rm i}}
\def\e{{\rm e}}
\newcommand{\Exp}[1]{\mathrm{e}^{#1}}
\renewcommand{\d}{\operatorname{d}\!}
\newcommand{\inv}{^{-1}}

	\DeclareMathOperator{\Mat}{Mat}
	\DeclareMathOperator{\spr}{spr}
	\renewcommand{\sp}{\operatorname{sp}}

	\newcommand{\tr}{\operatorname{Tr}}
	\DeclareMathOperator{\ran}{Ran}

\newcommand{\spacefont}[1]{\mathcal{#1}}
	\renewcommand{\H}{\spacefont{H}}
	\newcommand{\B}{\spacefont{B}}
	\newcommand{\I}{\spacefont{I}}
	\newcommand{\D}{\spacefont{D}}

\newcommand{\sys}{\mathcal{S}}
\newcommand{\env}{\mathcal{E}}
\newcommand{\T}{\spacefont{T}}
\newcommand{\ealpha}{^{(\alpha)}}
\renewcommand{\L}{\mathcal{L}}

	\newcommand{\proj}{p}
	\newcommand{\Kraus}{\mathfrak K}
	\def\inte{\operatorname{int}}
	\def\clos{\operatorname{cl}}

\newcommand{\invar}{^\textnormal{inv}}
\newcommand{\init}{^\textnormal{i}}
\newcommand{\fin}{^\textnormal{f}}
\DeclareDocumentCommand \envstate { o } {%
  \IfNoValueTF {#1} {%
    \xi%
  }{%
    \xi_{#1}%
  }%
}
\DeclareDocumentCommand \sysstate { o } {%
  \IfNoValueTF {#1} {%
    \rho%
  }{%
    \rho_{#1}%
  }%
}

\newcommand{\ham}{h}
	\newcommand{\hsys}{\ham_{\sys}}
	\DeclareDocumentCommand \henv { o } {%
	  \IfNoValueTF {#1} {%
	    \ham_{\mathcal{E}}%
	  }{%
	    \ham_{\mathcal{E}_{#1}}%
	  }%
	}

\newcommand{\E}{\mathbb{E}}
\newcommand{\ee}{\mathbb{E}}

\newcommand{\sP}{\mathcal{P}}

\newcommand{\bP}{\mathbb{P}}
\newcommand{\pp}{\mathbb{P}}

\newcommand{\cK}{\mathcal{K}}
\newcommand{\cW}{\mathcal{W}}

\newcommand{\Ai}{{A\init}}
\newcommand{\Af}{{A\fin}}
\newcommand{\ai}{{a\init}}
\newcommand{\af}{{a\fin}}
\newcommand{\rvY}{\Delta y^{\text{tot}}_{T}}
\newcommand{\rvZ}{{-\Delta a_T + \Delta y^{\text{tot}}_{T}}}
\newcommand{\rvW}{\Delta a_T}
\newcommand{\oldsigmaT}{\sigma^\textnormal{tot}_{T}}
\newcommand{\rhoadiab}{\rho_\textnormal{adiab}}
\newcommand{\invalpha}{\mathrm{I}\ealpha}

\begin{document}
\maketitle

\begin{abstract}
	We analyze Landauer's principle for \textit{repeated interaction systems} consisting of a reference quantum system $\mathcal{S}$ in contact with an environment $\mathcal{E}$ which is a chain of independent quantum probes. The system $\mathcal{S}$ interacts with each probe sequentially, for a given duration, and the Landauer principle relates the energy variation of $\mathcal{E}$ and the decrease of entropy of $\mathcal{S}$ by the entropy production of the dynamical process. We consider refinements of the Landauer bound at the level of the \textit{full statistics} (FS) associated to a two-time measurement protocol of, essentially, the energy of $\mathcal{E}$.  The emphasis is put on the \textit{adiabatic regime} where the environment, consisting of $T \gg 1$ probes, displays variations of order $T^{-1}$ between the successive probes, and the measurements take place initially and after $T$ interactions. We prove a large deviation principle and a central limit theorem as $T \to \infty$ for the classical random variable describing the entropy production of the process, with respect to the FS measure. In a special case, related to a detailed balance condition,
	we obtain an explicit  limiting distribution of this random variable without rescaling. At the technical level, we obtain a non-unitary adiabatic theorem generalizing that of \cite{HJPR1} and analyze the spectrum of complex deformations of families of irreducible completely positive trace-preserving maps.
\end{abstract}

\section{Introduction} \label{sec:Introduction}

The present paper studies a refinement of Landauer's principle in terms of a two-time measurement protocol (better known as ``full counting statistics'') for repeated interaction systems, in an adiabatic regime. We describe shortly the various elements we study.

Landauer's principle is a universal principle commonly formulated as a lower bound for the energetic cost of erasing a bit of information in a fixed system $\sys$ by interaction with an environment $\env$ initially at thermal equilibrium. It was first stated by Landauer in \cite{La}. A recent, mathematically sound derivation (in \cite{RW13}, later extended to the case of infinitely extended systems in  \cite{JP14}) is based on the entropy balance equation, given by $\Delta S_{\sys}+ \sigma = \beta \Delta Q_{\env}$ where $\Delta S_{\sys}$ is the average decrease in entropy of $\sys$ during the process, $\Delta Q_{\env}$ the average increase in energy of $\env$, and $\beta$ is the inverse temperature of the environment\footnote{we will always set the Boltzmann constant to $1$, so that $\beta=1/\Theta$, $\Theta$ the temperature.}. The term $\sigma$ is called the entropy production of the process.  As it can be written as a relative entropy, the entropy production is non-negative  which yields the inequality $\Delta S_{\sys} \leq  \beta \Delta Q_{\env}$. One of the questions of interest regarding Landauer's principle concerns the saturation of that identity, i.e. the vanishing of $\sigma$. It is a general physical principle that when the system--environment coupling is described by a time-dependent Hamiltonian, the entropy production $\sigma$ vanishes in the adiabatic limit, that is, when the coupling between $\sys$ and $\env$ is a slowly varying time-dependent function. More precisely, if the typical time scale of the coupling is $T$, one considers the regime $T\to\infty$.

A repeated interaction system (or RIS) is a system where the environment consists of a sequence of ``probes'' $\env_k$, $k=1,\ldots,T$, initially in a thermal state at inverse temperature $\beta_k$, and $\sys$ interacts with $\env_k$ (and only $\env_k$) during the time interval $\big[ k\tau, (k+1)\tau\big)$. In such a situation, the entropy balance equation becomes $\sum_{k=1}^T \Delta S_{\sys} + \sum_{k=1}^T \sigma_k = \sum_{k=1}^T \beta_k \Delta Q_{\env,k}$, where each term with index $k$ corresponds to the interaction between $\sys$ and $\env_k$. We describe the repeated interaction system as an ``adiabatic RIS'' when the various parameters of the probes are sampled from sufficiently smooth functions on $[0,1]$ as the values at times $k/T$, $k=1,\ldots,T$. This is the setup that was studied in \cite{HJPR1}; there we showed that the total entropy production $\lim_{T\to\infty} \sum_{k=1}^T \sigma_k$ was finite only under the condition $X(s)=0$ for all $s\in[0,1]$, where $X(s)$ is a quantity depending on the probe parameters at time $s\in[0,1]$ which we discuss below, see~\eqref{eq:def-X}. The proof of this result relied mostly on a new discrete, non-unitary adiabatic theorem that allowed us to control a product of $T$ slowly varying completely positive, trace-preserving (CPTP) maps that represent the reduced dynamics acting on $\sys$.

A refinement of the above formulation of Landauer's principle is however possible using the so-called full counting statistics. Full  counting statistics were first introduced in the study of charge transport, and have met with success in the study of fluctuation relations and work in quantum mechanics (see Kurchan \cite{Ku00} and Tasaki \cite{Ta00}). An example of their use in improving Landauer's principle was given in \cite{BFJP,GCGPVP}. In the present situation, the formulation of Landauer's principle in terms of full counting statistics can be stated by defining random variables $\Delta s_{\sys}$ and $\Delta q_{\env_k}$ which are outcomes of simple physical experiments, which we now describe. In such an experiment, one initially measures the quantity $-\log \rho_\sys$ ($\rho_\sys$ is the state of the small system) and the energies $h_{\env_k}$ for each $k$ ($h_{\env_k}$ is the free Hamiltonian of $\env_k$), then lets the system interact with the chain of probes, then measures again the same quantities.
With the right sign conventions, the changes in these quantities are random variables which we denote $\Delta s_{\sys}$ and $\Delta q_{\env_k}$. Our refinement discusses the connections between the probability distributions of $\Delta s_{\sys}$ and $\sum_k \beta_k \Delta q_{\env_k}$. One can show that the expectations of these distributions are $\Delta S_{\sys}$ and $\sum_k \beta_k \Delta Q_{\env_k}$ respectively; there is, therefore, more information in these distributions than in the previously considered scalar quantities.

We consider an adiabatic repeated interaction system and study the limiting distributions of the above random variables as $T\to\infty$. Again, we show that in the case $X(s)\equiv 0$ we have the expected refinement of Landauer's principle, which is essentially that when $T\to\infty$, one has $\Delta s_{\sys}=\sum_k \beta_k \Delta q_{\env_k}$ almost-surely.
In the case $X(s)\not\equiv 0$, we show that $\sum_k \beta_k \Delta q_{\env_k}$ satisfies a law of large numbers, a central limit theorem, and a large deviation principle, all of these for the time scale $T$, and with explicit parameters. In particular, $\sum_k \beta_k \Delta q_{\env_k}$ is of order $T$, whereas $\Delta s_\sys$ is a bounded quantity. All results in the case $X(s)\not\equiv 0$ can actually be extended to the case where the probe observables measured at each step $k$ are not simply $\beta_k h_{\env_k}$ but a more general observable, or when the system observables are not $-\log\rho_\sys$.

We show in addition that the random variable $\varsigma_T= \sum_k \beta_k \Delta q_{\env_k}- \Delta s_{\sys}$ can be expressed as a relative information random variable between the probability measure describing the experiment outcomes, and the probability measure corresponding to a backwards experiment. Since we obtain a full large deviation principle for this random variable as $T\to \infty$, this connects these results with the appearance of the arrow of time (see \cite{ABL,BJPP}). We discuss in particular the appearance of symmetries in the moment generating functions, and their implications in terms of Gallavotti--Cohen type symmetries.

To study the limiting distributions, we relate their moment generating functions to products of deformations of the completely positive, trace-preserving maps representing the reduced dynamics. We study the peripheral spectrum and associated spectral projector of these deformed dynamics. However, because little can be said about the spectral data of those deformed maps, studying the asymptotics of these quantities requires an improvement of the adiabatic theorem of \cite{HJPR1}. These technical results, concerning the spectral study of deformations of CPTP maps, and the improved discrete non-unitary adiabatic theorem, are of independent interest, and we describe them in wider generality than required for our present endeavor.

This approach gives an improvement over~\cite{HJPR1} in various aspects. First of all, Theorem~\ref{thm:X-is-0-conv-MGF} (in the case $X(s)\equiv 0$) and Theorem~\ref{theo_ldp} together with Corollary~\ref{coro_LLN} and Theorem~\ref{theo_clt} (in the general case) characterize the limiting distributions of relevant random variables, whereas in~\cite{HJPR1} we only derived information about the behaviour of their expectations. We recover our former results (and more) about these expectations, as Theorem~\ref{thm:X-is-0-conv-MGF} implies in particular the convergence of $\lim_{T\to\infty} \sum_{k=1}^T \sigma_k$ to an explicit quantity when $X(s)\equiv0$, and Theorem \ref{theo_clt} gives the divergence of the same quantity under generic assumptions when $X(s)$ does not vanish identically. In addition, Corollary \ref{cor:adiab_for_alpha=0} gives an expression for the adiabatic evolution of any initial state. {Most of all, our adiabatic theorem can be applied to a wider range of situations, as illustrated here by its application to deformed dynamics.}
\smallskip

The structure of the present paper is as follows: in Section~\ref{sec_framework} we describe our general framework and notation, and more precisely we describe repeated interaction systems, Landauer's principle (for unitary evolutions), and full counting statistics. We describe our full counting statistics for probe observables $Y_k$ more general than just $\beta_k h_{\env_k}$, leading to random variables $\rvY=\sum_k \Delta y_k$, and we generalize  $\Delta s_{\sys,T}$ (emphasising the $T$ dependence in the notation), to random variables $\Delta a_T$ as well.
In Section \ref{sec_propertiesfullstats} we discuss the various properties of the full statistics random variables: we give an entropy balance equation ``at the level of trajectories'', i.e.\ almost-sure identities between the different random variables, relate the moment generating functions of e.g.~$\rvY$ to deformations of reduced dynamics, and give a general adiabatic result for products of these deformations.
In Section \ref{sec:XisZero} we describe the limiting distribution of the pair $(\rvY,\Delta s_{\sys,T})$ as $T\to\infty$ in the case $Y=\beta h_\env$ when $X(s)\equiv 0$.
In Section~\ref{sec:Xnonzero} we derive a large deviation principle for $\rvY$ in the general case, which in turn implies a law of large numbers and a central limit theorem. Our technical results regarding the peripheral spectrum and associated spectral projectors of deformations of completely positive, trace-preserving maps are given in Appendix \ref{sec_peripheralspectrum}. Our improved discrete, non-unitary adiabatic theorem is given in Appendix~\ref{app:DNUAT}. Various proofs are collected in Appendix \ref{sec:proofs_for_FS}.

\paragraph{Acknowledgements} The research of {E.H.} was partly supported by ANR contract ANR-14-CE25-0003-0 and the Cantab Capital Insitute for the Mathematics of Information (CCIMI) and he would like to thank Cambyse Rouz\'{e} for informative discussions. The research of {Y.P.} was supported by ANR contract ANR-14-CE25-0003-0. The research of {R.R.} was partly supported by NSERC, FRQNT and ANR project RMTQIT ANR-12-IS01-0001-01. {R.R.} would like to thank the Institut Fourier, where part of this research was carried, for its support and hospitality. {E.H.} and {R.R.} would like to also thank  the organizers of the \textit{Stochastic Methods in Quantum Mechanics} summer school (Autrans, July 2017) for the informative and hospitable event. All four authors would like to thank Vojkan Jak\v{s}i\'{c} for stimulating discussions regarding this project.

\section{General framework} \label{sec_framework}
In this section we will introduce our general framework. We will use the following notational conventions: for $\mathcal X$ a Banach space, we denote by $\B(\mathcal X)$ the space of bounded linear operators on $\mathcal X$ and by $\one$ the identity on $\mathcal X$. For $\H$ a Hilbert space, we denote by $\I_1(\H)$ the space of trace-class, linear operators on $\H$, and by $\D(\H)$ the set of of density matrices on~$\H$, i.e. elements of $\I_1(\H)$ which are non-negative  operators with trace one. We will freely use the word ``state'' for an element $\rho\in\D(\H)$, therefore identifying the density matrix and the linear map $\B(\mathcal H)\ni A\mapsto \rho(A)$. We say that a state is faithful if the density matrix $\rho$ is positive-definite.
Scalar products will generally be denoted by $\braket{\phi,\psi}$ and are respectively linear and antilinear in the right and left variable. We denote by $\ket{\psi}\!\bra{\phi}$ the map on the Hilbert space defined by $\kappa\mapsto \braket{\phi,\kappa}\psi$.

\subsection{Repeated interaction systems} \label{subsec_ris}

A quantum repeated interaction system (RIS) consists of a system~$\sys$ interacting sequentially with a chain $\env_1, \env_2, \dotsc$ of probes (or environments). This physical model can describe for example an electromagnetic cavity which undergoes repeated indirect measurement by probes; its physical archetype is the one-atom maser (see \cite{MWM}). For more detail we refer the reader to the review \cite{BJM08}.

We will describe the quantum system~$\sys$ by a finite-dimensional Hilbert space~$\H_\sys$, a (time-independent) Hamiltonian~$h_\sys=h_\sys^*\in \B(\H_\sys)$, and an initial state~$\sysstate\init \in \D(\H_\sys)$. Likewise, the $k$th quantum probe~$\env_k$ will be described by a finite dimensional Hilbert space~$\H_{\env,k}$, Hamiltonian~$\henv[k]=\henv[k]^* \in \B(\H_{\env,k})$, and initial state $\envstate[k]\init \in \D(\H_{\env,k})$. We will assume the probe Hilbert spaces $\H_{\env,k}$ are all identical, $\H_{\env,k} \equiv \H_\env$, and that the initial state of each probe is a Gibbs state at inverse temperature $\beta_k>0$:
$$
\envstate[k]\init = \frac{\Exp{-\beta_k \henv[k]}}{\tr(\Exp{-\beta_k \henv[k]})}.
$$
We will at times use $Z_{\beta,k}$ to denote the trace $\tr(\Exp{-\beta_k \henv[k]})$.

The state of the system~$\sys$ evolves by interacting with each probe, one at a time, as follows. Assume that after interacting with the first $k-1$ probes the state of the system is~$\sysstate[k-1]$. Then the system and the $k$th probe, with joint initial state $\sysstate[k-1] \otimes \envstate[k]\init$, evolve for a time $\tau$ via the free Hamiltonian plus interaction $v_k$ according to the unitary operator
$$
U_k := \exp\big({-\i}\tau (\hsys \otimes \id + \id \otimes \henv[k] +  v_k)\big),
$$
yielding a joint final state $U_k(\sysstate[k-1] \otimes \envstate[k]\init)U_k^*$. The probe~$\env_k$ is traced out, resulting in the system state
$$
\sysstate[k] := \tr_\env\big(U_k(\sysstate[k-1] \otimes \envstate[k]\init)U_k^*\big),
$$
where $\tr_\env$ is the partial trace over $\H_\env$, mapping $\I_1(\H_\sys\otimes \H_\env)$ to $\I_1(\H_\sys)$, with $\tr_\env(X\otimes Y)=\tr(Y)\,X$. We define similarly $\tr_\sys$, the partial trace over $\H_\sys$ and, for later use, also introduce $\envstate[k]\fin := \tr_\sys\big(U_k(\sysstate[k-1] \otimes \envstate[k]\init)U_k^*\big)$. The evolution of the system~$\sys$ during the $k$th step is given by the \emph{reduced dynamics}
\begin{align}
\L_k : \I_1(\H_\sys) &\to \I_1(\H_\sys) \notag \\
\eta &\mapsto \tr_\env\big(U_k(\eta \otimes \envstate[k]\init)U_k^*\big), \label{eq_defLk}
\end{align}
that is $\sysstate[k] = \L_k(\sysstate[k-1])$; remark that $\L_k$ maps $\D(\H_\sys)$ to $\D(\H_\sys)$. By iterating this evolution, we find that the state of the system~$\sys$ after $k$ steps is given by the composition
\begin{equation}\label{eq:final-state}
\sysstate[k] = (\L_k \circ \dotsb \circ \L_1)(\sysstate\init).
\end{equation}
We will often omit the parentheses and composition symbols. For more details about the dynamics of RIS processes in various regimes, see \cite{BJM14, HJPR1}. We now turn to energetic and entropic considerations on RIS, at the root of Landauer's principle.

\subsection{Landauer's principle and the adiabatic limit} \label{subsec_landauer}

In what follows, for $\eta,\zeta$ in $\D(\H)$, $S(\eta)$ denotes the von Neumann entropy of the state~$\eta$ and $S(\eta|\zeta)$ denotes the relative entropy between the states~$\eta$ and~$\zeta$:
\begin{equation} \label{eq_defentropies}
	S(\eta)={-\tr(\eta \log \eta)}, \qquad S(\eta|\zeta)=\tr\big(\eta(\log\eta -\log\zeta)\big).
\end{equation}
We recall that $S(\eta)\geq 0$ and $S(\eta|\zeta)\geq0$.
For each step~$k$ of the RIS process, we define the quantities \begin{gather*}
\Delta S_k := S(\sysstate[k-1]) - S(\sysstate[k])
,\\
\Delta Q_k := \tr_\env(\henv[k]\envstate[k]\fin) - \tr_\env( \henv[k] \envstate[k]\init),
\end{gather*}
that represent the decrease in entropy of the small system, and the increase in energy of probe~$k$, respectively, and
\begin{equation} \label{eq:def_sigmak}
\sigma_k := S\big(U_k (\sysstate[k-1]\otimes \envstate[k]\init) U_k^* |  \sysstate[k]\otimes \envstate[k]\init\big),
\end{equation}
the \emph{entropy production} of step~$k$. For notational simplicity, we at times omit the ``i'' superscript in $\envstate[k]\init$. Also, we omit tensored identities for operators acting trivially on the environment or on the system, whenever the context is clear.

These quantities are related through the \emph{entropy balance equation}
\begin{equation} \label{eq:step-balance}
\Delta S_k + \sigma_k = \beta_k \Delta Q_k,
\end{equation}
(see {e.g.} \cite{RW13} for this computation). This equation, together with $\sigma_k\geq 0$, {i.e.} the nonnegativity of the entropy production term, encapsulates the more general \emph{Landauer principle}: when a system undergoes a state transformation by interacting with a thermal bath, the average increase in energy of the bath is bounded below by $\beta^{-1}$ times the average decrease in entropy of the system. This principle was first presented in 1961 by Landauer \cite{La} and its saturation in quantum systems has more recently been investigated by Reeb and Wolf \cite{RW13} and Jak\v{s}i\'{c} and Pillet \cite{JP14}, the latter providing a treatment of the case of infinitely extended quantum systems.
\smallskip

If we consider a RIS with $T$ steps, then summing \eqref{eq:step-balance} over $k = 1, \dotsc, T$ yields the \emph{total entropy balance equation}
\begin{equation} \label{eq:balance}
\Delta S_{\sys,T}+\oldsigmaT  =  \sum_{k=1}^T \beta_k \Delta Q_k,
\end{equation}
where $\Delta S_{\sys,T}=S(\sysstate\init) - S(\sysstate\fin)$ and $\sysstate\fin = \sysstate_T$ is the state of~$\sys$ after the final step of the RIS process (see~\eqref{eq:final-state}) and
\begin{equation}\label{eq:def_sigma_tot}
\oldsigmaT := \sum_{k=1}^T \sigma_{k},
\end{equation}
is the expected \emph{total entropy production}.

In \cite{HJPR1}, the present authors analyzed the Landauer principle and its saturation in the framework of an adiabatic limit of RIS that we briefly recall here. We introduce the \emph{adiabatic parameter} $T \in \N$ and consider a repeated interaction process with $T$ probes, such that the parameters governing the $k$th probe and its interaction with $\sys$, namely $(\henv[k], \beta_k, v_k)$, are chosen by sampling sufficiently smooth functions as described by the following assumption. Below, we say that a function $f$ is $C^2$ on $[0,1]$ if it is $C^2$ on $(0,1)$, and its first two derivatives admit limits at $0^+$, $1^-$.
\begin{description}
	\labitem{ADRIS}{ADRIS} We are given a family of RIS processes indexed by an adiabatic parameter~$T\in\nn$ such that there exist $C^2$ functions $s\mapsto \henv(s)$, $\beta(s)$, $v(s)$ on $[0,1]$ for which
	\begin{equation} \nonumber
	\henv[k] = \ham_\env(\tfrac {k}{T}), \qquad \beta_{k}=\beta(\tfrac {k}{T}), \qquad  v_{k} =v(\tfrac {k}{T})
	\end{equation}
	for all $k = 1, \dotsc, T$ when the adiabatic parameter has value~$T$.
\end{description}

In this case, we may define
\begin{equation}
\begin{aligned} \label{eq_versionscontinues}
U(s)&=\exp \big({-\i}\tau \big(\ham_\env(s) + \henv(s) +  v(s)\big)\big),\\
\L(s)&=\tr_\env\big( U(s)\big(\cdot\otimes \,\envstate(s)\big)U(s)^*\big),
\end{aligned}
\end{equation}
where~$\envstate(s)$ is the Gibbs state at inverse temperature~$\beta(s)$ for the Hamiltonian~$\henv(s)$ and $\tau$ is kept constant. Then, $[0,1]\ni s \mapsto \L(s)$ is a  $\B(\I_1(\H_\sys))$-valued $C^2$ function, and $\L_{k}=\L(\tfrac {k}{T})$ when the adiabatic parameter has value~$T$. Note that for each $s \in [0,1]$, the map $\L(s)$ is completely positive (CP) and trace preserving (TP). For some results, we will need to make some extra hypotheses on the family $(\L(s))_{s \in [0,1]}$. We introduce such conditions:
\begin{description}
	\labitem{Irr}{Irr} For each $s \in [0,1]$, the map~$\L(s)$ is \emph{irreducible}, meaning that it has (up to a multiplicative constant) a unique invariant, which is a faithful state.
	\labitem{Prim}{Prim}  For each $s \in [0,1]$, the map~$\L(s)$ is \emph{primitive}, meaning that it is irreducible and $1$ is its only eigenvalue of modulus one.
\end{description}
We recall in Appendix \ref{sec_peripheralspectrum} equivalent definitions and implications of these assumptions. We recall in particular that the peripheral spectrum of  an irreducible completely positive, trace-preserving map is a subgroup of the unit circle. We denote by $z(s)$ the order of that subgroup for $\L(s)$.

In \cite{HJPR1}, the present authors used a suitable adiabatic theorem to characterize the large~$T$ behaviour of the total entropy production term~\eqref{eq:def_sigma_tot}, which monitors the saturation of the Landauer bound in the adiabatic limit (note that the terms in the sum~\eqref{eq:def_sigma_tot} are $T$-dependent through \ref{ADRIS}). Briefly, under suitable assumptions, convergence of $\oldsigmaT$ is characterized by the fact that the term $X(s)$ defined in \eqref{eq:def-X} below vanishes identically.

\subsection{Full statistics of two-time measurement protocols} \label{subsec:FS_two_time_measurements}

We now describe a two-time measurement protocol for repeated interaction systems with~$T$ probes. The outcome of this protocol is random, and we will relate its expectation to the quantities involved in the balance equation~\eqref{eq:balance}. Note that a similar protocol was considered in~\cite{HorPar} (see also~\cite{BJPP,BCJP}).

For the purpose of defining the full statistics measure for an RIS, we will consider observables to be measured on both the system~$\sys$ and the probes~$\env_k$, $k\in\nn$.

First, we assume we are given two observables $\Ai$ and $\Af$ in $\B(\H_\sys)$ with spectral decomposition
$$
\Ai = \sum_{\ai} \ai \, \pi\init_{\ai}, \qquad
\Af = \sum_{\af} \af \, \pi\fin_{\af}
$$
where~$\ai$, $\af$ run over the distinct eigenvalues of~$\Ai$, $\Af$ respectively, and~$\pi\init_{\ai}$, $\pi\fin_{\af}$ denote the corresponding spectral projectors. When we consider increasing the number of probes $T$, we assume the observable $\Ai$ is independent of $T$ (as we measure it on $\sys$ before the system interacts with any number of probes), but allow $\Af$ to depend on $T$, as long as the family $(\Af)_{T=1}^\infty$ is uniformly bounded in $T$.

On the chain, we consider probe observables $Y_k \in \B(\H_\env)$ to be measured on the probe $\env_k$. We require that each observable commutes with the corresponding probe Hamiltonian: $$[Y_k, \henv[k]] = 0.$$ We write the spectral decomposition of each $Y_k$ as
\[Y_k=\sum_{i_k} y_{i_k} \Pi^{(k)}_{i_k}.\]
If the $k$th probe is initially in the state $\xi$, a measurement of~$Y_k$ before the time evolution will  yield~$y_{i_k}$ with probability $\tr(\xi \Pi_{i_k}^{(k)})$.

When assuming \ref{ADRIS} and discussing measured observables $Y$, we will always assume
\begin{description}
	\labitem{Comm}{Comm}
	There is a twice continuously differentiable $\B(\H_\env)$-valued function $s \mapsto Y(s)$ on $[0,1]$ such that~$[Y(s), \henv(s)] = 0$ at all~$s \in [0,1]$ for which, when the adiabatic parameter has value~$T$,
	\[ Y_k = Y(\tfrac{k}{T}), \qquad k = 1, \dotsc, T. \]
\end{description}
The family of probe Hamiltonians themselves~$Y(s) = \henv(s)$ are suitable, but in our applications to Landauer's principle, we will be particularly interested in~$Y(s) = \beta(s) \henv(s)$.
\smallskip

Associated to the observables~$\Ai$,~$\Af$ and~$(Y_k)_{k=1}^T$ and the state~$\sysstate\init$, we define two processes: the \emph{forward process}, and the \emph{backward process}.

\paragraph*{The forward process} The system $\sys$ starts in some initial state~$\sysstate\init \in \D(\H_\sys)$ and the probe~$\env_k$ starts in the initial Gibbs state $\envstate[k]\in\D(\H_{\env_{k}})$; we write the state of the chain of~$T$ probes $\Xi = \bigotimes_{k=1}^T \envstate[k]$. We measure~$\Ai$ on~$\sys$ and measure~$Y_k$ on~$\env_k$ for each $k = 1, \dotsc, T$. We obtain results~$\ai$ and~$\vec{\imath} = (i_k)_{k=1}^T$ with probability
$$
\tr\big((\sysstate\init \otimes \Xi)(\pi\init_{\ai} \otimes \Pi_{\vec{\imath}})\big),
$$
where $\Pi_{\vec{\imath}} := \bigotimes_{k=1}^T \Pi_{i_k}^{(k)}$. Then the system interacts with each probe, one at a time, starting at $k=1$ until $k=T$, via the time evolution
\[
U_{k} := \exp\big(-\i\tau (\hsys + \henv[k] + v_{k})\big).
\]
Next, we measure~$\Af$ on the system and measure~$Y_k$ on~$\env_k$ for each $k = 1, \dotsc, T$, yielding outcomes~$\af$ and~$\vec \jmath =(j_k)_{k=1}^T$. Using the rules of measurement in quantum mechanics and conditional probabilities, the quantum mechanical probability of measuring the sequence~$(\ai,\af,\vec{\imath},\vec{\jmath})$ of outcomes is given by
$$
\tr\big( U_{T}\dotsm U_1   ( \pi\init_{\ai} \otimes \Pi_{\vec \imath}) (\rho\init \otimes \Xi)( \pi\init_{\ai} \otimes \Pi_{\vec \imath}) U_1^* \dotsm U_T^*  (\pi\fin_{\af}\otimes \Pi_{\vec \jmath})\big).
$$
We emphasize that the outcomes are labelled by $(\ai,\af,\vec{\imath},\vec{\jmath})$ which refers to the eigenprojectors of the operators involved, but not to the corresponding eigenvalues which only need to be distinct. Also, we may write the second measurement projector $\pi\fin_{\af}\otimes \Pi_{\vec \jmath}$ only once by cyclicity of the trace.

\paragraph{The backward process} The system starts in state
\[
\rho\fin_T := \tr_\env\big( U_{T}\dotsm U_1    (\sysstate\init \otimes \Xi) U_1^* \dotsm U_T^* \big),
\]
and the probe $\env_k$ starts in the state $\xi_k$. We measure observable $\Af$ on~$\sys$ and $Y_k$ on~$\env_k$ for each $k = 1, \dotsc, T$, yielding outcomes $\af$ and $(j_k)_{k=1}^T$. Then the system interacts with each probe, one at a time, starting with $k=T$ until $k=1$, via the time evolution
\[
U_k^* = \exp\big(\i\tau(\hsys + \henv[k] + v_k)\big).
\]
Then we measure~$\Ai$ on~$\sys$ and $Y_k$ on~$\env_k$ for each $k = 1, \dotsc, T$, yielding outcomes $\ai$ and $(i_k)_{k=1}^T$. The probability of these outcomes is given by
$$
\tr \big( U_1^*\dotsm U_T^*(\pi\fin_{\af} \otimes \Pi_{\vec \jmath})(\rho\fin_T \otimes \Xi) (\pi\fin_{\af} \otimes \Pi_{\vec \jmath}) U_T \dotsm U_1 (\pi\init_{\ai} \otimes \Pi_{\vec \imath})\big).
$$

\paragraph{The full statistics associated to the two-step measurement process}

For notational simplicity, we assume that the cardinality of~$\sp Y(s)$
does not depend on $k$. We can therefore use the same index set~$\mathfrak I$ for all eigenvalue sets: $\sp Y_k=(y_{i_k})_{i_k\in \mathfrak I}$ for all $k=1,\ldots,T$.
We define the space
\[
\Omega_T := \sp \Ai \times \sp \Af \times \mathfrak I^T \times \mathfrak I^{T}
\]
and equip it with the maximal $\sigma$-algebra $\sP(\Omega_T)$. We will refer to elements $(\ai,\af,\vec{\imath},\vec{\jmath})$ of~$\Omega_T$ as \emph{trajectories}, and denote them by the letter~$\omega$.

\begin{definition}
	On $\Omega_T$, we call the law of the outcomes for the forward process,
	\begin{equation}\label{eq:def_P_F^T}
	\bP^F_{T}(\ai,\af,\vec \imath,\vec \jmath) := \tr\big( U_{T}\dotsm U_1   ( \pi\init_{\ai} \otimes \Pi_{\vec \imath}) (\rho\init \otimes \Xi)( \pi\init_{\ai} \otimes \Pi_{\vec \imath}) U_1^* \dotsm U_T^*  (\pi\fin_{\af}\otimes \Pi_{\vec \jmath})\big),
	\end{equation}
	the forward \emph{full statistics measure}. We denote by~$\E_T$ the expectation with respect to~$\bP^F_{T}$. We also consider the backward full statistics measure
	\begin{equation}\label{eq:def_P_B^T}
	\bP^B_{T}(\ai,\af,\vec \imath,\vec \jmath) := \tr \big( U_1^*\dotsm U_T^*(\pi\fin_{\af} \otimes \Pi_{\vec \jmath})(\rho\fin_T \otimes \Xi) (\pi\fin_{\af} \otimes \Pi_{\vec \jmath}) U_T \dotsm U_1 (\pi\init_{\ai} \otimes \Pi_{\vec \imath})\big)
	\end{equation}
	which is the law of the outcomes for the backward process. Let us emphasize here that $\pp^F_{T}$ and~$\pp^B_{T}$ depend on the spectral projectors $(\pi\init_{\ai})_{\ai}$ of $\Ai$, $(\pi\fin_{\af})_{\af}$ of $\Af$, and $(\Pi_{\vec \imath})$ of the $(Y_k)_k$, and not on the spectral values of these operators. In particular, the probabilities $\pp^F_{T}$ and $\pp^B_{T}$ associated with two families of observables $(Y(s))_{s\in [0,1]}$, $(Y'(s))_{s\in [0,1]}$ that have the same spectral projectors (as e.g.\ $Y(s) = \beta(s) h_\env(s)$ and $Y'(s) = h_\env(s)$) will be the same.

	To $(Y_k)_{k=1}^T$, $\Ai$, and $\Af$, we associate two generic classical random variables on $(\Omega_T, \sP(\Omega_T))$:
	\begin{gather}
		\rvW(\ai, \af, \vec \imath, \vec \jmath) := {\ai-\af}, \label{eq_defW} \\
		\rvY(\ai, \af, \vec \imath, \vec \jmath) := \sum_{k=1}^T (y^{(k)}_{j_k} - y^{(k)}_{i_k}) \label{eq_defY}.
	\end{gather}
\end{definition}
Note that the choice of defining $\rvW$ as $\ai-\af$, i.e.\ as the decrease of the quantity $a$, is consistent with the standard formulation of Landauer's principle as given in Section \ref{subsec_landauer}. Additionally, the assumption that $(\Af)_{T=1}^\infty$ has uniformly bounded norm yields that the random variable $\rvW$ has $L^\infty$ norm  uniformly bounded in $T$.

\begin{remark} \label{remark_noconsistency}
	When we work with an \ref{ADRIS}, the dependence in $T$ of the $U_k$ (remember that in this case $U_k$ is of the form $U(k/T)$) prevents the family $(\pp_T^F)_T$ from being consistent. The $\pp_T^F$ are therefore a priori not the restrictions of a probability $\pp^F$ on the space $\Omega_\infty$, as is the case in \cite{BJPP} where the environments $\env_k$ and the parameters $h_{\env_k}$, $\beta_k$ and $v_k$ do not depend on $k$.
\end{remark}

\section{Properties of the full statistics} \label{sec_propertiesfullstats}

In the present section we obtain a relation between classical random variables arising from the protocol defined in Subsection \ref{subsec:FS_two_time_measurements}, and the quantity \eqref{eq:balance}. We study the relevant properties of the distributions~$\pp^F_{T}$ and~$\pp^B_{T}$, their relative information random variable, and its moment generating function.

\subsection{Entropy production and entropy balance on the level of trajectories} \label{subsec:bal-eq}

We turn to obtaining an analogue of~\eqref{eq:balance} for random variables on the probability space~$(\Omega_T, \sP(\Omega_T), \bP^F_{T})$.
Remark first that $\bP^F_{T}(\ai,\af,\vec{\imath},\vec{\jmath})$ and $\bP^B_{T}(\ai,\af,\vec{\imath},\vec{\jmath})$ are of the form
$\bP^F_{T}(\ai,\af,\vec{\imath},\vec{\jmath})=\tr\big((\rho\init\otimes \Xi) \, S^* S\big)$ and $ \bP^B_{T}(\ai,\af,\vec{\imath},\vec{\jmath})=\tr\big((\rho_T\fin\otimes \Xi) \, S S^*\big)$.
Under the assumption that  $\rho\init$ and $\rho_T\fin$ are faithful we therefore have
\[\bP^F_{T}(\ai,\af,\vec{\imath},\vec{\jmath})=0 \mbox{ if and only if }\bP^B_{T}(\ai,\af,\vec{\imath},\vec{\jmath})=0.\]
Since the image of a faithful state by an irreducible CPTP map is faithful (see the discussion following Definition \ref{def_irreducibility} below),  $\rho\init$ and $\rho_T\fin$ will be faithful as soon as $\rho\init$ is faithful and assumption~\ref{Irr} holds.

This allows us to give the following definition:
\begin{definition} \label{def:class-rv}
	If $\rho\init$ and $\rho_T\fin$ are faithful, we define the classical random variable
	\begin{gather*}
	\varsigma_{T}(\ai,\af,\vec{\imath},\vec{\jmath}) :=\log \frac{\bP^F_{T}(\ai,\af,\vec{\imath},\vec{\jmath})}{\bP^B_{T}(\ai,\af,\vec{\imath},\vec{\jmath})},
	\end{gather*}
	on $(\Omega_T, \sP(\Omega_T),\pp^F_{T})$, which we call the entropy production of the repeated interaction system associated to the trajectory~$\omega = (\ai,\af,\vec{\imath},\vec{\jmath})$.
\end{definition}

Note that the random variable $\varsigma_{T}$ is the logarithm of the ratio of likelihoods, also known as the relative information random variable between $\pp^F_{T}$ and $\pp^B_{T}$ (see e.g. \cite{CoverThomas}). It is well-known that the distribution of such a random variable is related to the distinguishability of the two distributions (here $\pp^F_{T}$ and $\pp^B_{T}$): see e.g.\ \cite{BicDok}. Distinguishing between $\pp^F_{T}$ and $\pp^B_{T}$ amounts to testing the arrow of time; we refer the reader to \cite{JOPS,BJPP} for a further discussion of this idea.

We have the following result, essentially present in~\cite{HorPar}, whose proof is given in Appendix~\ref{sec:proofs_for_FS}.
\begin{lemma}\label{lem:balance_eq_traj}
	Assume $\sysstate\init$ and $\sysstate\fin_T$ are faithful. If
	\begin{enumerate}[label=\roman*.]
		\item $\pi\init_{\ai}\sysstate\init\pi\init_{\ai} = \frac{\tr(\sysstate\init\pi\init_{\ai})}{\dim \pi\init_{\ai}}\,\pi\init_{\ai}$ for each $\ai$,
		\item $\pi\fin_{\af}\sysstate\fin_T\pi\fin_{\af} = \frac{\tr(\sysstate\fin_T\pi\fin_{\af})}{\dim \pi\fin_{\af}}\,\pi\fin_{\af}$ for each $\af$,
		\item for each $k = 1, \dotsc, T$, the state  $\envstate[k]$ (or equivalently $h_{\env_k}$) is a function of $Y_k$,
	\end{enumerate}
	then
	\begin{equation}\label{eq:balance_eq_for_trajectories}
	\varsigma_{T}(\ai,\af,\vec{\imath},\vec{\jmath}) = \log \Big( \frac{\tr(\pi\init_{\ai} \sysstate\init)}{ \tr(\pi\fin_{\af} \sysstate\fin_T)} \,\frac{\dim \pi\fin_{\af}}{\dim \pi\init_{\ai}}\Big) + \sum_{k=1}^T \beta_k (E_{j_k}^{(k)} - E_{i_k}^{(k)}),
	\end{equation}
	where $E_{i_k}^{(k)} = \frac{\tr(\henv[k]\Pi_{i_k}^{(k)})}{\dim \Pi_{i_k}^{(k)}}$ are the energy levels of the $k$th probe.
\end{lemma}

\begin{remark}
	The first two hypotheses are automatically satisfied if, for example, $\Ai$ and $\Af$ are non-degenerate (all their spectral projectors are rank-one). All three hypotheses are automatically satisfied if, for example, $\sysstate\init$, $\sysstate\fin_T$ and $\envstate[k]$ can be written as functions of $\Ai$, $\Af$ and $Y_k$ (for each $k =1, \dotsc, T$) respectively.
\end{remark}

Again, $\varsigma_{T}$ depends on the spectral projectors of the observables~$\Ai$,~$\Af$ and~$(Y_k)_{k=1}^T$, but not on their eigenvalues. However, with the choices~$\Ai = - \log \sysstate\init$, $\Af = - \log \sysstate\fin_T$ and~$Y(s) = \beta(s) \henv(s)$, and writing the spectral decompositions $\rho\init=\sum r_a\init \pi\init_a$, $\rho\fin_T=\sum r_a\fin \pi\fin_a$, the relation~\eqref{eq:balance_eq_for_trajectories} takes the simpler form of a sum of differences of the obtained eigenvalues (measurement results):
\begin{align*}
\varsigma_{T}(\omega) &= (-\log r_{a\fin}\fin)-(-\log r_{a\init}\init)  + \sum_{k=1}^T \beta_k (E_{j_k}^{(k)} - E_{i_k}^{(k)}),
\end{align*}
which is the random variable introduced earlier as $\rvZ(\omega)$ (again, in the case $Y=\beta h_\env$). In this case, $\rvW = (-\log r_{a\init}\init)-(-\log r_{a\fin}\fin)$ is a classical random variable that is the difference of measurements of entropy observables on the system~$\sys$, which we call $\Delta s_{\sys,T}(\omega)$. On the other hand, $\sum_{k=1}^T \beta_k (E_{j_k}^{(k)} - E_{i_k}^{(k)})$ is a classical random variable that encapsulates Clausius' notion of the entropy increase of the chain $(\env_k)_{k=1}^T$ on the level of trajectories, which we call $\Delta s_{\env,T}(\omega)$.  Then,
\begin{equation}\label{eq:decomp-entropy-prod-traj}
	\Delta s_{\sys,T}(\omega) +\varsigma_{T}(\omega) = \Delta s_{\env,T}(\omega),
\end{equation}
and $\varsigma_{T}(\omega)$ measures the difference between these two entropy variations, on the trajectory~$\omega$.

Moreover, Proposition~\ref{prop:averaged_balance_equation} {below}, whose proof is also left for the Appendix, links expression~\eqref{eq:balance_eq_for_trajectories} to the entropy balance equation \eqref{eq:balance}. Indeed, by showing that under suitable hypotheses the two terms on the right hand side of~\eqref{eq:balance_eq_for_trajectories} average to the corresponding terms in~\eqref{eq:balance}, we show that $\E_T( \varsigma_{T}) = \oldsigmaT$. In other words, $\oldsigmaT$ coincides with the relative entropy or Kullback--Leibler divergence $D(\bP^F_{T}||\, \bP^B_{T})$ between the classical distributions $\bP^F_{T}$ and $\bP^B_{T}$. Recall that $D(\bP^F_{T}||\, \bP^B_{T})=0$ if and only if $\bP^F_{T}=\bP^B_{T}$. Hence, we will refer to~\eqref{eq:balance_eq_for_trajectories} as the \emph{entropy balance equation on the level of trajectories}.

\begin{proposition} \label{prop:averaged_balance_equation}
	Assume that $\rho\init$ is faithful and a function of $\Ai$, that $\rho\fin_T$ is faithful and a function of $\Af$, and the state~$\envstate[k]$ (or equivalently~$h_{\env_k}$) is a function of~$Y_k$ for each~$k = 1, \dotsc, T$, then
	\begin{gather}
	\E_T\Big( \log \big( \frac{\tr(\pi\init_{\ai} \sysstate\init)}{ \tr(\pi\fin_{\af} \sysstate\fin)} \frac{\dim \pi\fin_{\af}}{\dim \pi\init_{\ai}}\big) \Big) = -\E_T(\Delta s_{\sys,T}) = S(\sysstate\fin) - S(\sysstate\init) \label{eq:ssys-avg}
	\intertext{and}
	\E_T\Big( \sum_{k=1}^T \beta_k (E_{j_k}^{(k)} - E_{i_k}^{(k)}) \Big) = \E_T(\Delta s_{\env,T}) = \sum_{k=1}^T \beta_k \Delta Q_k. \label{eq:senv-avg}
	\end{gather}
	Therefore,
	\begin{gather}
	\E_T( \varsigma_{T}) = \oldsigmaT, \label{eq:sprod-avg}
	\end{gather}
	and relation \eqref{eq:balance_eq_for_trajectories} reduces to the entropy balance equation \eqref{eq:balance} upon taking expectation with respect to~$\bP^F_{T}$.
\end{proposition}
Before we move on with our program, let us make a number of remarks on the choice of $\Ai = - \log \sysstate\init$ and $\Af = - \log \sysstate\fin_T$.
\begin{remarks}\label{rem:}\,\hfill
	\begin{itemize}
		\item We made the assumption above that the operator $\Af$ was uniformly bounded in $T$. This is true for $\Af = - \log \sysstate\fin_T$, as mentioned in Remark~\ref{rem:Af=logrhof_unif_bounded} below.
		\item The observable $-\log\sysstate\init$ is the analogue of the information random variable in classical information theory.
		\item The observable $-\log\sysstate\fin$ has the same interpretation and is the initial condition for the backward process in \cite{EHM,HorPar} but it might seem odd that the observer is expected to have access to $\sysstate\fin=\L_T\circ\ldots\circ \L_1(\sysstate\init)$. However, one can see that the reduced state of the probe after the forward experiment is random with $\bP^F_T$--expectation equal to $\sysstate\fin$. In addition, $\varsigma_T$ is a relative information random variable, and as such is relevant only to an observer who knows both distributions (here~$\bP^F_{T}$ and~$\bP^B_{T}$). Such an observer, knowing the possible outcomes for the random states after the experiment, and their distribution, would necessarily know their average $\sysstate\fin$.
	\end{itemize}
\end{remarks}
We are interested in the full statistics of the random variables $\varsigma_{T}(\omega)$ that we will address through its cumulant generating functions in the limit $T\to\infty$. We will consider two cases: $\lim_{T\to\infty}\oldsigmaT < \infty$, and $\lim_{T\to\infty}\oldsigmaT =\infty$. The behaviour of this averaged quantity was investigated in~\cite{HJPR1}. For a RIS satisfying the assumptions \ref{ADRIS} and \ref{Prim}, the condition
\[
\limsup_{T\to\infty} \oldsigmaT < \infty
\]
can be shown to be equivalent to the identity $X(s) \equiv 0$, where
\begin{equation}
X(s) := U(s) \big(\sysstate\invar(s) \otimes \xi\init(s)\big) U(s)^* - \sysstate\invar(s) \otimes \xi\init(s), \label{eq:def-X}
\end{equation}
and $\sysstate\invar(s)$ is the unique invariant state of $\L(s)$. If the assumption $X(s)\equiv 0$ does not hold, then $\lim_{T\to\infty}\oldsigmaT =\infty$. It was proven in \cite{HJPR1} that the condition $X(s)\equiv 0$ is equivalent to the existence of a family $(k_\sys(s))_{s \in [0,1]}$ of observables on $\H_\sys$ such that $[k_\sys(s)+h_{\env}(s),U(s)] \equiv 0$.
\smallskip

We will consider the case $X(s) \equiv 0$ in Section~\ref{sec:XisZero}, and the other case, $\sup_{s\in[0,1]}\|X(s)\|_1 > 0$, in Section~\ref{sec:Xnonzero}. In either case, our main object of interest will be the moment generating function of the variables $\rvY$ and $\rvW$, which we can relate to deformations $\L\ealpha(s)$ of $\L(s)$.

\subsection{Moment generating functions and deformed CP maps}
We recall that the quantities $\rvW$ and $\rvY$ are defined in \eqref{eq_defW} and \eqref{eq_defY}. We also recall that the \emph{moment generating function} (MGF) of a real-valued random variable $V$ (with respect to the probability distribution $\bP^F_{T}$, which will always be implicit in the present paper) is defined as the map $M_V:\alpha\mapsto \E_T\big(\e^{\alpha V}\big)$, and the MGF of a pair $(V_1,V_2)$ as the map $M_{(V_1,V_2)}:(\alpha_1,\alpha_2)\mapsto \E_T\big(\e^{\alpha_1 V_1+\alpha_2 V_2}\big)$. When $V$ or $(V_1,V_2)$ are given by the random variables $\rvY$, $\rvW$, the above functions $M_{V}$ (resp.\ $M_{(V_1,V_2)}$) are defined for all $\alpha\in\cc$ (resp.\ for all $(\alpha_1,\alpha_2)\in\cc^2$). For relevant properties of moment generating functions we refer the reader to Sections 21 and 30 of \cite{Bill}.

Our main tool to study these moment generating functions is the following proposition:
\begin{proposition} \label{prop:phiTalpha}
	For $\alpha\in\cc$, define an analytic deformation of~$\L(s)$ by the complex parameter~$\alpha$ corresponding to the observable~$Y(s)$:
	\begin{align}
	\L_Y^{(\alpha)}(s) :  \I_1(\H_\sys) &\to \I_1(\H_\sys) \nonumber \\
	 \eta &\mapsto \tr_\env\big(\e^{ \alpha Y(s) }  U(s)   (\eta \otimes \xi(s))\e^{-\alpha Y(s)} U(s)^*\big). \label{eq_defLk_alpha}
   \end{align}
	Under assumption \textup{\ref{Comm}}, the moment generating function of $\rvY$ is given by
	\[
	M_{\rvY}(\alpha) = \tr_\sys \big( \L_Y^{(\alpha)}(\tfrac{T}{T})  \dotsm \L_Y^{(\alpha)}(\tfrac{1}{T}) (\sum_{\ai} \pi\init_{\ai}\rho\init\pi\init_{\ai})\big).
	\]
	If in addition $[\Ai,\sysstate\init]=0$, then the moment generating function of the pair $(\rvY, \rvW)$
	is given~by
	\begin{gather*}
	M_{(\rvY,\rvW)}(\alpha_1,\alpha_2) = \tr \big(\Exp{-\alpha_2 \Af} \L_Y^{(\alpha_1)}(\tfrac{T}{T})  \dotsm \L_Y^{(\alpha_1)}(\tfrac{1}{T}) (\Exp{+\alpha_2 \Ai}\rho\init)\big).
	\end{gather*}
	so that in particular the moment generating function of $\rvZ$ is given by
	\[	M_{\rvZ}(\alpha) = \tr_\sys \big(\Exp{+\alpha \Af}  \L_Y^{(\alpha)}(\tfrac{T}{T})  \dotsm \L_Y^{(\alpha)}(\tfrac{1}{T})(\Exp{-\alpha \Ai}\rho\init)\big).\]
\end{proposition}

See Appendix~\ref{sec:proofs_for_FS} for the proof. In Section~\ref{sec:Xnonzero} we will analyze the above moment generating functions, with the help of an adiabatic theorem for the non-unitary discrete time operators~$\L_Y^{(\alpha)}$.  The case $Y(s)=\beta(s) h_{\env}(s)$ plays a particular role for the analysis of Landauer's principle. The complex deformation of the map~$\L(s)$ we consider is similar to the deformations introduced in \cite{HMO} for hypothesis testing on spin chains, and to the complex deformation of Lindblad operators introduced in \cite{JPW} suited to the study of entropy fluctuations for continuous time evolution.

Dropping the $s$-dependence from the notation, below, we first provide the expression for the adjoint of the deformation of $\L(s)$ with respect to the duality bracket on $\B(\H_\sys)$, $\braket{C_1,C_2}=\tr_\sys (C_1^*C_2)$. We temporarily make explicit the dependency of $\L_{Y}^{(\alpha)}$ in $\tau$ by denoting it $\L_{Y}^{(\alpha; \tau)}$; in particular, $\L_{Y}^{(\alpha;-\tau)}$ is obtained by replacing the unitary $U(s)$ with its adjoint $U^*(s)$.

\begin{lemma} \label{lem:adjoint}
	The adjoint of the operator $\L_Y^{(\alpha)}$ is given by
	\begin{equation}
	{\L_Y^{(\alpha)}}^*: \eta \mapsto \tr_\env\big(\Exp{-(\overline{\alpha}Y+\beta \henv)}  U^*   (\eta \otimes \xi)\, \Exp{(\overline{\alpha}Y+\beta \henv)} U\big).
	\end{equation}
	In particular, for $Y=\beta h_\env$ we have ${\L_{\beta h_\env}^{(\alpha; \tau)}}^*=\L_{\beta h_\env}^{(-\overline{\alpha}-1; -\tau)}$.
\end{lemma}
\begin{proof} First note the identity  $\tr_\env((\one\otimes C)D)=\tr_\env(D(\one\otimes C)$ for any operators $C$ and $D$ on $\H_\env$ and $\H_\sys\otimes \H_\env$, respectively.
	Let $C_1, C_2 \in \B(\H_\sys)$. The straightforward computation
	\begin{align*}
	\braket{\L_Y^{(\alpha; \tau)}(C_1), C_2}
	&= \tr_\sys\big(\big(\tr_\env(\Exp{\alpha Y} U^{(\tau)} (C_1 \otimes \xi)\,\Exp{-\alpha Y} U^{(-\tau)})\big)^* C_2\big) \\
	&= \tr_\sys\big(\tr_\env(U^{(\tau)} \Exp{-\overline{\alpha} Y} (C_1^* \otimes \xi) U^{(-\tau)} \Exp{\overline{\alpha} Y } ) \,C_2\big) \\
	&= Z_\beta^{-1} \tr\big((\Exp{\overline{\alpha} Y} U^{(\tau)} (C_1^* \otimes \Exp{-\beta \henv}) \Exp{-\overline{\alpha} Y}  U^{(-\tau)}   ) (C_2 \otimes \one)\big) \\
	&= Z_\beta^{-1} \tr\big( (C_1^* \otimes \one) \Exp{-(\overline{\alpha}Y+\beta \henv)}  U^{(-\tau)}  \\
      &\qquad\qquad\qquad (C_2 \otimes \Exp{-\beta \henv}) \Exp{(\overline{\alpha}Y+ \beta \henv)} U^{(\tau)}\big) \\
	&=  \tr_\sys\big( C_1^* \tr_\env(\Exp{-(\overline{\alpha}Y+\beta \henv)}  U^{(-\tau)}   (C_2 \otimes \xi) \Exp{(\overline{\alpha}Y+\beta \henv)} U^{(\tau)})\big)
	\end{align*}
	directly yields the result.
\end{proof}
{We now consider the Kraus form of this deformation.}
\begin{lemma}\label{lem:contr-CP}
Assume that \textup{\ref{Comm}} holds. For all $\alpha\in \rr$, $\L_Y^{(\alpha)}$ is a completely positive map. In addition, there exists a Kraus decomposition
\[	\L_Y(\eta) = \sum_{i,j}K_{i,j} \eta K_{i,j}^*\]
of $\L_Y$ such that $\L_Y\ealpha$ admits the Kraus decomposition
\[\L_Y^{(\alpha)}(\eta) = \sum_{i,j}\e^{\alpha (y_j-y_i) }K_{i,j} \eta K_{i,j}^*.\]
\end{lemma}

\begin{proof}
	Let $\{\psi_m\}_{m=1}^{\dim \H_\env}$ be an orthonormal basis of eigenvectors of $Y$. We introduce
	\begin{align*}
	&\one\otimes \ket{\psi_m}: \H_\sys\rightarrow \H_\sys\otimes \H_\env \ \ \mbox{defined by  } \ \varphi \mapsto \varphi\otimes \psi_m\\
	&\one\otimes \bra{\psi_m}:\H_\sys\otimes \H_\env\rightarrow \H_\sys\ \ \mbox{defined by  }\ \varphi\otimes \psi \mapsto  \braket{\psi_m, \psi}
	\varphi,
	\end{align*}
	so that $(\one\otimes \ket{\psi_m})^*=\one\otimes \bra{\psi_m}$. We observe that using $[Y,\xi]=0$
	{and $\tr_\env((\one\otimes C)D)=\tr_\env(D(\one\otimes C)$ for any operators $C$ and $D$ on $\H_\env$ and $\H_\sys\otimes \H_\env$, respectively},
	we can write
	\begin{equation}
	\label{eq:sym-L}
   \begin{split}
	\L_Y^{(\alpha)}(\eta) &= \tr_\env\big((\one\otimes \e^{ \alpha Y/2 })  U  (\one \otimes \e^{ -\alpha Y/2 }\xi^{1/2})  \\
      &\qquad\qquad\qquad
	(\eta \otimes \one)(\one \otimes \xi^{1/2}\e^{ -\alpha Y/2 })U^*(\one\otimes \e^{ \alpha Y/2 })\big).
   \end{split}
	\end{equation}
	This shows that $\L_Y\ealpha$ is a completely positive map. Then we express the partial trace on $\H_\env$ using the orthonormal basis $\{\psi_m\}_{m=1}^{\dim \H_\env}$ by means of the set of operators on $\H_\sys$
	\begin{align} \label{eq_formulaKij}
	K_{i,j}^{(\alpha)} :\!&= (\id \otimes \bra{\psi_j}) (\id \otimes \e^{\alpha Y /2 }) U (\id \otimes \e^{-\alpha Y/2}) (\id \otimes \xi^{1/2} \ket{\psi_i})
	\end{align}
	(again $K_{i,j}^{(\alpha)}$ depends on the choice of $Y$). Thus for any $\eta \in \I_1(\H_\sys)$, and all $\alpha\in \R$,
	\begin{equation} \label{eq_KrausLalpha}
		\L_Y^{(\alpha)}(\eta)  = \sum_{i,j} K_{i,j}^{(\alpha)} \eta (K_{i,j}^{({\alpha} )})^*.
	\end{equation}
	This yields the Kraus decomposition of $\L_Y^{(\alpha)}$.
	Moreover, we note that
	\[
	K_{i,j}^{(\alpha)} = \e^{\alpha (y_j-y_i) /2 } K_{i,j}^{(0)},
	\]
	and letting $ K_{i,j}:=K_{i,j}^{(0)}$ gives our final statement.
\end{proof}
Lemma \ref{lem:contr-CP} proves in particular that $\L_Y\ealpha$ is a deformation of $\L$ in the sense of Appendix~\ref{sec_peripheralspectrum}.
Let us now address the regularity of $\L_Y^{(\alpha)}(s)$ in $(s,\alpha)$.
\begin{lemma}\label{lem:C2-reg}
	Assume \textup{\ref{ADRIS}} and suppose $s \mapsto Y(s)\in C^2\big([0,1],\B(\H_\sys)\big)$. Then, the map $$[0,1]\times \C \ni(s,\alpha)\mapsto \L_Y^{(\alpha)}(s)\in \B(\H_\sys)$$ is of class $C^2$.
\end{lemma}

\begin{proof}
	First observe that since the dimensions of $\H_\sys$ and $\H_\env$ are finite, it is enough to check regularity of the matrix elements of
	$\L_Y^{(\alpha)}(s)$. From the Kraus decomposition above, if $\{\varphi_k\}_{k=1}^{\dim\H_\sys}$ and $\{\psi_i\}_{i=1}^{\dim\H_\env}$ are fixed orthonormal bases of $\H_\sys$ and $\H_\env$, it is enough to check regularity of the $\C$ valued functions
	$$
	\braket{\varphi_k | K_{i,j} \varphi_l}=\braket{\varphi_k\otimes \psi_j | e^{\alpha Y(s)}U(s)e^{-\alpha Y(s)} \xi(s)^{1/2} \varphi_l\otimes \psi_i}.
	$$
	By \ref{ADRIS} and the explicit dependence in $\alpha$ of the matrices involved, one gets immediately the result.
\end{proof}

\begin{remark}\label{rem:reg}
	In case the regularity assumption in $s$ in \ref{ADRIS} and that of $Y$ are understood in the operator norm sense, and $h_\env(s)$ is such that $\xi(s)$ is $C^2$ in the trace norm sense on $\H_\env$, the map $(s,\alpha)\mapsto \L_Y^{(\alpha)}(s)\in \B(\I_1(\H_\sys))$ is $C^2$ in the norm sense, irrespectively of the dimensions of $\H_\sys$ and $\H_\env$; see  Appendix \ref{sec:proofs_for_FS}.
\end{remark}

We conclude this section with a discussion of the effect of time-reversal on the operator $K_{i,j}$, and therefore on the operator $\L\ealpha_Y$. A relevant assumption will be the following:
\begin{description}
	\labitem{TRI}{TRI} We say that an \ref{ADRIS} satisfies \emph{time-reversal invariance} if for every $s\in[0,1]$ there exist two antiunitary involutions $C_\sys(s) : \H_\sys \to \H_\sys$ and $C_\env(s) : \H_\env \to \H_\env $ such that if $C(s)=C_\sys(s)\otimes C_\env(s)$ one has for all $s\in[0,1]$
	\[  [\hsys,C_\sys(s)] = 0, \qquad [\henv(s),C_\env(s)] = 0,  \qquad [v(s),C(s)] = 0 .\]
\end{description}
This holds for example if each $\hsys, \henv$ and $v$ are real valued matrices in the same basis, and $C_\sys$, $C_\env$ are complex conjugation in the corresponding basis.

In the following result we denote by $K_{i,j}^{(\tau)}(s)$ the operator $K_{i,j}(s)$ associated with the unitary $U(s)$ as defined in \eqref{eq_versionscontinues}. The operator $K_{i,j}^{(-\tau)}(s)$ is therefore associated in the same way with the unitary~$U^*(s)$.

\begin{lemma} \label{lemma_symmetries}
	Assume that an \textup{\ref{ADRIS}} satisfies \textup{\ref{Comm}} and \textup{\ref{TRI}}. Then for all $i,j$ and all $s\in[0,1]$ one has
	\begin{equation} \label{eq_symKij}
		C_\sys(s) K_{i,j}^{(\tau)}(s) C_\sys(s) =K_{i,j}^{(-\tau)}(s).
	\end{equation}
	This implies in particular that for $Y=\beta h_\env$ and all $s\in[0,1]$,
	\begin{equation} \label{eq_symlambda}
		\lambda\ealpha(s)=\lambda^{(-1-\alpha)}(s),
	\end{equation}
	where $\lambda\ealpha(s)$ is the spectral radius of $\L_Y\ealpha(s)$.
	Equivalently the function $\alpha\mapsto \lambda\ealpha(s)$ is symmetric about $\alpha=-1/2$ for all $s\in[0,1]$.
\end{lemma}

\begin{proof}
	Once again we drop the $s$ variable. Using \eqref{eq_formulaKij}, $K_{i,j}^{(\tau)}$ we have
	\begin{align*}
		\langle  \varphi_1, C_\sys K_{i,j}^{(\tau)} C_\sys\varphi_2\rangle
		&= 	\overline{\langle C_\sys  \varphi_1, K_{i,j}^{(\tau)} C_\sys\varphi_2\rangle}\\
		&=	\overline{\langle C_\sys \varphi_1\otimes \psi_j, U \xi^{1/2} C_\sys \varphi_2 \otimes \psi_i\rangle }\\
	\intertext{and by \ref{Comm} we can choose the basis $(\psi_i)_i$ such that $C_\env \psi_i=\psi_i$, $C_\env \psi_j=\psi_j$, so that}
		&= \overline{\langle C (\varphi_1\otimes \psi_j), U \xi^{1/2} C (\varphi_2 \otimes \psi_i)\rangle }\\
		&= {\langle \varphi_1\otimes \psi_j, CUC \xi^{1/2}  \varphi_2 \otimes \psi_i\rangle }\\
		&= {\langle \varphi_1\otimes \psi_j, U^* \xi^{1/2}  \varphi_2 \otimes \psi_i\rangle }
	\end{align*}
	and this proves relation \eqref{eq_symKij}. If we now denote $K_\sys = C_\sys \cdot C_\sys$ the map on $\B(\H_\sys)$, then for $\alpha\in\rr$ this implies $K_\sys \circ \L^{(\alpha,\tau)} \circ K_\sys= \L^{(\alpha;-\tau)}$ for any $Y$ satisfying \ref{Comm}. By Lemma~\ref{lem:adjoint}, for $Y=\beta h_\env$ this implies $K_\sys \circ \L^{(\alpha,\tau)} \circ K_\sys=\L^{(-{\alpha}-1; \tau)}{}^*$ which in turn implies $\lambda^{(\alpha)}=\lambda^{(-\alpha-1)}$.
\end{proof}

\subsection{A general adiabatic result} \label{subsec_adiabaticresult}

The moment generating functions of $\rvY$ and $\rvW$ have been related in Proposition \ref{prop:phiTalpha} to products of $T$ operators $\L_Y\ealpha(s)$  {that differ little from each other, locally, by an amount of order $1/T$.} The general result stated below will allow us to discuss the asymptotic behaviour of $\rvY$ as $T\rightarrow \infty$.

Let $s\mapsto \L(s)$ be a family of CPTP maps satisfying \ref{Irr}, and define $\L_Y\ealpha(s)$ by \eqref{eq_defLk_alpha}. For each $s\in[0,1]$, the map $\L_Y\ealpha(s)$ satisfies~\eqref{eq_KrausLalpha}, and is therefore a deformation of $\L(s)$ in the sense of Appendix~\ref{sec_peripheralspectrum}. From Proposition~\ref{prop_periphspectrumPhialpha}, there exist maps~$\lambda_Y\ealpha(s)$,~$\invalpha_Y(s)$ and~$\rho\ealpha_Y(s)$ from $[0,1]\times \rr \ni(s,\alpha)$ to, respectively, $\rr_+^*$, the set of positive-definite operators, and the set of faithful states of $\H_\sys$; and maps $z(s), u(s)$ from $[0,1]$ to, respectively, $\nn$ and the set of unitary operators, with the following properties:
\begin{itemize}
	\item the identities  $[u(s),\invalpha_Y(s)]=[u(s),\rho\ealpha_Y(s)]=0$, and $u(s)^{z(s)}=\id$ hold;
	\item the peripheral spectrum of $\L_Y\ealpha(s)$ is $\lambda\ealpha_Y(s) S_{z(s)}$, where $S_{z}=\{\theta^m \, | \, \theta=\e^{2\i \pi /z}, {m=0,\ldots,z-1}\}$;
	\item the spectral decomposition $u(s) = \sum_{m=1}^{z(s)} \e^{2\i \pi m/z(s)} p_m(s)$ holds;
	\item the map $\eta\mapsto \tr( \invalpha_Y(s) u(s)^{-m}  \eta) \rho\ealpha_Y(s) u(s)^{m}$ is the spectral projector of $\L_Y\ealpha(s)$ associated with  $\lambda_Y\ealpha(s) \,\e^{2\i \pi m/ z(s)}$;
	\item the unitary $u(s)$ and cardinal $z(s)$ of the peripheral spectrum of $\L_Y\ealpha(s)$ do not depend on~$\alpha$ or $Y$.
\end{itemize}
Note that we have $\lambda^{(0)}_Y(s)=1$, $\mathrm{I}^{(0)}_Y(s)=\id$ and $\rho^{(0)}_Y(s)=\rho\invar(s)$ for all $Y$ and $s$. As mentioned above, the case  $Y=\beta h_\env$ will be particularly relevant to the discussion of the Landauer principle. We therefore drop the indices $Y$, and simply denote by $\lambda\ealpha(s)$, $\invalpha(s)$ and $\rho\ealpha(s)$ the above quantities, in the case where $Y=\beta h_\env$.
We define
\[\tilde\L_Y\ealpha(s)=\big(\lambda\ealpha_Y(s)\big)\inv \L_Y\ealpha(s).\]
The following result will be our main technical tool.

 \begin{proposition}\label{prop_adiabaticlemma}
	Consider an \textup{\ref{ADRIS}} with the family $(\L(s))_{s \in [0,1]}$ satisfying \textup{\ref{Irr}} with $z(s) \equiv z$. Then, there exist continuous functions $\rr \ni \alpha \mapsto \ell'(\alpha)  \in (0,1)$ and $\rr \ni \alpha \mapsto C(\alpha)  \in \rr_+$,
   and a function $\alpha\mapsto T_0(\alpha)\in\nn$ that is bounded on any compact set of $\rr$, such that for all $\alpha \in \rr$, $T\geq T_0(\alpha)$, and~$k\leq T$,
	\begin{align*}
		&\Big\|\tilde \L_Y\ealpha(\tfrac kT) \ldots  \tilde \L_Y\ealpha(\tfrac 1T) \rho\init -z\e^{-\vartheta\ealpha_Y} \sum_{m=0}^{z-1} \tr\big(\invalpha(0) p_m(0) \rho\init\big)  \rho\ealpha_Y(\tfrac kT) p_{m - k}(\tfrac kT) \Big\|\\
		&\hspace{0.5\textwidth} \leq  \frac{C(\alpha) }{T(1-\ell'(\alpha) )}+ C(\alpha)\ell'(\alpha)^k.
	\end{align*}
	where the index of the spectral projector $p_{m-k}(\tfrac kT)$ is interpreted modulo $z$, and
	$$
		\vartheta\ealpha_Y:=\int_0^{k/T} \tr\big(\invalpha_Y(s) \,\frac\partial{\partial s}\rho\ealpha_Y(s)\big)\d s.
	$$
\end{proposition}
\begin{proof}
	By expression \eqref{eq_KrausLalpha}, the map $\rr\times [0,1]\ni (\alpha,s)\mapsto \L_Y\ealpha(s)$ is real analytic in $\alpha$ and $C^2$ in~$s$. We know from Proposition \ref{prop_periphspectrumPhialpha} that the spectral radius $\lambda\ealpha_Y(s)$ of $\L_Y\ealpha(s)$ is a simple eigenvalue for $\L_Y\ealpha(s)$ with eigenvector $\rho\ealpha_Y(s)$, and for $\L_Y\ealpha(s)^*$ with eigenvector $\invalpha(s)$.
	By standard perturbation theory, the maps $(\alpha,s)\mapsto \lambda_Y\ealpha(s), \invalpha_Y(s), \rho\ealpha_Y(s)$ are $C^2$ functions of $s\in[0,1]$ and real analytic functions of $\alpha\in\rr$. The unitary $u(s)$ is an eigenvector for the isolated eigenvalue $\theta=\e^{2\i \pi/z}$ of $\L^*(s)$, and is therefore a $C^2$ function of $s$.
   The peripheral spectrum of $\tilde\L_Y\ealpha(s)$ is the set $S_z=\{\theta^m \,|\, m=0,\ldots,z-1\}$, each peripheral eigenvalue $\theta^m$ is simple, and the associated peripheral projector $\eta\mapsto P\ealpha_m(s)= \tr\big( \invalpha_Y(s) u(s)^{-m}  \eta\big) u(s)^{m} \rho\ealpha_Y(s)$ is therefore a $C^2$ function of~$s$ and a real analytic function of $\alpha$.
   In addition, denoting $Q\ealpha(s)=\id-\sum_{m=1}^z P\ealpha_m(s)$ the quantity $\ell(\alpha)=\sup_{s\in[0,1]}\spr \tilde \L_Y\ealpha(s) Q_\alpha(s)<1$ is a continuous function of $\alpha$.

	From the above discussion, the family $s\mapsto \tilde\L_Y\ealpha(s)$ satisfies \ref{it:cont}--\ref{it:spr-Q} and is therefore admissible, in the sense of Appendix \ref{app:DNUAT}, with simple peripheral eigenvalues. We can therefore apply Corollary \ref{coro:P-Q-decomp}. Denote by $\vartheta\ealpha_{Y,m}$ the integral appearing in the exponential factor:
	\begin{equation} \label{eq_defpsijalpha}
	\vartheta\ealpha_{Y,m} = \int_0^{k/T} \tr \Big( \invalpha_Y(t) u(t)^{-m} \frac{\partial}{\partial t}\big(u^m(t) \rho\ealpha_Y(t)\big)\Big) \d t.
	\end{equation}
	We can prove that $\vartheta\ealpha_{Y,m}$ does not depend on $m$:
	\begin{lemma} \label{lemma_psijconstant}
		We have for $m=0,\ldots,z-1$
		\[\vartheta\ealpha_{Y,m} = \vartheta\ealpha_Y:=\int_0^{k/T} \tr\big(\invalpha_Y(t) \,\frac\partial{\partial t}\rho\ealpha_Y(t)\big)\d t.\]
	\end{lemma}
	\begin{proof}
		The proof follows from a simple expansion of $\frac{\partial}{\partial t}\big(u^m(t) \rho\ealpha_Y(t)\big)$ and commutation properties:
		\begin{align*}
		\vartheta\ealpha_{Y,m}
		&= \vartheta\ealpha_Y + \int_0^{k/T} \tr\big(\invalpha_Y(t) u^{-m}(t) \sum_{k=0}^{m-1} u^k(t) \frac{\partial u}{\partial t}(t) u^{m-k-1}(t)\rho\ealpha_Y(t)\big)\d t\\
		&= \vartheta\ealpha_Y + m \int_0^{k/T} \tr\big(\invalpha_Y(t)\, u^{-1}(t) \frac{\partial u}{\partial t}(t)\rho\ealpha_Y(t)\big)\d t.
		\end{align*}
		However, as $u(t)^z=\id$, plugging $m=0$ or $m=z$ in the right-hand side of expression \eqref{eq_defpsijalpha} gives the same expression $\vartheta\ealpha_Y$, so that necessarily
		\[\int_0^{k/T} \tr\big(\invalpha_Y(t) u^{-1}(t) \frac{\partial u}{\partial t}(t)\rho\ealpha_Y(t)\big)\d t=0\]
		and the conclusion follows.
	\end{proof}
	Lemma \ref{lemma_psijconstant} and Corollary \ref{coro:P-Q-decomp} therefore imply that for any $\ell'(\alpha)\in(\ell(\alpha),1)$, there exist $T_0(\alpha) \in \nn$, and $C(\alpha) > 0$ which is a (fixed) continuous function of
	\[c_P(\alpha)=\sup_{s\in[0,1]}\max_{m=1,\ldots,z}\max\big(\|\phi\ealpha_{Y,m}(s)\|,\|\phi\ealpha_{Y,m}{}'(s)\|,\|\psi\ealpha_{Y,m}(s)\|,\|\psi^{(\alpha)}_{Y,m}{}'(s)\|\big)\]
	with
	\[ \phi\ealpha_{Y,m}(s)=\rho\ealpha_Y(s) u(s)^{m}\qquad \psi\ealpha_{Y,m}(s)=\invalpha_Y(s)u(s)^m,\]
such that for any $T\geq T_0(\alpha)$,
\begin{equation}\label{eq:prevprop311}
\begin{split}
&\Big\|	\tilde \L_Y\ealpha(\tfrac kT) \dotsb  \tilde \L_Y\ealpha(\tfrac 1T) \rho\init
-\e^{-\vartheta\ealpha_Y} \!\!\sum_{m=0}^{z-1} \theta^{mk}\, \sysstate_Y\ealpha(\tfrac kT) u^m(\tfrac kT) \tr\big(\invalpha(0) u^{-m}(0) \sysstate\init\big)\Big\|
   \\ & \hspace{0.5\textwidth} \leq \frac{C(\alpha) }{T\big(1-\ell'(\alpha) \big)}+ C(\alpha)\,\ell'(\alpha){}^k.
\end{split}
\end{equation}
In addition, $T_0(\alpha)$ can be chosen depending on $c_P(\alpha)$ and $\ell'(\alpha)$ alone.

Recall that  we have the spectral decomposition $u = \sum_{m=0}^{z-1} \theta^m p_m$.
Then, with all sums from $0$ to $z-1$ understood modulo~$z$, we have by a discrete Fourier-type computation
\begin{align*}
   &\sum_m \theta^{m k}\tr\big(\invalpha(0) u^{-m}(0) \sysstate\init\big)\sysstate_Y\ealpha(\tfrac kT) u^m(\tfrac kT)\\
   & \qquad\qquad = \sum_{m} \theta^{mk} \sum_n \theta^{-nm} \tr\big(\invalpha(0) p_n(0) \rho\init\big) \sum_\ell \theta^{\ell m}\rho\ealpha_Y(\tfrac kT) p_\ell(\tfrac kT) \\
   & \qquad\qquad = \sum_{n, \ell}   \tr\big(\invalpha(0) p_n(0) \rho\init\big)  \rho\ealpha_Y(\tfrac kT) p_\ell(\tfrac kT) \sum_m \theta^{m(k-n+\ell)} \\
   & \qquad\qquad = \sum_{n, \ell}   \tr\big(\invalpha(0) p_n(0) \rho\init\big)  \rho\ealpha_Y(\tfrac kT) p_\ell(\tfrac kT) \, z\ind_{\ell= n - k} \\
   & \qquad\qquad = z \sum_{n} \tr\big(\invalpha(0) p_n(0) \rho\init\big)  \rho\ealpha_Y(\tfrac kT) p_{n - k}(\tfrac kT).
\end{align*}
This expression along with \eqref{eq:prevprop311} yields the result.
\end{proof}

By taking $\alpha=0$, this result allows adiabatic approximation of the state of $\sys$ under the physical evolution $\L(\frac kT) \dotsm   \L(\frac1T)$ after $k$ steps of an irreducible RIS. This corresponds to a generalization of the results of \cite{HJPR1}, which could only treat the primitive case, i.e. $z=1$.
\begin{corollary} \label{cor:adiab_for_alpha=0}
Consider an \textup{\ref{ADRIS}} with the family $(\L(s))_{s \in [0,1]}$ satisfying \textup{\ref{Irr}} with $z(s) \equiv z$. Then, there exists $\ell' < 1$, $C>0$,  and $T_0>0$ such that for all $T\geq T_0$, and $k\leq T$,
\[
\Big\| \L(\frac kT) \dotsm   \L(\frac1T) \, \rho\init -\rhoadiab(k,T)
		\Big\| \leq  \frac{C }{T(1-\ell')}+ C\ell'^k
\]
	where
\begin{equation}
	\rhoadiab(k,T):= z \sum_{n=0}^{z-1} \tr\big( p_n(0) \rho\init\big) \rho\invar(\tfrac kT)  p_{n - k}(\tfrac kT) \label{eq:def_rho_adiab}
\end{equation}
	is a state, and the index of the spectral projector $p_{n-k}(\tfrac kT)$ is interpreted modulo $z$. Moreover, if $\rho\init$ is faithful, we have the uniform bound
	\[
	\inf_{T>1} \inf_{k\leq T} \inf \sp \rhoadiab(k,T) \geq  z \Big(\min_{1\leq j \leq z}  \tr\big( p_j(0) \rho\init\big) \Big) \inf_{s\in [0,1]} \inf \sp \rho\invar(s) > 0.
	\]
\end{corollary}
\begin{proof}
We apply Proposition~\ref{prop_adiabaticlemma} for $\alpha=0$, and use that $I^{(0)}(s)\equiv \id$, $\rho_Y^{(0)}= \rho\invar$, and $\vartheta\ealpha_Y = 0$ which follows from $\tr \rho^{(0)}_Y(s)\equiv 1$. Next, we check the formula $\tr(\rho\invar(\tfrac kT) p_\ell(\tfrac kT)) =\frac{1}{z}$  for each $\ell=0,\dotsc,z-1$. We drop the argument $\tfrac kT$ in what follows, and write $\L_{k/T}(\cdot) = \sum_i V_i \cdot V_i^*$ the Kraus decomposition.
Recalling that $p_\ell V_i = V_i p_{\ell+1}$ for all $i$ and $\ell$ as discussed in Appendix~\ref{sec_peripheralspectrum},
\begin{align*}
   \tr(\rho\invar p_\ell) &= \tr(\L(\rho\invar) p_\ell) = \sum_{i} \tr(V_i \rho\invar V_i^* p_\ell) \\
      &= \sum_i \tr(V_i \rho\invar p_{\ell+1} V_i^*) = \tr(\L(\rho\invar p_{\ell+1})),
\end{align*}
so $\tr(\rho\invar p_\ell) =\tr( \rho\invar p_{\ell+1})$ using that $\L$ is trace-preserving.
As $\sum_\ell \tr(\rho\invar p_\ell) = \tr \rho\invar =1$, we must have $\tr(\rho\invar p_\ell) = \frac{1}{z}$. Therefore,
\begin{align*}
   \tr (\rhoadiab(k,T)) &= z \sum_{n} \tr\big( p_n(0) \rho\init\big)\tr( \rho\invar(\tfrac kT)  p_{n - k}(\tfrac kT))= \sum_{n} \tr\big( p_n(0) \rho\init\big) \\
      &= \tr \rho\init = 1.
\end{align*}
Moreover, given a normalized vector $\psi \in \H$, we have
\begin{align*}
\braket{\psi,\rhoadiab(k,T) \psi} &= z \sum_{n} \tr\big( p_n(0) \rho\init\big) \braket{\psi,\rho\invar(\tfrac kT)  p_{n - k}(\tfrac kT)\psi}\\
&=z \sum_{n} \tr\big( p_n(0) \rho\init\big) \braket{p_{n - k}(\tfrac kT)\psi,\rho\invar(\tfrac kT)  p_{n - k}(\tfrac kT)\psi}.
\end{align*}
using $[\rho\invar(\tfrac kT),p_{n - k}(\tfrac kT)]=0$. Since $\rho\invar(\tfrac kT) > 0$ and $ \tr\big( p_n(0) \rho\init\big) > 0$, each term in the sum is non-negative, and we have
\begin{align*}
\braket{\psi,\rhoadiab(k,T) \psi} &\geq z  \big(\min_{1\leq j \leq z}  \tr\big( p_j(0) \rho\init\big) \big) \sum_n \braket{\psi,\rho\invar(\tfrac kT)  p_{n - k}(\tfrac kT)\psi} \\
&= z  \big(\min_{1\leq j \leq z}  \tr\big( p_j(0) \rho\init\big) \big) \braket{\psi,\rho\invar(\tfrac kT) \psi}\\
&\geq z \big(\min_{1\leq j \leq z}  \tr\big( p_j(0) \rho\init\big) \big) \inf_{s\in [0,1]} \inf \sp \rho\invar(s).\qedhere
\end{align*}
\end{proof}

\begin{remarks}\label{rem:Af=logrhof_unif_bounded}\,\hfill
	\begin{itemize}
		\item Given an \textup{\ref{ADRIS}} the family $(\L(s))_{s \in [0,1]}$ satisfying \textup{\ref{Irr}}, for faithful $\rho\init$
		the state
		$
		\rho\fin_T = \L_{T} \dotsm \L_{1} \rho\init
		$ is faithful for each $T > 1$ (see the remark after Definition~\ref{def_irreducibility}). Corollary \ref{cor:adiab_for_alpha=0} and Weyl's inequalities (see Section III.2 in \cite{Bhatia}) give the stronger result $\inf_{T>1} \inf\sp \rho\fin_T > 0$.
		In particular, we may make the choice $\Af = - \log \rho\fin_T$ which is bounded uniformly in $T$.
		\item If we assume, in the notation of \cite{HJPR1}, that $\rho\init=(P_0^1+Q_0)\rho\init$ (i.e.\ $\rho\init$ has no components corresponding to the peripheral eigenvalues of $\L(0)$ other than $1$), then one can check that $\rhoadiab(k,T)=\rho\invar(\frac kT)$.
	\end{itemize}
\end{remarks}
\section{Special case:  bounded adiabatic entropy production} \label{sec:XisZero}

In this section we consider $\big(\L(s)\big)_{s \in [0,1]}$ satisfying~\ref{ADRIS} and the primitivity assumption~\ref{Prim}, with $X(s) \equiv 0$, where $X(s)$ is defined in \eqref{eq:def-X}.  We specialize to the case $Y(s)=\beta(s)h_\env(s)$ for all $s\in[0,1]$, and thus drop the subscript $Y$ in the notation.
We recall that when $X(s)\equiv 0$, there exists a family $(k_\sys(s))_{s \in [0,1]}$ of observables on $\H_\sys$ satisfying $[k_\sys(s)+h_{\env}(s),U(s)] \equiv 0$. We claim that for any $\alpha\in\C$,
\begin{equation}\label{eq:X-is-zero-syystate-invar-alpha}
\Exp{-\beta(s)(1+\alpha)k_\sys(s)}
\end{equation}
is an invariant for $\L\ealpha(s)$,
the deformation of $\L(s)$ corresponding to $Y(s)=\beta(s) h_\env(s)$.  This follows from the straightforward computation
\begin{align*}
	&\L\ealpha(s) \big(\Exp{-\beta(s)(1+\alpha)k_\sys(s)}\big) \\
   &\quad= \tr_\env \Big(\Exp{\alpha\beta(s) h_\env(s)} U(s) \Big(\Exp{-\beta(s)(1+\alpha)k_\sys(s)}  \otimes \frac{\Exp{-\beta(s) h_\env(s)}}{Z_\beta(s)}\Big)
         \Exp{-\alpha\beta(s) h_\env(s)} U(s)^*\Big) \\
		&\quad = \Exp{-\beta(s)(1+\alpha)k_\sys(s)}.
\end{align*}
Since $\L\ealpha(s)$ is completely positive and irreducible {for $\alpha\in\rr$} (see Appendix~\ref{sec_peripheralspectrum} for details), and $\Exp{-\beta(s)(1+\alpha)k_\sys(s)}$ is positive-definite, $1$ is necessarily the spectral radius of $\L\ealpha(s)$. We therefore have  $\lambda\ealpha(s)=1$, for all $s$ and $\alpha\in \rr$,  and in addition,
\begin{equation} \label{eq_rhoinvalphaXzero}
\sysstate\ealpha(s) = \frac{\Exp{-(1+\alpha)\beta(s)k_\sys(s)}}{\tr\big(\Exp{-(1+\alpha)\beta(s)k_\sys(s)}\big)},
\end{equation}
with $\rho^{(0)}(s)=\rho\invar(s)$, the invariant state of $\L(s)$.
Similarly, using Lemma~\ref{lem:adjoint}, $\Exp{\bar \alpha \beta(s) k_\sys(s)}$ is an invariant for $\L\ealpha(s)^*$, so that for $\alpha\in\rr$
\begin{equation}\label{eq:X-is-zero-X-alpha}
	\invalpha(s) = \frac{\tr\big(\Exp{-(1+\alpha)\beta(s)k_\sys(s)}\big)}{\tr\big(\Exp{-\beta(s)k_\sys(s)}\big)}\,\Exp{\alpha \beta(s) k_\sys(s)}
\end{equation}
and satisfies the normalization condition $\tr\big(\invalpha(s)\sysstate\ealpha(s)\big)\equiv 1$.

\begin{lemma}\label{lem:product-prim-case}
	Under the assumptions \textup{\ref{ADRIS}} and~\textup{\ref{Prim}}, and with $X(s) \equiv 0$, and $\ell'(\alpha)  \in (0,1)$, $C(\alpha)  \in \rr_+$, $T_0(\alpha)\in\nn$ as in Proposition \ref{prop_adiabaticlemma}, for all $\alpha \in \rr$ and $T\geq T_0(\alpha)$, all $\eta\in\mathcal I_1(\H)$,
	\begin{equation*}
		\big\|\L\ealpha(\tfrac TT) \cdots  \L\ealpha(\tfrac 1T) \eta -
		\tr(\rho\invar(0)^{-\alpha}\eta)\,\rho\invar(1)^{1+\alpha}\big\|
      \leq  \frac{C(\alpha) }{T(1-\ell'(\alpha) )})+ C(\alpha)(\ell'(\alpha))^T,
	\end{equation*}
\end{lemma}
\begin{proof}
	First note that in the primitive case, $z \equiv 1$ so that $\theta = 1$ and only the term with $m = 0$ is present. In addition, as we have proved above, for $X(s)\equiv 0$ one has $\lambda\ealpha(s)\equiv 1$ and therefore $\tilde \L\ealpha(s)=\L\ealpha(s)$. Proposition~\ref{prop_adiabaticlemma} together with expressions \eqref{eq_rhoinvalphaXzero} and \eqref{eq:X-is-zero-X-alpha} then yield
	\begin{align*}
		&\Big\|\L\ealpha(\tfrac TT) \cdots  \L\ealpha(\tfrac 1T) \eta \\
         &\qquad\qquad {} -\e^{-\vartheta\ealpha}  \frac{\tr\big(\Exp{-(1+\alpha)\beta(0)k_\sys(0)}\big)\tr\big(\Exp{\alpha \beta(0) k_\sys(0)}   \eta\big)}{\tr\big(\Exp{-(1+\alpha)\beta(1)k_\sys(1)}\big)\tr\big(\Exp{-\beta(0)k_\sys(0)}\big)} \, \Exp{-(1+\alpha)\beta(1)k_\sys(1)} \Big\|\\
			&\hspace{0.5\textwidth} \leq \frac{C(\alpha) }{T(1-\ell'(\alpha) )})+ C(\alpha)(\ell'(\alpha))^T.
	\end{align*}
	with
		\begin{align*}
		\vartheta\ealpha&=\int_0^1 \tr\Big( \frac{\tr\big(\Exp{-\beta(s)(1+\alpha)k_\sys(s)}\big)}{\tr\big(\Exp{-\beta(s)k_\sys(s)}\big)}\,\Exp{\alpha \beta(s) k_\sys(s)} \,\frac\d{\d s}\frac{\Exp{-\beta(s)(1+\alpha)k_\sys(s)}}{\tr(\Exp{-\beta(s)(1+\alpha)k_\sys(s)})}\Big)\d s\\
		&=  \int_0^1  \frac{\tr(\Exp{-\beta(s)(1+\alpha)k_\sys(s)})}{\tr(\Exp{-\beta(s)k_\sys(s)})}\tr\Big( \Exp{\alpha \beta(s) k_\sys(s)} \,\frac\d{\d s}\frac{\Exp{-\beta(s)(1+\alpha)k_\sys(s)}}{\tr(\Exp{-\beta(s)(1+\alpha)k_\sys(s)})}\Big)\d s\\
		&=  \int_0^1  \big(\tr(\Exp{-\beta(s)k_\sys(s)})\big)^{-1}\tr\Big( \Exp{\alpha \beta(s) k_\sys(s)} \,\frac\d{\d s}\Exp{-\beta(s)(1+\alpha)k_\sys(s)}\Big)\d s\\
		&\qquad - \int_0^1 \big(\tr(\Exp{-\beta(s)(1+\alpha)k_\sys(s)})\big)^{-1}\frac\d{\d s}\tr\big( {\Exp{-\beta(s)(1+\alpha)k_\sys(s)}}\big)\d s.
		\end{align*}
	Thanks to the general formula  $\frac{\d}{\d s} \Exp{A(s)}=\int_0^1 \Exp{x A(s)}\frac{\d}{\d s} A(s)\, \Exp{(1-x) A(s)}\d x$
	and to the cyclicity of the trace, we have
\begin{align*}
	\tr\big( \Exp{\alpha \beta(s) k_\sys(s)} \,\frac\d{\d s}\Exp{-\beta(s)(1+\alpha)k_\sys(s)}\big)&= -(1+\alpha) \tr\Big(\frac{\d }{\d s}\big(\beta(s) k_\sys(s)\big)\e^{-\beta(s)k_\sys(s)}\Big)\\
	&= (1+\alpha) \frac{\d }{\d s}\tr(\e^{-\beta(s)k_\sys(s)}),
\end{align*}
so that
	\begin{equation*}
	\vartheta\ealpha =(1+\alpha)\int_0^1 \frac{\d}{\d s} \log \tr(\e^{-\beta(s)k_\sys(s)}) \,\d s  - \int_0^1 \frac{\d}{\d s} \log \tr(\e^{-(1+\alpha)\beta(s)k_\sys(s)}) \,\d s
	\end{equation*}
and
	\begin{equation*}
	\e^{-\vartheta\ealpha}= \frac{\big(\tr(\e^{-\beta(0) k_\sys(0)})\big)^{1+\alpha} \,\tr(\e^{-(1+\alpha)\beta(1) k_\sys(1)})}{\big(\tr(\e^{-\beta(1) k_\sys(1)})\big)^{1+\alpha} \,\tr(\e^{-(1+\alpha)\beta(0) k_\sys(0)})}.
	\end{equation*}
The rest of the proof is obtained by direct computation.
\end{proof}
Recall that for $\sysstate\init$, $\sysstate\fin$ two faithful states, $\Ai = -\log \sysstate\init$, $\Af = -\log \sysstate\fin$ and $Y(s) = \beta(s) h_\env(s)$, we have the decomposition $\varsigma_{T} = -\Delta s_{\sys,T} + \Delta s_{\env,T}$ (see~\eqref{eq:decomp-entropy-prod-traj}).
\begin{theorem}\label{thm:X-is-0-conv-MGF}
	 Under the assumptions \textup{\ref{ADRIS}} and~\textup{\ref{Prim}}, and with $X(s) \equiv 0$, the distribution of the pair $(\Delta s_{\sys,T},\Delta s_{\env,T})$ converges weakly to a probability measure characterized by its moment generating function
	\[M_{(\Delta s_{\env},\Delta s_{\sys})}(\alpha_1,\alpha_2)=\tr\big(\rho\invar(0)^{-\alpha_1}(\rho\init)^{1-\alpha_2}\big)\,\tr\big(\rho\invar(1)^{1+\alpha_1+\alpha_2}\big).\]
	This probability measure has finite support, contained in the set
	\begin{equation}
	\begin{aligned}\label{eq_supportrandomvariable}
	&\big(\log \sp \rho\invar(1)-\log \sp \rho\invar(0)\big)\times\big(\log \sp \rho\invar(1)-\log \sp \rho\init\big) =\\
	& \big\{(\log r_1-\log r_0,\log r_1'-\log r^{\mathrm{i}})\,|\, r_1,r'_1\in\sp \rho\invar(1), r_0\in \sp\rho\invar(0), r^{\mathrm{i}} \in \sp \rho\init\big\}.
	\end{aligned}
	\end{equation}
In particular, the limiting moment generating function $M_\varsigma(\alpha):=\lim_{T\to\infty}M_{\varsigma_{T}}(\alpha)$ satisfies
\[\log M_\varsigma(\alpha)=S_{-\alpha}(\sysstate\invar(0) | \sysstate\init),\]
where $S_\alpha$ denotes the (unnormalized) R\'enyi relative entropy
$$
S_\alpha(\eta | \zeta) := \log \tr(\eta^\alpha \zeta^{1-\alpha}).
$$
\end{theorem}
\begin{proof}
	By a direct application of Proposition \ref{prop:phiTalpha}, for all $\alpha_1,\alpha_2\in\rr$ we have
\[M_{(\Delta s_{\env,T},\Delta s_{\sys,T})}(\alpha_1,\alpha_2) = \tr \big(\Exp{-\alpha_2 \Af} \L^{(\alpha_1)}(\tfrac{T}{T})  \dotsm \L^{(\alpha_1)}(\tfrac{1}{T}) (\Exp{+\alpha_2 \Ai}\rho\init)\big),\]
so that by Lemma \ref{lem:product-prim-case}
$$
   \big|M_{(\Delta s_{\env,T},\Delta s_{\sys,T})}(\alpha_1,\alpha_2) - \tr(\rho\invar(0)^{-\alpha_1}(\sysstate\init)^{1-\alpha_2})\,\tr\big(\rho\invar(1)^{1+\alpha_1}(\sysstate\fin)^{\alpha_2}\big) \big|
$$
converges to 0 as $T \to\infty$
and again from Lemma \ref{lem:product-prim-case} with $\alpha=0$, $\lim_{T\to\infty}\sysstate\fin=\rho\invar(1)$ and the latter state is faithful. This shows that the moment generating function converges as $T\to\infty$ for all $(\alpha_1,\alpha_2)$, to the desired identity. By the results in Section~30 of~\cite{Bill}, this shows the convergence in distribution of the pair~$(\Delta s_{\env,T},\Delta s_{\sys,T})$.
\end{proof}
\begin{remark}  \label{remark_formulasigmatot}
	Relation \eqref{eq:sprod-avg}, and the fact that the derivative of $S_{-\alpha}(\eta | \zeta)$  is the relative entropy $S(\eta|\zeta)=\tr\big(\eta(\log \eta-\log \zeta)\big)$ imply in particular (again see Section 30 of \cite{Bill}) that under the assumptions \ref{ADRIS} and~\ref{Prim}, and with $X(s) \equiv 0$,
	\[\lim_{T\to\infty} \oldsigmaT = S(\rho\invar(0)|\rho\init).\]
	Theorem \ref{thm:X-is-0-conv-MGF} therefore gives us a refinement of the results of \cite{HJPR1}, where an explicit expression of the limit was missing.
	Remark also that the quantity $S_{-\alpha}(\sysstate\invar(0) | \sysstate\init)$ can be expressed as the cumulant generating function of an explicit distribution related to the relative modular operator for $\sysstate\invar(0)$ and $\sysstate\init$ (see e.g.\ Chapter 2 in \cite{JOPP}).
\end{remark}

\begin{corollary}
	Under the assumptions of Theorem \ref{thm:X-is-0-conv-MGF}, if in addition $\rho\init=\rho\invar(0)$ then the limiting distribution for $(\Delta s_{\env,T},\Delta s_{\sys,T})$ has support on the diagonal  and equivalently the limiting distribution for $\varsigma_{T}$ is a Dirac measure at zero. If we write the spectral decomposition of $\rho\invar(0)$, $\rho\invar(1)$ as
	\[ \rho\invar(0)= \sum_j \tau_j(0)\, \pi_j(0)\qquad \rho\invar(1)= \sum_j \tau_j(1)\, \pi_j(1)\]
	then the limiting distribution for $\Delta s_{\env,T}$ gives the following weight to $s\in\rr$
	\[\sum_{k,j} \tr\big(\rho\invar(0) \pi_j(0)\big)\tr\big(\rho\invar(1) \pi_k(1)\big) \ind_s\big(\log \tau_k(1)-\log \tau_j(0)\big),\]
	where $\ind_s(t)=1$ if $s=t$ and $0$ otherwise.
\end{corollary}
\begin{proof}
	If $\rho\init=\rho\invar(0)$ then with the notation of Theorem \ref{thm:X-is-0-conv-MGF}, one has $\log M_\varsigma(\alpha)\equiv0$, so that the limiting distribution for $\varsigma_{T}$ is a Dirac measure at zero. In addition, the limiting moment generating function for $\Delta s_{\env,T}$ is
	\[\tr\big(\rho\invar(0)^{1-\alpha}\big)\tr\big(\rho\invar(1)^{1+\alpha}\big)\]
and the expression of the corresponding distribution follows by inspection.
\end{proof}
\begin{example} \label{example:RWA}
Let us recall the simplest non-trivial RIS, which is considered in  \cite[Example 6.1]{HJPR1}, for which the system and probes are 2-level systems, with $\H_\sys = \H_\env = \C^2$, along with Hamiltonians $h_\sys := E a^*a$ and $\henv[k] \equiv \henv  := E_0 b^*b$ where $a/a^*$ (resp. $b/b^*$) are the Fermionic annihilation/creation operators for $\sys$ (resp. $\env$), with $E, E_0 > 0$  constants with units of energy. As matrices in the (ground state, excited state) bases $\{|0\rangle, |1\rangle\}$ for $\sys$ and $\env$, we write
\[
a=  b= \begin{pmatrix}
0 & 1 \\ 0 & 0
\end{pmatrix}, \quad a^* = b^* = \begin{pmatrix}
0 & 0 \\ 1 & 0
\end{pmatrix}, \qquad a^*a=b^*b = \begin{pmatrix}
0 & 0 \\ 0 &1
\end{pmatrix}.
\]
We consider a constant potential $v_\text{RW}\in \B(\H_\sys \otimes \H_\env)$,
\[
v_\text{RW} = \frac{\mu_1}{2}(a^* \otimes b + a \otimes b^*)
\]
where $\mu_1 =1$ with units of energy. Given $s\mapsto \beta(s) \in [0,1]$ a $C^2$ curve of inverse probe temperatures, an interaction time $\tau>0$ and coupling constant $\lambda>0$, we let
\[U = \exp\big(-\i \tau (\hsys + \henv + \lambda v_\text{RW})\big).\]
Then $\sp \L(s)$ is independent of $s$, with $1$ as a simple eigenvalue with eigenvector $$\rho\invar(s) = \exp(-\beta^*(s) \hsys)/\tr(\exp(-\beta^*(s) \hsys))$$ for $\beta^*(s) = \frac{E_0}{E}\beta(s)$.

With $\nu := \sqrt{(E-E_0)^2 + \lambda^2}$,  the assumption $\nu\tau\not \in 2\pi \Z$ yields that $\L(s)$ is primitive, and moreover, the fact that $[v_\text{RW}, a^*a +  b^*b] = 0$ yields $X(s)\equiv 0$. Here, we may take $k_\sys \equiv \frac{E_0}{E}\hsys$ independently of $s$, which satisfies $[k_\sys + \henv, U ] \equiv 0$.

We choose an initial system state $\rho\init >0$, and
set $Y(s) := \beta(s) \henv$, $\Ai := \log \rho\init$, $\Af := \log \rho\fin_T$, for $\rho\fin_T := \L(\tfrac TT) \dotsm \L(\tfrac 1T) \rho\init$. By considering the forward and backward processes of Section \ref{subsec:FS_two_time_measurements}, we define the forward  (resp. backward) probability distribution $\bP^F_{T}$ (resp. $\bP^B_{T}$), and the entropy production $\varsigma_{T} = \log \frac{\bP^F_{T}}{\bP^B_{T}}$ on $\Omega = \sp \rho\init \times \sp \rho \fin \times \{0,1\}^T \times \{0,1\}^T$.
Then Theorem \ref{thm:X-is-0-conv-MGF} yields the asymptotic moment generating function of $\varsigma_{T}$:
let $\rho\init=r_0 |v_0\rangle\langle v_0|+r_1 |v_1\rangle\langle v_1|>0$ be the spectral decomposition of the initial state, then
\begin{equation}
\begin{split}
\lim_{T\to\infty}M_{\varsigma_{T}}(\alpha) =
(1+e^{-\beta(0)E_0})^\alpha\Big(&r_0^{1+\alpha}(|\langle 0|v_0\rangle|^2+|\langle 1|v_0\rangle|^2e^{\alpha\beta(0)E_0})\\
      &\,\, + r_1^{1+\alpha}(|\langle 0|v_1\rangle|^2+|\langle 1|v_1\rangle|^2e^{\alpha\beta(0)E_0})\Big).
\end{split}
\end{equation}
\end{example}

\section{General case: large deviations and the central limit theorem} \label{sec:Xnonzero}

In this section we consider a general observable $Y(s)$ with $[Y(s),h_\env(s)]\equiv 0$. We will prove a large deviation principle, and essentially deduce from it a law of large numbers and a central limit theorem.

Our main technical tool will be Proposition \ref{prop_adiabaticlemma}, together with the following result:
\begin{lemma}\label{lem:prod-is-well-behaved}
	Under the assumptions \textup{\ref{ADRIS}} and \textup{\ref{Irr}} with $z(s) \equiv z$, for any faithful initial state~$\sysstate\init > 0$, for any~$\alpha$ in $\rr$ we have
   \begin{equation*}
      0 < \liminf_{T\to\infty}\tr\big(\tilde \L_Y\ealpha(\tfrac TT) \dotsb \tilde \L_Y\ealpha(\tfrac 1T)\rho\init\big)  \limsup_{T\to\infty}  \tr\big(\tilde \L_Y\ealpha(\tfrac TT) \dotsb \tilde \L_Y\ealpha(\tfrac 1T)\rho\init\big)<\infty.
   \end{equation*}
\end{lemma}

\begin{proof}
	By Proposition~\ref{prop_adiabaticlemma},
	\begin{equation}\label{controltrace}
      \begin{split}
          \tr \big(\tilde \L_Y\ealpha(\tfrac TT) \dotsb \tilde \L_Y\ealpha(\tfrac 1T)\rho\init\big)  &= z\e^{-\vartheta\ealpha_Y}   \sum_{n=0}^{z-1} \tr\big(\invalpha_Y(0) p_n(0) \rho\init\big)  \tr\big(\rho\ealpha_Y(1) p_{n - T}(1)\big)
            \\ &\qquad
            + O\left(\frac{C}{T(1-\ell')}\right)+ O(C{\ell '}^T).
      \end{split}
	\end{equation}
	Because $\invalpha_Y(0)$, $\sysstate\init$, are strictly positive matrices we have $\tr\big(\invalpha_Y(0) p_n(0) \rho\init\big)>0$ for all $n$, and because $\sysstate_Y\ealpha(1)$ is a trace one non-negative matrix,$$ \tr\big(\rho\ealpha_Y(1) p_{n - T}(1)\big)>0$$ for some $n$.
		By strict positivity of~$\e^{-\vartheta\ealpha_Y}$, the leading term in (\ref{controltrace}) is therefore strictly positive.
	 This proves the lower bound, whereas the upper bound follows from continuity of the operator-valued maps~$(\alpha, s) \mapsto \invalpha_Y(s), \sysstate_Y\ealpha(s)$.
\end{proof}

\begin{lemma} \label{lemma_cvgCGF}
	Under the assumptions \textup{\ref{ADRIS}} and \textup{\ref{Irr}} with $z(s) \equiv z$, for any faithful initial state~$\sysstate\init > 0$, for any~$\alpha \in \rr$, the moment generating function of the random variable~$\rvY$ with respect to~$\pp^F_{T}$ satisfies
	\begin{equation} \label{eq_defLambdaalpha}
		\lim_{T \to \infty} \frac{1}{T} \log M_{\rvY}(\alpha) = \int_0^1 \log \lambda_Y^{(\alpha)}(s) \d s=: \Lambda_Y(\alpha).
	\end{equation}
\end{lemma}

\begin{proof}\label{lem:log-conv}
	By Lemma~\ref{lem:prod-is-well-behaved},
	$$
		\lim_{T \to \infty} \frac{1}{T} \log \tr\big(\tilde \L_Y\ealpha(\tfrac TT) \dotsb \tilde \L_Y\ealpha(\tfrac 1T)(\rho\init)\big) = 0.
	$$
	But the moment generating function reads
	\begin{equation*}
		M_{\rvY} = \tr\big( \L_Y\ealpha(\tfrac TT) \dotsb \L_Y\ealpha(\tfrac 1T)\rho\init\big) =\Big( \prod_{k=1}^T \lambda_Y\ealpha(\tfrac{k}{T}) \Big) \tr\big(\tilde \L_Y\ealpha(\tfrac TT) \dotsb \tilde \L_Y\ealpha(\tfrac 1T)(\rho\init)\big)
	\end{equation*}
	by definition of $\tilde \L_Y\ealpha(s)$. Hence, the result follows from the Riemann sum convergence
	\begin{align*}
		\lim_{T \to \infty} \frac{1}{T} \log \Big( \prod_{k=1}^T \lambda_Y\ealpha(\tfrac{k}{T}) \Big)
		&= \lim_{T \to \infty} \sum_{k=1}^T(\tfrac{k}{T} - \tfrac{k-1}{T}) \log  \lambda_Y\ealpha(\tfrac{k}{T}) \\
		&= \int_0^1 \log \lambda_Y^{(\alpha)}(s) \d s. \qedhere
	\end{align*}
\end{proof}

\begin{remark} \label{remark_sigmasE}
	Lemma~\ref{lemma_cvgCGF} also holds for e.g.\ the random variable $\Delta a_T(\omega)+\rvY$ in place of~$\rvY$ because Lemma~\ref{lem:prod-is-well-behaved} holds with additional factors of~$\Exp{-\alpha \Ai}$ and~$\Exp{\alpha \Af}$ inside the trace.
   Alternatively, one may remark that~$\rvY(\omega)$ and~$\rvZ(\omega)$ only differ by a uniformly bounded term~$\rvW(\omega)$.
\end{remark}
The regularity of $\Lambda$, and the value of its first and second derivatives at zero, are relevant to the asymptotic behaviour of $\rvY$. We therefore give the following simple lemma.
\begin{lemma} \label{lemma_differentialsLambda}
	Assume \textup{\ref{ADRIS}} and \textup{\ref{Irr}}. Then the function $\Lambda_Y$ is twice continuously differentiable on $\rr$, with
	\begin{equation}
   \begin{split}
	\Lambda'_Y(0) &=\int_0^1 \frac{\partial \lambda_Y\ealpha}{\partial \alpha}(s)|_{\alpha=0}\,\d s, \\
	\Lambda''_Y(0) &=\int_0^1 \Big(\frac{\partial^2 \lambda_Y\ealpha}{\partial \alpha^2}(s)|_{\alpha=0}-\big(\frac{\partial \lambda_Y\ealpha}{\partial \alpha}(s)|_{\alpha=0}\big)^2\Big) \,\d s, \label{eq_LambdaLambdaprime}
   \end{split}
   \end{equation}
	and
	\begin{equation}
	\begin{aligned}\label{eq_lambdaprime1}
	\frac{\partial \lambda_Y\ealpha}{\partial \alpha}(s)|_{\alpha=0} &= \sum_{i,j}(y_j-y_i) \,\tr\big(K_{i,j}(s) \rho\invar(s) K_{i,j}^*(s)\big),
	\\
	\frac{\partial^2 \lambda_Y\ealpha}{\partial \alpha^2}(s)|_{\alpha=0} &= \sum_{i,j}(y_j-y_i)^2 \,\tr\big(K_{i,j}(s) \rho\invar(s) K_{i,j}^*(s)\big) \\
      &\qquad + 2\sum_{i,j}(y_j-y_i) \,\tr\big(K_{i,j}(s) \eta(s) K_{i,j}^*(s)\big)
	\end{aligned}
	\end{equation}
	where $\eta(s)$ is the unique solution with zero trace of
   \begin{align*}
      \big(\id - \L(s)\big) (\eta)
         &= \sum_{i,j}(y_j-y_i) \,K_{i,j}(s) \rho\invar(s) K_{i,j}^*(s) \\
            &\qquad{} - \sum_{i,j}(y_j-y_i) \,\tr\big(K_{i,j}(s) \rho\invar(s) K_{i,j}^*(s)\big) \rho\invar(s).
   \end{align*}
	In particular, for $Y(s)=\beta(s)h_\env(s)$, one has
	\begin{equation}\label{eq_lambdaprime2}
		\frac{\partial \lambda\ealpha}{\partial \alpha}(s)|_{\alpha=0}
		=\beta(s)\,\tr\Big(X(s) \big(\id\otimes h_\env(s)\big)\Big).
	\end{equation}
\end{lemma}
\begin{proof}
	That $\Lambda_Y$ is twice continuously differentiable is clear from the expression \eqref{eq_defLambdaalpha} and the fact that $(\alpha,s)\mapsto \lambda_Y\ealpha(s)$ is $C^2$ in $s$ and analytic in $\alpha$, bounded and bounded away from zero. The expressions \eqref{eq_LambdaLambdaprime} follow from the dominated convergence theorem. The expressions \eqref{eq_lambdaprime1} are obtained by an explicit expansion to second order in $\alpha$ of the relation $\L_Y\ealpha(\rho\ealpha_Y)=\lambda_Y\ealpha\rho\ealpha_Y$, together with the fact that $\tr(\rho\ealpha_Y)~\equiv~1$. Last, remark that $\eta(s)$ is uniquely determined as $1$ is a simple eigenvalue of $\L(s)$, and the associated eigenvectors have nonzero trace.
\end{proof}

The values of the derivatives of $\Lambda_Y$ at $\pm\infty$ will also be relevant. We have:
\begin{lemma} \label{lemma_drvminmaxLambdaY}
	Assume \textup{\ref{ADRIS}} and \textup{\ref{Irr}}. Denote by
	\begin{gather*}
		\nu_{Y,+}(s)=\max\{y_j(s)-y_i(s)\,|\, K_{i,j}(s)\neq 0\}, \\
		\nu_{Y,-}(s)=\min\{y_j(s)-y_i(s)\,|\, K_{i,j}(s)\neq 0\}.
	\end{gather*}
	Then
	\[\lim_{\alpha\to\pm\infty}\Lambda_Y'(\alpha)=\nu_{Y,\pm}:=\int_0^1 \nu_{Y,\pm}(s)\,\d s.\]
\end{lemma}
\begin{proof}
	Because $\Lambda_Y$ is convex and everywhere differentiable, $\lim_{\alpha\to\pm\infty} \Lambda_Y'(\alpha)= \lim_{\alpha\to\pm\infty} \frac1\alpha \Lambda_Y(\alpha)$. Besides, it follows immediately from Proposition \ref{prop_asymptoticssprphialpha} that
	\[\lim_{\alpha\to\pm\infty} \frac1\alpha{\log \lambda_Y\ealpha(s)}=\nu_{Y,\pm}(s). \qedhere \]
\end{proof}

The above technical results allow us to give a large deviation principle for $\Delta s_{\env,T}$ or, equivalently, for $\varsigma_{T}$. In the statement below we denote, for $E$ a subset of $\rr$, by $\inte E$ and $\clos E$ its interior and closure respectively.
\begin{theorem} \label{theo_ldp}
	Assume \textup{\ref{ADRIS}} and \textup{\ref{Irr}}, and that the initial state~$\sysstate\init$ is faithful. Let $\Lambda_Y$ be defined by relation \eqref{eq_defLambdaalpha} and denote by $\Lambda_Y^*$ the Fenchel--Legendre transform of $\Lambda_Y$, i.e.\ for $x\in\rr$ let
	\[\Lambda_Y^*(x)=\sup_{\alpha\in\rr}\big(\alpha x - \Lambda_Y(\alpha)\big).\]
Then $\Lambda_Y^*(x)=+\infty$ for $x\not\in [\nu_{Y,-},\nu_{Y,+}]$, and for any Borel set $E$ of $\rr$ one has
\begin{align*}
   -\inf_{x\in \inte E} \Lambda_Y^*(x) &\leq \liminf_{T\to\infty} \frac1T \log \bP^F_{T}\big(\frac{\rvY}T \in \inte E\big) \\
      &\leq \limsup_{T\to\infty} \frac1T \log \bP^F_{T}\big(\frac{\rvY}T \in \clos E\big)\leq -\inf_{x\in \clos E} \Lambda_Y^*(x).
\end{align*}
The same statement holds with $\rvZ$ in place of $\rvY$. In particular, for $Y=\beta h_\env$, one has
\begin{align*}
   -\inf_{x\in \inte E} \Lambda^*(x) &\leq \liminf_{T\to\infty} \frac1T \log \bP^F_{T}\big(\frac{\varsigma_T}T \in \inte E\big) \\
      &\leq \limsup_{T\to\infty} \frac1T \log \bP^F_{T}\big(\frac{\varsigma_T}T \in \clos E\big)\leq -\inf_{x\in \clos E} \Lambda^*(x)
\end{align*}
and the same statement holds with $\Delta s_{\env,T}$ in place of $\varsigma_T$.
\end{theorem}
\begin{proof}
	From Lemma \ref{lemma_differentialsLambda}, $\Lambda_Y$ is continuously differentiable on $\rr$ and its derivative takes values in $[\nu_{Y,-},\nu_{Y,+}]$. The definition of $\Lambda_Y^*$ implies that it is $+\infty$ outside this interval. The rest of the statement follows from the G\"artner--Ellis theorem (Theorem 2.3.6 in \cite{DZ}) because $\Lambda_Y$, being differentiable everywhere, is (in the language of \cite{DZ}) essentially smooth. That the same statement holds with $\rvZ$ in place of~$\rvY$ follows from Remark~\ref{remark_sigmasE}.
\end{proof}
\begin{remarks}\,\hfill
	\begin{itemize}
	\item In the case of $Y(s)=\beta(s)h_\env(s)$, and under the assumption \eqref{TRI}, the symmetry \eqref{eq_symlambda} in Lemma \ref{lemma_symmetries} is at the very heart of the Gallavotti--Cohen Theorem relating in a parameter free formulation the probabilities to observe opposite signs entropies. Indeed, it implies that $\Lambda$ is symmetric about the $\alpha=-1/2$ axis. A direct computation shows that
	\begin{equation} \label{eq_relationGC}
		\Lambda^*(x)=x+\Lambda^*(-x).
	\end{equation}
	A consequence of Theorem \ref{theo_ldp} together with this equality is that if e.g.\ $\Lambda''(0)\neq 0$,
	\[\lim_{\delta\to 0}\lim_{T\to\infty} \frac1T \log \frac{\pp(+\Delta s_{\env,T}\in[s-\delta,s+\delta])}{\pp(-\Delta s_{\env,T}\in[s-\delta,s+\delta])} = -s.\]
	This is obtained by observing that in the present case, $\Lambda$ is analytic in a neighbourhood of the real axis, and so is $\Lambda^*$ if $\Lambda$ is strictly convex.  See e.g.\ \cite{EHM,JOPP,CJPS1} for more information on the role of symmetries such as \eqref{eq_relationGC}.
	\item	In the case of $Y(s)=\beta(s)h_\env(s)$, and under the assumption that $X(s)\equiv 0$, we have observed in Section \ref{sec:XisZero} that $\lambda\ealpha(s)=1$ for all $\alpha\in\rr$ and $s\in[0,1]$. In that case $\Lambda(\alpha)\equiv 0$ and
	\[\Lambda^*(x) = \left\{
	\begin{array}{cl}
		0 & \mbox{ if } x=0,\\
		+\infty & \mbox{ otherwise.}
	\end{array}
	\right.\]
	The above large deviation statement therefore gives a concentration of~$\tfrac 1T \varsigma_{T}$ or $\tfrac 1T {\Delta s_{\env,T}}$ at zero which is faster than exponential.
	\end{itemize}
\end{remarks}
A first consequence is a result similar to a law of large numbers for~$\rvY$:
\begin{corollary} \label{coro_LLN}
	Under the same assumptions as in Theorem \ref{theo_ldp}, for all $\epsilon>0$ there exists $r_\epsilon>0$ such that for $T$ large enough
	\begin{equation}  \label{eq_expcvg}
	\pp^F_{T}\big(|\frac1T \,\rvY - \Lambda_Y'(0)|>\epsilon\big)\leq \exp -r_\epsilon T.
	\end{equation}
\end{corollary}
\begin{proof}
	See e.g.\ Theorem II.6.3 in \cite{Ellis}.
\end{proof}
\begin{remark} \label{remark_expcvg}
	Such a result is sometimes called exponential convergence. If one could replace~$\pp^F_{T}$ by a $T$-independent probability measure $\pp^F$ in~\eqref{eq_expcvg} (see Remark~\ref{remark_noconsistency}) then the Borel--Cantelli lemma would imply that $\frac1T \rvY$ converges $\pp^F$-almost-surely to $\Lambda_Y'(0)$.
   It implies, however, that $\lim_{T\to\infty} \frac1T \ee(\rvY)=\Lambda_Y'(0)$. In the case $Y=\beta h_\env$, the positivity of $\oldsigmaT$ implies $\Lambda'(0)\geq 0$. Formula~\eqref{eq_lambdaprime2} shows that $\Lambda'(0)=0$ if $X(s)\equiv 0$. The proof of Corollary 6.4 in \cite{HJPR1} shows $\Lambda'(0)>0$ if $X(s)\not\equiv 0$.
\end{remark}

We also obtain a central limit-type result by a slight improvement of the results in Theorem~\ref{theo_ldp}.
\begin{theorem} \label{theo_clt}
	Under the same assumptions as in Theorem \ref{theo_ldp} we have
	\[ \frac1{\sqrt T}\big(\rvY- T \,\Lambda_Y'(0)\big)\underset{T\to\infty}\to \mathcal N\big(0,\Lambda_Y''(0)\big)\]
	in distribution.
\end{theorem}

\begin{proof}
	From Corollary \ref{coro_specperiphalphacomplex}, for fixed $s\in [0,1]$, there exists a complex neighbourhood $N(s)$ of the origin such that for $\alpha\in N(s)$ the peripheral spectrum of $\L_Y\ealpha(s)$ is of the form $\{\lambda_Y\ealpha(s) \theta^m \,|\, m=0,\ldots,z-1\}$, and each $\lambda_Y\ealpha(s) \theta^m$ is a simple eigenvalue.
   Denote by $\phi_m\ealpha(s) \psi_m\ealpha{}^*(s)$ the corresponding spectral projector, parameterized so that $(\alpha,s)\mapsto \phi_m\ealpha(s), \psi_m\ealpha{}(s)$ are $C^2$ functions. By compactness of $[0,1]$, we can find a complex neighbourhood $N$ of the origin containing $\bigcap_{s\in[0,1]} N(s)$ such that the following holds: for $\alpha\in N$, the family $\big(\tilde\L\ealpha(s)\big)_{s\in[0,1]}$ satisfies \textup{\ref{it:cont}--\ref{it:spr-Q}}. Moreover,
	\begin{gather*}
	\sup_{\alpha \in N} \sup_{s\in[0,1]} |1-\lambda\ealpha_Y(s)|<1/2,\\
	\sup_{\alpha \in N} \sup_{s\in[0,1]} \max\big(\|\phi_m\ealpha(s)\|,\|\phi_m\ealpha{}'(s)\|, \|\psi_m\ealpha{}(s)\|,\|\psi_m\ealpha{}'(s)\|\big)<\infty,\\
	\sup_{\alpha \in N} \Big|\int_0^1 \psi_m\ealpha{}^* \big(\phi_m\ealpha{}'(t)\big)\,\d t\Big|<\infty.
	\end{gather*}
	We therefore have
	\[\sup_{\alpha\in N}\sup_{T\in \nn} \Big|\frac1T \sum_{k=1}^T\log \lambda_Y\ealpha\big(\frac kT\big)\Big|<\infty\]
	and by Corollary \ref{coro:P-Q-decomp},
	\[\sup_{\alpha\in N}\sup_{T\in \nn}  \Big|\frac{1}{T} \log \tr\big(\tilde \L_Y\ealpha(\tfrac TT) \dotsb \tilde \L_Y\ealpha(\tfrac 1T)(\rho\init)\big)\Big| <\infty.\]
	This implies
	\[\sup_{\alpha\in N}\sup_{T\in \nn}  \Big|\frac1T\log M_{\rvY}(\alpha) \Big|<\infty.\]
	In addition, from Lemma \ref{lemma_cvgCGF}, $\frac1T \log M_{\rvY}(\alpha)$ converges as $T\to\infty$ for $\alpha\in N\cap \rr$. By Bryc's theorem~\cite{Bryc} (see also Appendix A.4 in \cite{JOPP}) as $T\to\infty$,  $\frac1{\sqrt T}\big(\rvY- T \,\Lambda_Y'(0)\big)$ converges in distribution to $\mathcal N\big(0,\Lambda_Y''(0)\big)$.
\end{proof}

\begin{remark} \label{remark_dvgoldsigmaT}
	In the case $Y=\beta h_\env$, Remark \ref{remark_expcvg} and Theorem \ref{theo_clt} show that, if $\Lambda''(0)\neq 0$ (which is generically expected) then $\oldsigmaT\to\infty$ as $T\to\infty$.
\end{remark}

\begin{example} \label{example:FD}
Let us consider the setup of Example~\ref{example:RWA} using the full-dipole interaction potential $v_\text{FD}\in \B(\H_\sys \otimes \H_\env)$,
\[
v_\text{FD} = \frac{\mu_1}{2} (a+a^*) \otimes (b +  b^*),
\]
instead of $v_\text{RW}$.
This example was considered in \cite[Section 7.1]{HJPR1}, where it was shown that \ref{Prim} is satisfied, and that $\sigma_T \to \infty$ with a finite and nonzero rate $\lim_{T\to\infty} \frac{1}{T}\sigma_T$ for generic choices of parameters $\{E,E_0,\tau\}$.

We take $Y(s) =\beta(s) \henv$ as in Example~\ref{example:RWA}. Introducing the shorthand $\eta := \sqrt{(E_0+E)^2 + \lambda^2}$, we compute a matrix expression for $\L\ealpha_s$ by  identifing $\I_1(\H_\sys) \cong \Mat_{2\times 2}(\C) \cong \C^4$ via $\begin{psmallmatrix}
\eta_{11} & \eta_{12}\\
\eta_{21} & \eta_{22}
\end{psmallmatrix} \mapsto \begin{psmallmatrix}
\eta_{11} \\ \eta_{12} \\ \eta_{21} \\ \eta_{22}
\end{psmallmatrix}$.
Working in the (ground state, excited state) basis for~$\sys$, we obtain
$$
   \L\ealpha(s)
      =
      \begin{pmatrix}
         a & 0 & 0 & d \\
         0 & b & c & 0 \\
         0 & c & e & 0 \\
         f & 0 & 0 & g
      \end{pmatrix},
$$
where
\begin{align*}
   a &= \frac{\left(2 (E_0+E)^2+\lambda ^2+\lambda ^2 \cos (\eta
     \tau )\right)}{2 \left(1+\Exp{E_0 \beta(s) }\right)\eta ^2\Exp{-E_0 \beta(s) } }+\frac{2 (E_0-E)^2+\lambda ^2+\lambda ^2 \cos (\nu  \tau
     )}{2 \left(1+\Exp{E_0 \beta(s) }\right)\nu ^2}, \\
   d &= \lambda ^2 \left(-\frac{2
   \Exp{-E_0 \alpha  \beta(s) } (\cos (\eta  \tau )-1)}{4 \left(1+\Exp{E_0 \beta(s) }\right)\eta ^2}-\frac{2 \Exp{E_0 (\alpha +1)
   \beta(s) } (\cos (\nu  \tau )-1)}{4 \left(1+\Exp{E_0 \beta(s) }\right)\nu ^2}\right), \\
   c &= \frac{\lambda ^2 \cosh \left(E_0 \beta(s)(\frac{1}{2}+\alpha) \right)
      \text{sech}\left(\frac{E_0 \beta(s) }{2}\right) \sin \left(\frac{\eta  \tau }{2}\right) \sin
      \left(\frac{\nu  \tau }{2}\right)}{\sqrt{E_0^4+2 \left(\lambda ^2-E^2\right)
      E_0^2+\left(E^2+\lambda ^2\right)^2}}, \\
   b &= \frac{\left(\mathrm{i} \eta  \cos \left(\frac{\eta  \tau }{2}\right)+(E_0+E) \sin
     \left(\frac{\eta  \tau }{2}\right)\right) \left((E_0-E) \sin \left(\frac{\nu  \tau
     }{2}\right)-\mathrm{i} \nu  \cos \left(\frac{\nu  \tau }{2}\right)\right)}{\sqrt{E_0^4+2
     \left(\lambda ^2-E^2\right) E_0^2+\left(E^2+\lambda ^2\right)^2}}, \\
   e &= \frac{
      \left(-\Exp{\mathrm{i} \nu  \tau } E_0+E_0-E+\nu +\Exp{\mathrm{i} \nu  \tau } (E+\nu
      )\right) \left(\eta  \cos \left(\frac{\eta  \tau }{2}\right)+\mathrm{i} (E_0+E) \sin
      \left(\frac{\eta  \tau }{2}\right)\right)}{2 \eta  \nu \Exp{\frac{1}{2} \mathrm{i} \nu  \tau }}, \\
   f &= \frac{\Exp{-E_0 \alpha  \beta(s) } \lambda ^2 }{4
     \left(1+\Exp{E_0 \beta(s) }\right)}\left(\frac{2-2 \cos (\nu  \tau )}{\nu ^2}-\frac{2
     \Exp{E_0 (2 \alpha +1) \beta(s) } (\cos (\eta  \tau )-1)}{\eta ^2}\right), \\
   g &= \frac{ \left(2
   (E_0+E)^2+\lambda ^2+\lambda ^2 \cos (\eta  \tau )\right)}{2 \left(1+\Exp{-E_0
   \beta(s) }\right)\eta ^2 \Exp{E_0 \beta(s) }}+\frac{2
   (E_0-E)^2+\lambda ^2+\lambda ^2 \cos (\nu  \tau )}{2 \left(1+\Exp{-E_0
   \beta(s) }\right)\nu ^2},
\end{align*}
which depend on $s$ through $\beta(s)$. The computation was performed with Mathematica, using \cite{Mathematica_code}. We make a particular choice of parameters, $\lambda = 2$, $\tau=0.5$, $E_0 = 0.8$, $E=0.9$, and two choices of $[0,1]\ni s \mapsto \beta(s)$:
\begin{equation}\label{eq:beta1}
\beta_1(s) = \frac{2(3+4\tanh(2s))}{3 + 2 \log(\cosh(2))}
\end{equation}
and
\begin{equation}\label{eq:beta2}
\beta_2(s) =a_1 \tanh(2 s) -a_2 \tanh\big(\frac{s}{2}\big) - a_3 s^3 + a_4 s^2 -a_5 s + a_6
\end{equation}
for $a_1 = 35.483$, $a_2=141.929$, $a_3=42.945$, $a_4 = 93.5$, $a_5=17.808$, $a_6 = 1.061$. We have $\beta_1(0) = \beta_2(0) = 1.06$, and $\beta_1(1) = \beta_2(1) = 2.43$, as well as $\int_0^1 \beta_1(s) \d s = \int_0^1 \beta_2(s) \d s = 2$. These are plotted in Figure~\ref{fig:FDbetas}.

We compute numerically the function $\Lambda(\alpha)$ for each choice of $s\mapsto \beta(s)$, as shown in Figure~\ref{fig:FD_Lambda}. Figures~\ref{fig:CLT_FD1} and \ref{fig:CLT_FD2} shows the convergence described by Theorem \ref{theo_clt} by simulating 2,000 instances of this repeated interaction system at four values of $T$.

\begin{figure}[ht]

 \centering

  \includegraphics{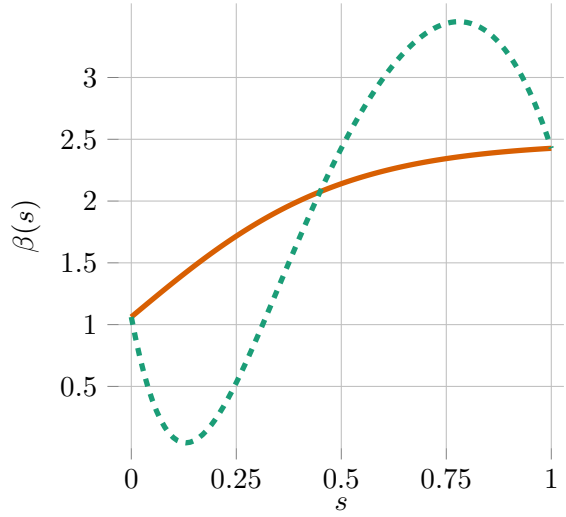}

\caption{Two choices of curves $s\mapsto \beta(s)$. In the solid orange line, $\beta(s)= \beta_1(s)$, given by \eqref{eq:beta1}, and in dashed green line, $\beta(s) = \beta_2(s)$, given by \eqref{eq:beta2}. \label{fig:FDbetas}}
\end{figure}

\vspace{2em}

\begin{figure}[ht]
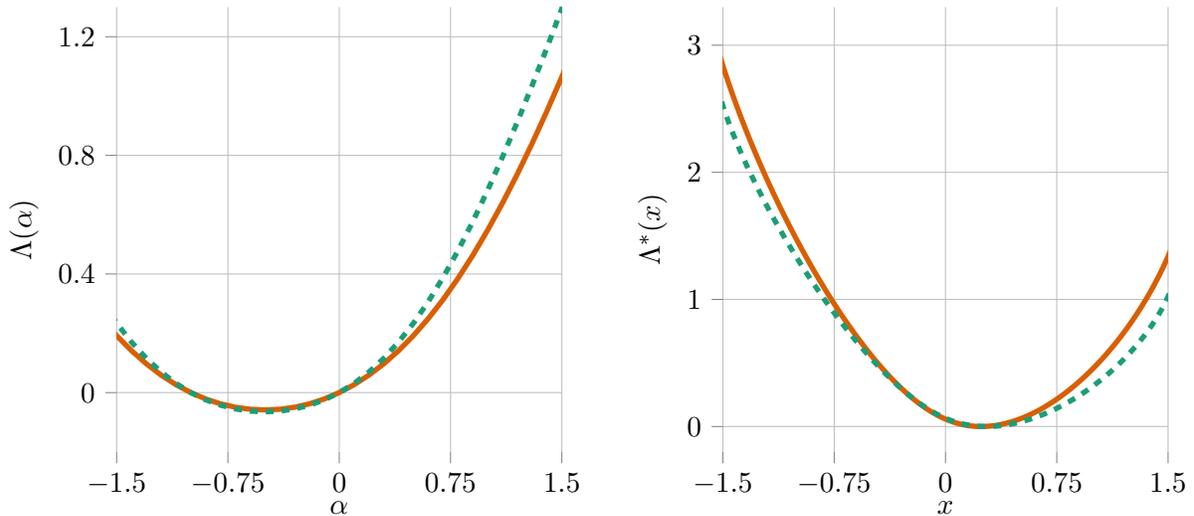

\centering
  \includegraphics{Figures/Lambda_plots.tikz}
\quad
  \includegraphics{Figures/Lambda_star_plots.tikz}

\caption{
Left: The function $\Lambda(\alpha)$ for $Y = \beta \henv$ in the system of Example~\ref{example:FD}, with $\lambda = 2$, $\tau=0.5$, $E_0 = 0.8$,  $E=0.9$, plotted for each choice of $\beta(s)$.
Right: The rate function $\Lambda^*(\alpha)$, for the same setup.
In each plot, the solid orange line corresponds to the choice $\beta(s) = \beta_1(s)$, defined in \eqref{eq:beta1}, and the dashed green line corresponds to $\beta(s) = \beta_2(s)$, defined in \eqref{eq:beta2}.  \label{fig:FD_Lambda}}
\end{figure}

\begin{figure}[ht]
\centering
   \includegraphics[width=.8\textwidth]{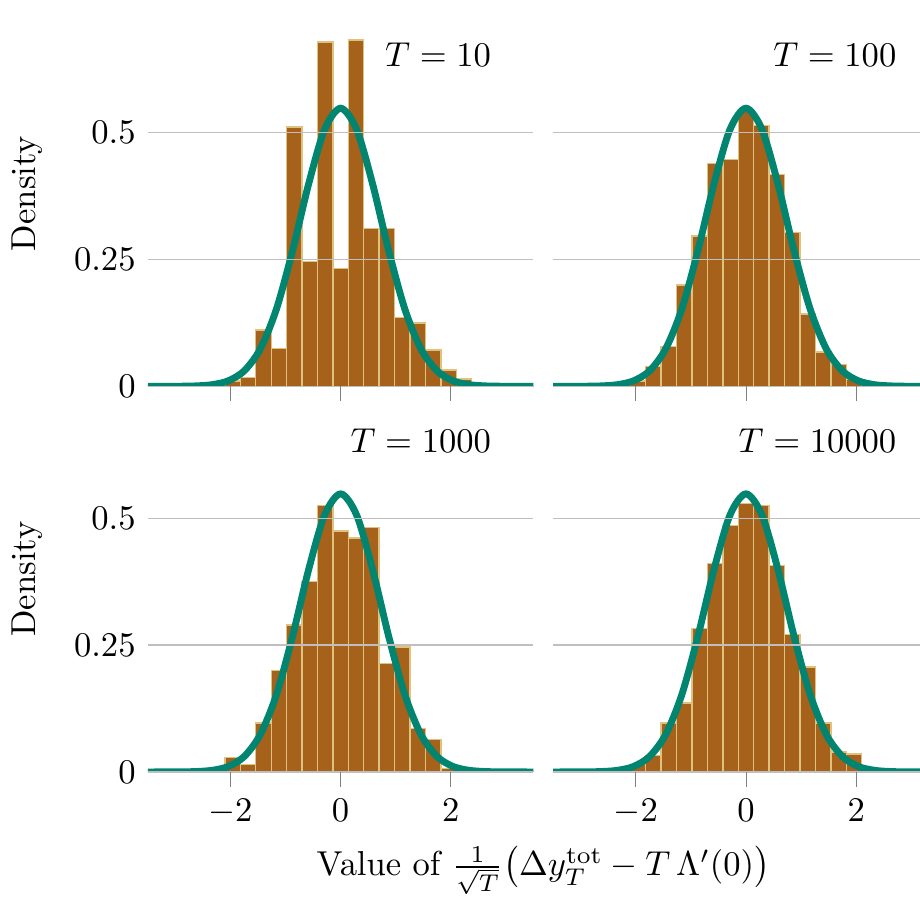}
   \caption{Convergence of $ \frac1{\sqrt T}\big(\rvY- T \,\Lambda'(0)\big)$ to a normal distribution, where $\Lambda'(0) \approx 0.240$, with $\beta(s)=\beta_1(s)$ given by~\eqref{eq:beta1}. Each plot was generated by simulating the two-time measurement protocol in 2,000 instances of the repeated interaction system described in Example~\ref{example:FD}. The value of $ \frac1{\sqrt T}\big(\rvY- T \,\Lambda'(0)\big)$ was calculated for each instance and plotted in a histogram in orange, with bar heights normalized to yield total mass 1. In green, the probability density function of $\mathcal{N}(0,\Lambda''(0))$ is plotted, where $\Lambda''(0) \approx 0.530$. As~$T$ increases, one sees qualitatively the convergence of $ \frac1{\sqrt T}\big(\rvY- T \,\Lambda'(0)\big)$ to the normal distribution, as guaranteed by Theorem~\ref{theo_clt}.  \label{fig:CLT_FD1}}
\end{figure}

\begin{figure}[ht]
	\centering
   \includegraphics[width=.8\textwidth]{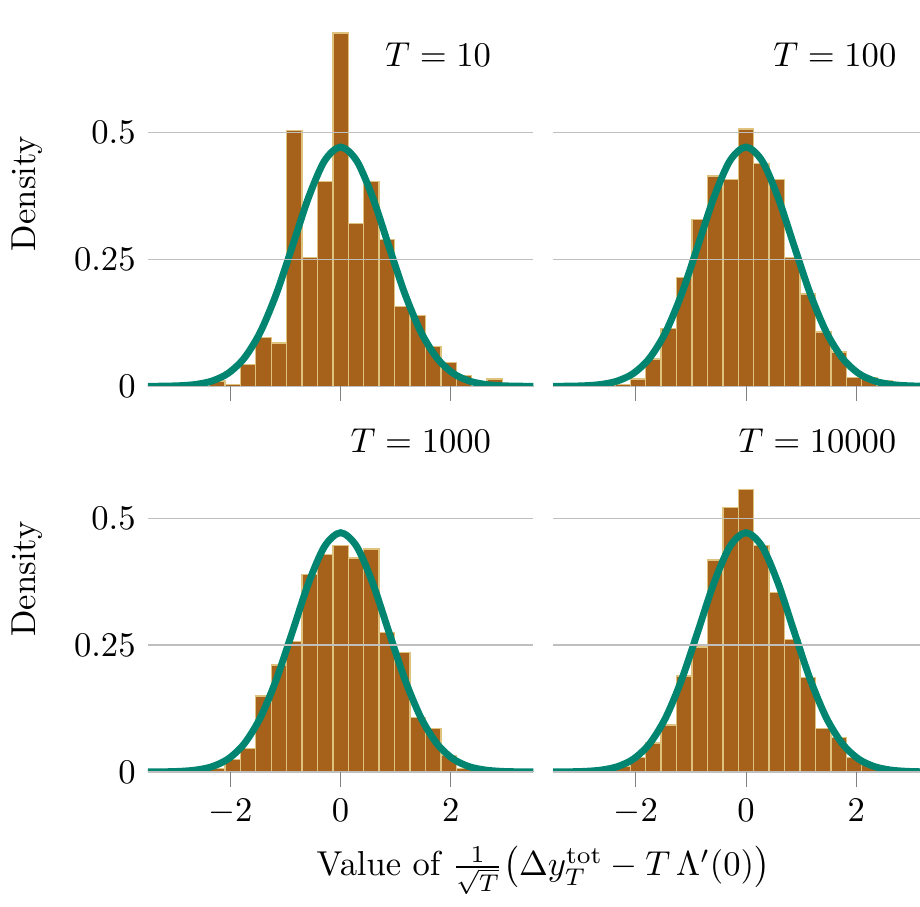}
   \caption{The same setup as Figure~\ref{fig:CLT_FD1}, with $\beta(s)=\beta_2(s)$ given by~\eqref{eq:beta2}. Here, $\Lambda'(0) \approx 0.275$, and $\Lambda''(0) \approx 0.716$. \label{fig:CLT_FD2}}
\end{figure}

\end{example}

\FloatBarrier

\appendix

\section{Peripheral spectrum of CPTP maps and their deformations} \label{sec_peripheralspectrum}

In this section we discuss a full study of the peripheral spectrum, and associated spectral projectors, of CPTP maps and their deformations. This will in particular apply to the deformed reduced dynamical operators $\L_Y^{(\alpha)}$.

We start by collecting various results from the seminal paper \cite{EHK}. Let us therefore consider a finite-dimensional Hilbert space $\H$, and $\Phi$ a completely positive, not necessarily trace-preserving map $\Phi$ on $\mathcal I_1(\H)$. Since $\H$ is finite-dimensional, we can identify $\mathcal I_1(\H)$ and $\B(\H)$, so that all definitions below apply to either $\Phi$ or $\Phi^*$. Any completely positive map on $\B(\H)$ (with finite-dimensional $\H$) admits a \emph{Kraus decomposition}, i.e.\ there exist maps $V_i\in \B(\H)$ for $i$ in a finite set $I$, such that $\Phi(\rho)=\sum_{i\in I} V_i \rho V_i^*$ for all $\rho$.

\begin{definition} \label{def_irreducibility}
If the completely positive map $\Phi$ satisfies either of the following equivalent properties
\begin{itemize}
	\item the only self-adjoint projectors~$P$ on~$\H$ satisfying $\Phi \big(P \I_1(\H) P \big) \subseteq P \I_1(\H) P$ are $\id$ and $0$,
	\item the only subspaces $E$ of $\H$ such that $V_i E \subset E$  for all $i\in I$ are $\{0\}$ and $\H$,
\end{itemize}
we say that $\Phi$ is \emph{irreducible}. If for any nonzero self-adjoint projector~$P$ on~$\H$, there exists $n$ such that the map $\Phi^n(P)$ is positive-definite, we say that $\Phi$ is \emph{primitive}.
\end{definition}
Clearly, if $\Phi$ is primitive then it is irreducible.
In addition, it is immediate to see from the above equivalences that $\Phi$ is irreducible (resp.\ primitive) if and only if $\Phi^*$ is irreducible (resp.\ primitive). Remark also that an irreducible completely positive map $\Phi$ will map a faithful state $\rho$ to a positive-definite operator, as otherwise the support projector $P$ of $\Phi(\rho)$ will satisfy $P\leq c\rho$, and therefore $\Phi(P)\leq c\Phi(\rho)\leq c' P$, for some $c,c'>0$, and therefore contradict the definition of irreducibility above.
\smallskip

It is shown in \cite{EHK} that,  if $\Phi$ is irreducible, then its spectral radius~$\lambda$ is a simple eigenvalue and the associated spectral subspace is generated by a positive-definite operator. An immediate consequence is that any positive-definite eigenvector of $\Phi$ must be an eigenvector for~$\lambda$.
\smallskip

If $\Phi$ is CPTP then necessarily $\lambda=1$. It is also shown in \cite{EHK} that, if $\Phi$ is {completely positive}, irreducible, and trace-preserving, then
\begin{itemize}
	\item the peripheral spectrum of $\Phi$ is a subgroup $S_z=\{\theta^m \,|\, m=0,\ldots,z-1\}$ of the unit circle, where $\theta=\e^{2\i \pi/z}$, and each $\theta^m$  is a simple eigenvalue,
	\item there exist a faithful state $\rho\invar$, and a unitary operator $u$ (called a Perron--Frobenius unitary of $\Phi$) satisfying $[\rho\invar,u]=0$, $u^z=\id$ and $u^k\neq\id$ for $k=0,\ldots, z-1$, such that
	\begin{equation} \label{eq_relationspfunitary}
	\begin{array}{c@{\,=\,}cl}
	\Phi(\rho u) & \theta \,\Phi(\rho) u & \forall \rho  \in \mathcal I_1(\H)\\
	\Phi^*(u X) & \overline \theta u\, \Phi^*(X) & \forall X \in \mathcal B(\H).
	\end{array}
	\end{equation}
\end{itemize}
A consequence of the above is that the (unique up to a multiplicative constant) eigenvector of $\Phi$ (resp.\ $\Phi{}^*$) associated with the eigenvalue $\theta^m$ is $\rho\invar u^m$ (resp. $u^{-m}$), and the spectral projector of $\Phi$ associated with $\theta^m$ is $\eta\mapsto \tr( u^{-m}  \eta) \rho\invar u^{m}$.

Remark that relations \eqref{eq_relationspfunitary} are equivalent to $V_i u=\theta u V_i$ (see \cite{FagPel}). Last, a CPTP map is primitive if and only if it is irreducible with $z=1$, or equivalently if and only if $\Phi^n$ is irreducible for any $n\in\nn$. Conversely, a CPTP map that admits a faithful state as a unique (up to a multiplicative constant) invariant is irreducible. If, in addition, $1$ is the only eigenvalue of modulus one, then $\Phi$ is primitive. In particular, our description of assumptions \ref{Irr} and \ref{Prim} are consistent with the above definitions.

In addition, the spectral decomposition of $u$ is of the form $u=\sum_{m=0}^{z-1} \theta^m \proj_m$, where the projectors $\proj_m$ satisfy $\proj_m V_i = V_i \proj_{m+1}$ for all $i$ and $m$ (here and below, $m+1$ means $m+1 \mod z$ whenever it appears as the index of a projector $\proj$ and we adopt the same convention for $m-1$). Each subspace $\mathcal B(\ran \proj_m)$ of $\mathcal B(\H)$ is therefore invariant by $\Phi^*{}^z$ and the restriction of {$\Phi^*{}^z$} to that subspace is primitive.

\smallskip
We now define deformations of CPTP maps, or, rather, of their Kraus decompositions. For this, fix a finite set $I$. We call a family $V=(V_i)_{i\in I}$ of operators on $\B(\H)$ an \emph{irreducible Kraus family} (indexed by $I$) if $\sum_{i\in I} V_i^* V_i=\id$, and the only subspaces $E$ of $\H$ such that $V_i E\subset E$ for all $i\in I$ are $\{0\}$ and $\H$. We fix a set $I$ and denote by $\Kraus_I$ the set of irreducible Kraus maps indexed by $I$. From the above discussions, any irreducible Kraus family $(V_i)_{i\in I}$ defines an irreducible CPTP map $\Phi$ by $\Phi(\rho)=\sum V_i \rho V_i^*$.
\begin{remark}
	Conversely, any irreducible CPTP admits an irreducible Kraus decomposition indexed by $I=\{1,\ldots,(\dim\H)^2\}$ (possibly with $V_i=0$ for some $i$). However, in applications of the present results in Section \ref{subsec_adiabaticresult}, where $\H=\H_\sys$, our model yields a Kraus family indexed by pairs $(i,j)\in \sp Y\times \sp Y$ where $Y$ is an operator acting on a Hilbert space $\H_\env$ unrelated to $\H_\sys$. We therefore need to consider Kraus families indexed by an arbitrary set $I$.
\end{remark}
Now fix $v=(v_i)_{i\in I}$ a family of strictly positive real numbers. For $(V_i)_{i\in I}$ an irreducible Kraus family and $\alpha\in \rr$ we define a map $\Phi^{(\alpha)}$ on $\I_1(\H)$ by
\[\Phi^{(\alpha)}(\rho)  =\sum_i v_i^{\alpha}\, V_i \rho V_i^*.\]
This map $\Phi\ealpha$  is a completely positive map, and since $\Phi^{(0)}=\Phi$, it can be viewed as a deformation of $\Phi$. We will prove the following result about the peripheral spectrum of $\Phi^{(\alpha)}$.
\begin{proposition} \label{prop_periphspectrumPhialpha}
	Let $(V_i)_{i\in I}$ an irreducible Kraus family, $v=(v_i)_{i\in I}$ a family of strictly positive real numbers, and define $\Phi$, $\Phi\ealpha$ as above. Let $u$ be a Perron--Frobenius unitary for $\Phi$, and denote by $p_m, m=0,\ldots,z-1$ its spectral projectors. There exist three smooth maps $\alpha\mapsto \lambda\ealpha, \invalpha,\rho\ealpha$ from $\R$ to, respectively, $\R_+^*$, the set of positive-definite operators, and the set of faithful states, such that for all $\alpha$ in $\rr$,
	\begin{itemize}
		\item the peripheral spectrum of $\Phi^{(\alpha)}$ is $\lambda\ealpha S_z = \{ \lambda\ealpha \theta^k \,|\, k=0,\dotsc,z-1 \},$
		\item one has the commutation relations $[\invalpha,u]=0$, and $[\rho\ealpha,u]=0$,
		\item one has $\tr(\rho\ealpha \,\invalpha)=1$ for all $\alpha\in \rr$,
		\item the (unique up to a multiplicative constant) eigenvector of $\Phi\ealpha$ (resp.\ $\Phi\ealpha{}^*$) associated with the eigenvalue $\lambda\ealpha \theta^m$ is $\rho\ealpha u^m$ (resp. $\invalpha u^{-m}$), and the spectral projector of $\Phi^{(\alpha)}$ associated with $\lambda\ealpha \theta^m$ is $$\eta\mapsto \tr( \invalpha u^{-m}  \eta) \rho\ealpha u^{m}.$$
	\end{itemize}
\end{proposition}
\begin{remark}
	For $\alpha=0$ we have $\lambda\ealpha=1$, $\invalpha=\id$ and $\rho\ealpha=\rho\invar$. Note also that $\lambda\ealpha$, $\invalpha$, $\rho\ealpha$ depend on the choice of $v=(v_i)_{i\in I}$.
\end{remark}

\begin{proof}
	By the criterion on irreducibility cited above, the map $\Phi\ealpha{}^*$ is completely positive and irreducible. Therefore, its spectral radius $\lambda\ealpha>0$ is a simple eigenvalue, which is locally isolated, with positive-definite eigenvector {$\invalpha$}. {Moreover, recall that  $\invalpha$ is the unique positive-definite eigenvector (up to a positive constant) associated to a positive eigenvalue.} By standard perturbation theory we can parameterize the map $\alpha\mapsto \invalpha$ to be analytic in a neighbourhood of the origin. This $\invalpha$ is defined up to a multiplicative constant, which we will specify later on. {We define a map $\widehat \Phi\ealpha$ and its adjoint $\widehat \Phi \ealpha{}^*$} by
	\begin{equation}
	\begin{aligned}\label{eqdefT}
	\widehat \Phi \ealpha{}(\eta)&=(\lambda\ealpha{})\inv \, ({\invalpha})^{1/2}  \Phi\ealpha{}\Big((\invalpha{})^{-1/2} \eta \,(\invalpha{})^{-1/2}\Big) (\invalpha{})^{1/2}\\
	\widehat \Phi \ealpha{}^*(X)&=(\lambda\ealpha{})\inv \, ({\invalpha})^{-1/2}  \Phi\ealpha{}^*\,\Big((\invalpha{})^{1/2} X (\invalpha{})^{1/2}\Big) (\invalpha{})^{-1/2}.
	\end{aligned}
	\end{equation}
	Note that $\widehat \Phi \ealpha{}$ writes $\widehat \Phi \ealpha{}(\rho)=\sum_{i\in I}\widehat V_i(\alpha)\rho \widehat V_i(\alpha)^*$, with
	\begin{align}\label{TonV}
	\widehat V_i(\alpha)={(v_i^\alpha/\lambda\ealpha)}^{1/2}(\invalpha{})^{+1/2}V_i\,(\invalpha{})^{-1/2}.
	\end{align}
	The application $\widehat \Phi \ealpha{}^*$ is completely positive, irreducible since $\invalpha$ entering in the definition of its Kraus operators is invertible, and satisfies $\widehat \Phi \ealpha{}^*(\id)=\id$.  Hence  the map $\widehat\Phi\ealpha$ is irreducible, completely positive and trace-preserving, so that $\widehat V=(\widehat V_i(\alpha))_{i\in I}\in \Kraus_I$. We can therefore define a map $T_v\ealpha$ on $\Kraus_I$ by $T\ealpha_v: V\mapsto \widehat V(\alpha)$. Note that $\widehat V_i(0)=V_i$, so that $\widehat V$ is a deformation of $V$. We have the following easy result:

	\begin{lemma} \label{lemma_Talphainverse}
		With the above notation (and fixed $v=(v_i)_{i\in I}$), for any $\alpha\in \rr$ the map $T_v\ealpha$ is invertible with inverse $T_v^{(-\alpha)}$.
	\end{lemma}
	\begin{proof}[Proof of Lemma \ref{lemma_Talphainverse}]
		For $\rho\in\I_1(\H)$ consider
      \begin{equation*}
         \rho \mapsto \sum_i v_i^{-\alpha}\, \widehat V_i(\alpha)\,\rho\,\widehat V_i(\alpha)^* = (\lambda\ealpha)\inv\sum_i (\invalpha{})^{+1/2} V_i \,(\invalpha{})^{-1/2} \,\rho\, (\invalpha{})^{-1/2} V_i^* (\invalpha{})^{+1/2}.
      \end{equation*}
		The dual of this map is
		$$X\mapsto (\lambda\ealpha)\inv \sum_i (\invalpha{})^{-1/2}V_i^* (\invalpha{})^{+1/2} \, X \, (\invalpha{})^{+1/2} V_i \,(\invalpha{})^{-1/2},$$
		which admits $(\invalpha{})\inv$ as an eigenvector for $(\lambda\ealpha{})\inv$. Since $(\invalpha{})\inv$ is positive-definite, $(\lambda\ealpha{})\inv$ is the spectral radius of this map, with associated eigenvector $(\invalpha{})\inv$. Applying the above definition of~$T_v^{(-\alpha)}$ therefore shows that $T_v^{(-\alpha)}(\widehat V)$ consists of maps $(v_i^{-\alpha}/\lambda\ealpha{}\inv)^{1/2} \, (\invalpha{})^{-1/2} \widehat V_i(\alpha) \,(\invalpha{})^{+1/2}=V_i$.
	\end{proof}

	As mentioned above, $\widehat\Phi\ealpha$ is an irreducible CPTP map. From the results recalled above, its peripheral spectrum is of the form $S_{z\ealpha}$, all peripheral eigenvalues are simple, and an eigenvector associated with $\theta\ealpha=\e^{2\i j \pi/z\ealpha}$ is of the form $\widehat\rho\ealpha (u\ealpha)^m$ with $\widehat\rho\ealpha\in \D(\H)$ positive-definite and $u\ealpha$ unitary.
   Remark already that since $\widehat\rho\ealpha$ is associated with the simple, isolated eigenvalue~$1$, we can parameterize $\alpha\mapsto\widehat\rho\ealpha$ to be analytic in a neighbourhood of the origin. In addition, an operator $\eta\in \mathcal I_1(\H)$ is an eigenvector of $\Phi\ealpha$ for the eigenvalue $\mu$ if and only if $(\invalpha{})^{+1/2}\eta \,(\invalpha{})^{+1/2}$ is an eigenvector of $\widehat \Phi \ealpha$ for the eigenvalue $(\lambda\ealpha{})\inv \mu$.
   Therefore, $\rho\ealpha=(\invalpha{})^{-1/2}\widehat\rho\ealpha\, (\invalpha{})^{-1/2}$ is an eigenvector of $\Phi\ealpha$ associated with $\lambda\ealpha$,
    the peripheral spectrum of $\Phi\ealpha$ is $\lambda\ealpha S_{z\ealpha}$, and the peripheral eigenvalues are simple.
   Because the definition of $\widehat \Phi \ealpha{}$ does not depend on the free multiplicative constant in $\invalpha$, this $\widehat\rho\ealpha$ is uniquely defined; we can therefore fix the constant in $\invalpha$ so that $\rho\ealpha$ has trace one.
   We now prove that $z\ealpha$ is independent of $\alpha$, and that $u\ealpha$ can be chosen to be constant equal to $u$.

	\begin{lemma} \label{lemma_zalphaconstant}
		With the above notations we have $z\ealpha=z$ for all $\alpha\in\rr$.
	\end{lemma}
	\begin{proof}[Proof of Lemma \ref{lemma_zalphaconstant}]
		We have $\Phi\ealpha{}^*(\proj_j X \proj_j)= \proj_{j+1} \Phi\ealpha{}^*(X) \proj_{j+1}$ for all $X$ in~$\B(\H)$ from the commutation relations for $p_j$ and $V_i$. In addition, each $p_j \invalpha p_j$ is nonzero since $\invalpha$ is positive-definite, and satisfies $\Phi\ealpha{}^*(p_j \invalpha p_j) = \lambda\ealpha p_{j+1} \invalpha p_{j+1}$. This implies by a direct computation that for any $n=0,\ldots,z-1$, the non-zero operator
		$\sum_{j=0}^{z-1} \e^{2\i\pi jn/z} p_j \invalpha p_j$ is an eigenvector of $\Phi\ealpha{}^*$ for the eigenvalue $\lambda\ealpha\e^{-2\i\pi n/z}$, so that $\Phi\ealpha{}^*$ has at least $z$ peripheral eigenvalues, and $z\leq z\ealpha$.

    Lemma \ref{lemma_Talphainverse} shows that this same inequality applied to $\widehat \Phi\ealpha$ in place of $\Phi$ and $-\alpha$ in place of $\alpha$ gives $z\ealpha \leq z$. We therefore have $z=z\ealpha$.
	\end{proof}
	This implies in turn that $u\ealpha$ is an eigenvector of $\widehat \Phi\ealpha{}^*$ for the simple isolated eigenvalue $\overline\theta$, so that we can parameterize $\alpha\mapsto u\ealpha$ to be analytic in a neighbourhood of the origin.

	\begin{lemma} \label{lem:understand_ualpha}
		With the above notation, we have  $u\ealpha=u$ and $[\invalpha,u]=0$, $[\rho\ealpha,u]=0$ for all $\alpha\in \rr$.
	\end{lemma}

	\begin{proof}[Proof of Lemma \ref{lem:understand_ualpha}]
		Consider the simple eigenvalue $\overline\theta$ of $\widehat \Phi\ealpha{}^*$. The associated eigenspace is one-dimensional and contains $u\ealpha$. We also show that $(\invalpha{})^{+1/2}u\, (\invalpha{})^{-1/2}$ is another eigenvector of $\overline\theta$ for $\hat \Phi\ealpha{}^*$:
		\begin{align*}
			\hat \Phi\ealpha{}^*\big((\invalpha{})^{+1/2} u\, (\invalpha{})^{-1/2}\big)
			&= \sum_i \frac{v_i^\alpha}{\lambda\ealpha} (\invalpha{})^{-1/2} V_i^* \, \invalpha u \, V_i (\invalpha{})^{-1/2}\\
			&= \overline\theta  \sum_i \frac{v_i^\alpha}{\lambda\ealpha} (\invalpha{})^{-1/2} V_i^* \invalpha  V_i u\,(\invalpha{})^{-1/2}\\
			&= \frac{\overline \theta}{ \lambda\ealpha} \,(\invalpha{})^{-1/2} \big(\sum_i v_i^\alpha V_i^* \invalpha V_i\big) u\,(\invalpha{})^{-1/2}\\
			&= \overline \theta \,(\invalpha{})^{+1/2} u \,(\invalpha{})^{-1/2}.
		\end{align*}
		We therefore have $(\invalpha{})^{+1/2} u \,(\invalpha{})^{-1/2}=\gamma\ealpha u\ealpha$ for some $\gamma\ealpha\in \cc$, and the relation $u^z=(u\ealpha)^z=\id$ requires that $\gamma\ealpha$ is a $z$th root of unity. Now, $(u\ealpha)^* u\ealpha=\id$ implies that $u^* \invalpha u = \invalpha$, so that $[\invalpha,u]=0$. This finally gives us $u=\gamma\ealpha u\ealpha$  and since we chose $u\ealpha$ to be analytic in~$\alpha$, the phase $\gamma\ealpha$ is necessarily $1$. Last, $[\widehat\rho\ealpha,u\ealpha]=0$ and this implies $[\rho\ealpha,u]=0$.
	\end{proof}

	We can now conclude the proof of Proposition \ref{prop_periphspectrumPhialpha}. The validity of our parameterizations rely only on the fact that the peripheral eigenvalues for $\Phi\ealpha$,  $\widehat\Phi\ealpha$ and $\widehat\Phi\ealpha{}^*$ are isolated. Since the peripheral spectra for these maps are, respectively, $\lambda\ealpha S_z$, $S_z$ and $S_z$, all peripheral eigenvalues are isolated uniformly for $\alpha$ in any compact set containing the origin. This allows us to extend all parameterizations to be analytic on $\rr$. Last, the eigenvector of $\Phi\ealpha$ (resp.\ $\Phi\ealpha{}^*$) associated with the eigenvalue $\lambda\ealpha \theta^m$ is $\rho\ealpha u^m$ (resp. $\invalpha u^{-m}$), and this gives the form of the corresponding spectral projectors.
\end{proof}
The preceding results also give some information about the peripheral spectrum of $\Phi\ealpha$ for complex~$\alpha$, as the following corollary shows:
\begin{corollary} \label{coro_specperiphalphacomplex}
	Let $(V_i)_{i\in I}$ and $v=(v_i)_{i\in I}$ be as in Proposition \ref{prop_periphspectrumPhialpha}.
   For any $\alpha_0$ in $\rr$, there exists a neighbourhood $N_{\alpha_0}$ of $\alpha_0$ in $\cc$, such that for $\alpha$ in $N_{\alpha_0}$ the peripheral spectrum of $\Phi\ealpha$ is of the form $\{\lambda\ealpha \theta^m \,|\, m=0,\ldots,z-1\}$ for some $\lambda\ealpha$ in $\cc$.
\end{corollary}
\begin{proof}
	By Proposition \ref{prop_periphspectrumPhialpha}, the peripheral spectrum of $\Phi^{(\alpha_0)}$ is $\{\lambda^{(\alpha_0)} \theta^m \,|\, m=0,\ldots,z-1\}$. By standard perturbation theory, for $m=0,\ldots,z-1$ there exist analytic functions $\alpha\mapsto \lambda^{(\alpha)}_m$ defined on a neighbourhood of $\alpha_0$, such that $\lambda^{(\alpha_0)}_m=\lambda^{(\alpha_0)} \theta^m$ and the $\lambda^{(\alpha)}_m$ are eigenvalues of $\Phi\ealpha$.
   In particular, there exists a (complex) neighbourhood $N_{\alpha_0}$ of $\alpha_0$ such that for $\alpha$ in $N_{\alpha_0}$, any eigenvalue of $\Phi\ealpha$ of maximum modulus is one of the $\lambda^{(\alpha)}_m$. Denote (consistently with the above notation) by $\invalpha$ an eigenvector of $\Phi\ealpha{}^*$ for $\overline{\lambda}^{(\alpha)}_0$.
   Since $u^{-1} V_i=\theta V_i u^{-1}$ we have $\Phi\ealpha{}^*(\invalpha u^{-m})=\overline{\lambda}^{(\alpha)}_0 \theta^m \invalpha u^{-m}$, so that $\overline{\lambda}^{(\alpha)}_0\,\theta^m$ is an eigenvalue of $\Phi\ealpha{}^*$ for $m=0,\ldots,z-1$.
   Since all such $\lambda^{(\alpha)}_0\,\theta^m$ have the same modulus, one has necessarily $\lambda^{(\alpha)}_m=\lambda^{(\alpha)}_0\,\theta^m$ for $m=0,\ldots,z-1$. The conclusion follows by letting $\lambda\ealpha:=\lambda^{(\alpha)}_0$.
\end{proof}

The next result gives the asymptotics of the spectral radius of~$\Phi\ealpha$ as~$\alpha\to\pm\infty$.
\begin{proposition} \label{prop_asymptoticssprphialpha}
	Let $(V_i)_{i\in I}$ and $v=(v_i)_{i\in I}$ be as in Proposition \ref{prop_periphspectrumPhialpha} and assume that $V_i\neq 0$ for all $i\in I$. Define $v_+=\max_{i\in I} v_i$, $v_-=\min_{i\in I} v_i$, and
	\begin{gather*}
		I_\pm=\{i\in I \,|\, v_i=v_\pm\}, \quad \Phi_\pm=\sum_{i\in I_\pm}V_i \cdot V_i^*, \quad\lambda_\pm=\spr \Phi_\pm.
	\end{gather*}
	Then
	\[\spr \Phi\ealpha = v_\pm^\alpha \big( \lambda_\pm+ o(1)\big)\ \mbox{ for }\ \alpha\to\pm\infty.\]
\end{proposition}
\begin{proof}
	Remark that $\lambda_\pm\neq 0$. The statement follows immediately from $\Phi\ealpha = v_\pm^\alpha \big(\Phi_\pm + o(1)\big)$ and standard perturbation theory.
\end{proof}

\section{Adiabatic theorem for discrete non-unitary evolutions} \label{app:DNUAT}

We devote this section to elements of adiabatic theory that are suitable for discrete non-unitary time evolution. To be precise, the theory is applicable to discrete dynamics arising from a family~$(F(s))_{s \in [0,1]}$ of maps from a Banach space $X$ to itself satisfying
\begin{enumerate}[label=\textbf{Hyp\arabic*},start=0]
	\item\label{it:cont}  The mapping $s \mapsto F(s)$ is a continuous $\B(X)$-valued function of~$s \in [0,1]$;
	\item\label{it:sp-in-D} For all $s \in [0,1]$, $\spr F(s) = 1$;
	\item\label{it:semi-s-no-cross} The peripheral spectrum of $F(s)$ consists of finitely many isolated semi-simple eigenvalues for all ~$s \in [0,1]$;
	\item\label{it:P-C2}  With $P(s)$ the spectral projector of~$F(s)$ onto the peripheral eigenvalues,  the map $s \mapsto F^P(s):=F(s)P(s)$ is a~$C^2$ $\B(X)$-valued function of~$s \in [0,1]$;
	\item\label{it:spr-Q} With $Q(s) := \one - P(s)$,
	$$
	\ell := \sup_{s \in [0,1]} \spr F(s) Q(s) < 1.
	$$
\end{enumerate}
We call such a family \emph{admissible} for our adiabatic theorems. We emphasize that hypotheses are stated in terms of spectral radii, and not of norms as was the case for the hypotheses~\cite{HJPR1}, which we recall here (adapting slightly the notation for coherence) for comparison:

{\footnotesize
	\begin{quote}
		\begin{itemize}
			\item[\textbf{\textup{H1}.}] For all $s \in [0,1]$, $\|F(s)\|\leq 1$, {i.e.} $F(s)$ is a contraction; \label{Hcontraction}

			\item[\textbf{\textup{H2}.}] There is a uniform gap $\epsilon>0$ such that, for $s \in [0,1]$, each \textit{peripheral eigenvalue} $e^j(s)\in \sp F(s)\cap S^1$ is simple, and $|e^j(s)-e^i(s)| > 2\epsilon$ for any $e^j(s)\neq e^i(s)$ in $\sp F(s)\cap S^1$;

			\item[\textbf{\textup{H3}.}] Let $P^m(s)$ be the spectral projector associated with $e^m(s) \in \sp\L(s) \cap S^1$, and $P(s)=\sum_m P^m(s)$ the \textit{peripheral spectral projector}. The map $s\mapsto F^P(s) := F(s) P(s)$ is $C^2$ on $[0,1]$;

			\item[\textbf{\textup{H4}.}]  With $Q(s) := \one - P(s)$,
			\begin{align*}
			\ell:=\sup_{s\in[0,1]}\| F(s)Q(s)\| < 1.
			\end{align*}
		\end{itemize}
	\end{quote}
}

In applications, the Banach space~$X$ is again~$\I_1(\H_\sys)$ equipped with the trace norm, and the role of~$F(s)$ is played by appropriate deformations of the reduced dynamics~$\L(s)$ arising from a repeated interaction system satisfying \ref{ADRIS}.

\subsection{Adiabatic theorem for products of projectors}
We start with a result about products of projectors.

\begin{definition} \label{def_intertwining}
	Let $\big(P^m(s)\big)_{s\in [0,1]}$, $m=1,\ldots,z$ be $C^1$ families of projector-valued operators in a Banach space~$X$, satisfying $\sum_{m=1}^z P^m(s)=\id$ for all $s\in[0,1]$. Let $W: [0,1]\rightarrow \B(X)$ be the family of \emph{intertwining operators} given by
	\begin{equation}\label{intertw}
	W'(s)=\sum_{m=1}^z P^m{}'(s)P^m(s)\,W(s), \ W(0)=\one,
	\end{equation}
	where $P^m{}'(s)$ is the derivative of $s\mapsto P^m(s)$.
\end{definition}
Standard results (see e.g. Section II.5 in \cite{Kato}) imply that \begin{equation} \label{eq_intertwining}
W(s)P^m(0)=P^m(s)W(s)
\end{equation}
for all $s\in[0,1]$ and $m=1,\ldots,z$, and that $W(s)$ is invertible, with inverse $W^{-1}(s)$ solution to $${V}'(s)=-\sum_{m=1}^z V(s)P^m{}'(s)P^m(s).$$ Note that, if we are given a single $C^1$ family $(P(s))_{s \in [0,1]}$ of operators then we can apply the above to $P^1(s)=P(s)$, $P^2(s)=\id-P(s)$.

Remark that we have the immediate relations
\begin{equation} \label{eq_basicrelations}
P^m(s) P^m{}'(s) P^m(s)=0
\end{equation}
for all $s\in[0,1]$ and $m=1,\ldots,z$. For any family $\big(P^m(s)\big)_{s\in [0,1]}$ as above we will denote by $C_P$ the quantity
\begin{equation} \label{eq_ssmPPp}
C_P=\sup_{s\in[0,1]}\sup_{m=1,\ldots,z}\max\big(\|P^m(s)\|,\|P^m{}'(s)\|\big).
\end{equation}

\begin{proposition}\label{adproj}
	Let $\big(P^m(s)\big)_{s\in [0,1]}$  and $W$ be as in Definition \ref{def_intertwining}. Then there exists $C>0$ and $T_0\in\nn$ such that for $T\geq T_0$ and $\{s_k\}_{k\in\{0,1,2,\dots, T\}}\subset [0,1]$ with $|s_k-s_{k-1}|=1/T$ with $k\leq T$ one has
	\begin{equation}\label{adprop}
	\big\|P^m(s_k)P^m(s_{k-1})\cdots P^m(s_0)-W(s_k)P^m(0)\,W^{-1}(s_0)\big\|\leq C/T
	\end{equation}
	where $C$ and $T_0$ depend on $C_P$ defined by \eqref{eq_ssmPPp} only, and $C$ is a continuous function of $C_P$.
\end{proposition}

\begin{remark} \label{remark_expdtproj} Before we prove this, let us mention
	\begin{enumerate}[label={\roman*.}]
		\item For simplicity, differentiability is understood in the norm sense in case $\dim (X)=\infty$.
		\item It is enough that the maps $s\mapsto P^m(s)$ be $C^1$ for this proposition to hold.
		\item The $s_k$'s need not be distinct, except those with consecutive indices.
		\item The norms $\|P^m(s)\|$ are in general larger than one.
	\end{enumerate}
\end{remark}

\begin{proof}
	For any $s, s'\in [0,1]$, we have
	\begin{equation*}
	P^m(s)P^m(s')=W(s)P^m(0)W^{-1}(s)W(s')P^m(0)W^{-1}(s'),
	\end{equation*}
	where, using the shorthand $K(s)=\sum_{k=1}^z P^k{}'(s)P^k(s)$,
	\begin{align}\label{eq:esww}
	W^{-1}(s)W(s')&=\one+W^{-1}(s)(W(s')-W(s))\nonumber\\
	&=\one+W^{-1}(s)\int_s^{s'} K(t)W(t)\,\d t \nonumber \\
	&=\one+\int_s^{s'} \big(W^{-1}(s)W(t)\big) W^{-1}(t)K(t)W(t)\, \d t\nonumber \\
	&=\one+\int_s^{s'}W^{-1}(t)K(t)W(t)\, \d t \nonumber \\
      &\qquad {} +\int_s^{s'}\int_s^t W^{-1}(s) K(u)W(u)\, \d u \; W^{-1}(t)K(t)W(t) \, \d t. \nonumber
	\end{align}
	In addition, by relation \eqref{eq_basicrelations}, $P^m(t)P^m{}'(t)P^m(t)\equiv 0$, so that
	\begin{equation}\label{eq:koff}
	P^m(0)W^{-1}(t)K(t)W(t)P^m(0)\equiv 0
	\end{equation}
	and
	\begin{equation*}
	P^m(s)P^m(s')= W(s)P^m(0)\Big(\one+\int_s^{s'}\int_s^t W^{-1}(s) K(u)W(u)W^{-1}(t)K(t)W(t)\, \d u \d t\Big) P^m(0)W^{-1}(s').
	\end{equation*}
	Denote by $J(s,s')$ the integral term:
	\[J(s,s')=\int_s^{s'}\!\!\int_s^t W^{-1}(s) K(u)W(u)W^{-1}(t)K(t)W(t)\, \d u\,\d t\]
	and
	\[\tilde J(s,s') = J(s,s')+[P^m(0),J(s,s')].\]
	We have
	\begin{equation*}
	P^m(0)J(s,s')P^m(0)=\tilde J(s,s')P^m(0),
	\end{equation*}
	with
	\[
	\max\big(\|J(s,s')\|,\|\tilde J(s,s')\|\big) \leq c(s-s')^2,\]
	for some $c>0$ which is a continuous function of $C_P$, \eqref{eq_ssmPPp}. Using these considerations iteratively on the product (\ref{adprop}), we get
	\begin{align*}
	W(s_k)P^m(0)W(s_k)^{-1}&W(s_{k-1})P^m(0)W^{-1}(s_{k-1})\dotsb W(s_0)P^m(0)W^{-1}(s_0)\\
	&\hspace{.4cm}=W(s_k)P^m(0)\big(\one+J(s_k,s_{k-1})\big)P^m(0)\big(\one+J(s_{k-1},s_{k-2})\big)P^m(0) \\
   &\hspace{0.4cm} \qquad \dotsb P^m(0)\big(\one+J(s_1,s_0)\big)P^m(0)W^{-1}(s_0)\nonumber\\
	&\hspace{.4cm}= W(s_k)\big(\one+\tilde J(s_k,s_{k-1})\big)\big(\one+\tilde J(s_{k-1},s_{k-2})\big) \\
   &\hspace{0.4cm} \qquad \dotsb \big(\one+\tilde J(s_1,s_0)\big)P^m(0)W^{-1}(s_0). \nonumber
	\end{align*}
	We denote by $\one + R_1$ the product
	\[\one + R_1 = \big(\one+\tilde J(s_k,s_{k-1})\big)\big(\one+\tilde J(s_{k-1},s_{k-2})\big)\dotsb \big(\one+\tilde J(s_1,s_0)\big).\]
	With $c$ as above, by a standard combinatorics argument, we get
	\begin{align*}
	\|R_1\| &\leq \sum_{k=1}^{k} (c/T^2)^k\begin{pmatrix}n \cr k\end{pmatrix}=\left(1+c/T^2\right)^k-1 \\
      &\leq \left(1+c/T^2\right)^{T}-1=e^{T\ln(1+c/T^2)}-1,
\end{align*}
	so that $\|R_1\|\leq C'/T$, for $T$ larger than some $T_0$ (which depends only on $c$), where $C'$ has the required properties. This yields the result with a $C$ as stated, since $\sup_{s\in [0,1]}\|W^{\pm 1}(s)\|$ and $\|P^m(0)\|$ satisfy the requirements as well.
\end{proof}

If the projectors $P(s)$ are rank one, they write $P(s)=\phi(s)\,\psi^*(s)$ with $\phi(s)\in X$ and $\psi(s)^*\in X^*$ such that $\psi^*(s)(\phi(s))=1$. In applications, we will consider linear forms associated to the inner product $(M, N) \mapsto \tr(M^* N)$ for~$M$ and~$N$ in~$\B(\H_\sys)$. We then have the following result.

\begin{corollary}\label{rankone}
	Let $\big(P^m(s)\big)_{s\in[0,1]}$, $m=1,\ldots,z$ be $C^1$ families of rank one projectors, i.e.\ $P^m(s)=\phi_m(s)\,\psi_m^*(s)$, $s\in [0,1]$, where the maps $s\mapsto \phi_m(s)$ and $s\mapsto \psi_m^*(s)$ are $C^1$ and non-vanishing for $m=1,\ldots,z$. Then there exist $C>0$  and $T_0\in\nn$ such that for $T\geq T_0$  one has for any $k\leq T$
	\begin{equation*}
	\sup_{m=1,\ldots,z}\big\|P^m(\tfrac{k}{T})P^m(\tfrac{k-1}{T})\dotsb P^m(\tfrac{1}{T})P^m(0)-\e^{-\int_0^{k/T}\psi_m^*(t)\big(\phi_m'(t)\big)\d t}\phi_m(1)\psi_m^*(0)\big\|
      \leq C/T
\end{equation*}
	where $C$ and $T_0$ depend on
	\begin{equation} \label{eq_phipsi}
	c_P=\sup_{s\in[0,1]}\max_{m=1,\ldots,z}\max\big(\|\phi_m(s)\|,\|\phi_m'(s)\|,\|\psi_m^*{}(s)\|,\|\psi_m^*{}'(s)\|\big)
	\end{equation}
	only, and $C$ is a continuous function of $c_P$.
\end{corollary}

\begin{proof}
	We apply Proposition \ref{adproj} with $s_k=k/T$. Since $s_0=0$, and $\psi^*_j(0)$ is a linear form, it is enough to compute $\tilde\phi_m(s):=W(s)\phi_m(0)$. By the intertwining property, $\tilde\phi_m(s)\in \ran P^m(s)$, {i.e.} $\tilde\phi_m(s)=v_m(s)\phi_m(s)$, where $v_m(s)\in \C$. Because of the differential equation \eqref{intertw} and the identity \eqref{eq_basicrelations}, $\tilde\phi_m(s)$ satisfies
	\begin{align*}
	0\equiv P^m(s)\tilde\phi_m'(s) &=\phi_m(s) \psi^*_m(s) \big(v_m'(s) \phi_m(s) + v_m(s) \phi_m'(s)\big) \\
      &= (v_m'(s) + v_m(s) \psi_m^*(s)\big(\phi_m'(s)\big)\big)\phi_m(s),
   \end{align*}
	so that $v_m(s)=\e^{-\int_0^s \psi_m^*(\phi_m(t))\d t}$. This concludes the proof.
\end{proof}

\subsection{Main result} \label{sec_mainadiabatic}
We now turn to our final adiabatic theorem. We recall that an admissible family $\big(F(s)\big)_{s\in[0,1]}$ is one that satisfies \ref{it:cont}--\ref{it:spr-Q}. We denote by $\lambda^m(s)$ and $P^m(s)$, $m=1,\ldots,z$ the peripheral eigenvalues and associated spectral projectors of $F(s)$; assumptions \ref{it:semi-s-no-cross}, \ref{it:P-C2} and standard pertubation theory ensure that one can parameterize eigenvalues such that both $s\mapsto \lambda^m(s)$ and $s\mapsto P^m(s)$ are $C^2$ functions. We denote by $C_P$ and $W(s)$ the constant \eqref{eq_ssmPPp} and the family of intertwining operators in Definition \ref{def_intertwining}.

\begin{theorem}\label{thm:P-Q-decomp}
	If the family $(F(s))_{s \in [0,1]}$ is admissible, then for any $\ell' \in (\ell,1)$ there exist $C>0$  and $T_0\in\nn$ such that for $T\geq T_0$  one has
   \begin{equation*}
      \sup_{k=1,\ldots,T}\big\|F(\tfrac{k}{T})\dotsb F(\tfrac{1}{T}) - \sum_{m=1}^z \big( \prod_{n=1}^k \lambda^m(\tfrac nT)\big) \, W(\tfrac kT) P^m(0)
      - F^Q(\tfrac{k}{T})\dotsb F^Q(\tfrac{1}{T})Q(0)\big\|
   	\leq \frac{C}{T(1-\ell')}
   \end{equation*}
	for all $T \geq T_0$, where $F^Q(s)$ denotes $Q(s)F(s)$. Moreover,
	$$
	\|F^Q(\tfrac{k}{T})\dotsb F^Q(\tfrac{1}{T})Q(0) \| \leq C {\ell'}^k,
	$$
	and $C$ depends on $C_P$ only, is a continuous function of $C_P$, and $T_0$ depends on $C_P$ and $\ell'$ only.
\end{theorem}

\begin{proof}
	The proof consists of revisiting the proof of Theorem 4.4 in \cite{HJPR1}, relaxing the hypotheses made there to \ref{it:cont}--\ref{it:spr-Q}.
	We first focus on the combinatorial part of the proof stated as Proposition 4.5 in \cite{HJPR1}, borrowing freely the notation used there. We denote in particular $F^P(s)=P(s) F(s)$ and $F_k^\#=F(\tfrac kT)^\#$, where $\#\in\{P,Q\}$. Under our hypotheses on the spectral radii instead of the assumptions on the norm used in \cite{HJPR1}, we will see below that the starting estimates
	\begin{equation} \|\prod_{a\in A_n}  F_a^Q \| \leq \ell^{|A_n|}, \qquad  \|\prod_{b\in B_n}  F_b^P \| \leq 1. \tag{A9 of \cite{HJPR1}} \label{eq:A9_of_HJPR}\end{equation}
	used in Proposition 4.5 in \cite{HJPR1} are replaced by the following bounds.
	For some constants $D,D'>1$, and $\ell' < 1$,
	\begin{equation} \|\prod_{a\in A_n}  F_a^Q \| \leq D(\ell')^{|A_n|}, \qquad  \|\prod_{b\in B_n}  F_b^P \| \leq D', \tag{NewBounds} \label{eq:modified_bounds_on_consec_products}\end{equation}
	where $D,D'$ are independent of the number of terms in the products, and satisfy the required dependence. Following Proposition 4.5 in \cite{HJPR1},
	we need to bound the norms of terms of the following forms:
	\begin{gather}
	\big(\prod_{a\in A_d} F_a^Q\big)\big(\prod_{b\in B_d} F_b^P\big)\ldots \big(\prod_{a\in A_1} F_a^Q\big)\big(\prod_{b\in B_1} F_b^P\big)\, P_0, \label{eq_factorform1}\\
	\big(\prod_{b\in B_{d+1}} F_b^P\big)\big(\prod_{a\in A_d} F_a^Q\big)\big(\prod_{b\in B_d} F_b^P\big)\ldots \big(\prod_{a\in A_1} F_a^Q\big)\big(\prod_{b\in B_1} F_b^P\big)\, P_0, \label{eq_factorform2}\\
	\big(\prod_{a\in A_{d+1}} F_a^Q\big)\big(\prod_{b\in B_{d+1}} F_b^P\big)\ldots \big(\prod_{a\in A_2} F_a^Q\big)\big(\prod_{b\in B_2} F_b^P\big) \big(\prod_{a\in A_1} F_a^Q\big)\, P_0,
	\label{eq_factorform3}\\
	\big(\prod_{b\in B_{d+1}} F_b^P\big)\big(\prod_{a\in A_d} F_a^Q\big)\big(\prod_{b\in B_{d-1}} F_b^P\big)\ldots \big(\prod_{a\in A_2} F_a^Q\big)\big(\prod_{b\in B_2} F_b^P\big) \big(\prod_{a\in A_1} F_a^Q\big)\, P_0. \label{eq_factorform4}
	\end{gather}
	In \cite{HJPR1}, we obtain the bounds
	\begin{equation}
	\begin{aligned}
	\|\eqref{eq_factorform1}\| &\leq (c/T)^{2d-1} \, \ell^{\sum_n |A_n|}, &
	\|\eqref{eq_factorform2}\| &\leq (c/T)^{2d} \, \ell^{\sum_n |A_n|},\\
	\|\eqref{eq_factorform3}\| &\leq (c/T)^{2d+1} \, \ell^{\sum_n |A_n|}, &
	\|\eqref{eq_factorform4}\| &\leq (c/T)^{2d} \, \ell^{\sum_n |A_n|},
	\end{aligned} \tag{Bounds from HJPR15}
	\end{equation}
	via \eqref{eq:A9_of_HJPR}. Instead, if we use \eqref{eq:modified_bounds_on_consec_products}, we obtain the bounds
	\begin{align*}
	\|\eqref{eq_factorform1}\| &\leq (c/T)^{2d-1} \, D^d\ell'^{\sum_n |A_n|} D'^d, &
	\|\eqref{eq_factorform2}\| &\leq (c/T)^{2d} \, D^{d+1}\ell'^{\sum_n |A_n|} D'^d,\\
	\|\eqref{eq_factorform3}\| &\leq (c/T)^{2d+1} \, D^{d+1}\ell'^{\sum_n |A_n|} D'^{d+1}, &
	\|\eqref{eq_factorform4}\| &\leq (c/T)^{2d} D^{d+1}\, \ell'^{\sum_n |A_n|} D'^d.
	\end{align*}
	For the bound on form \eqref{eq_factorform1}, for example, since we assumed $D,D'>1$, then we may simply change $c \to cDD'$ and obtain essentially the same bound as in \cite{HJPR1}. Thus, the rest of the combinatorial argument consisting in counting the number of such terms for each~$d$, multiplying by this bound, and summing over~$d$, yields the same final bounds as in Proposition~4.5 in \cite{HJPR1}, up to modified constants.

	We now turn to find $D, D',\ell'$ such that we have \eqref{eq:modified_bounds_on_consec_products}.
	\begin{lemma}
		Under the assumptions \textup{\ref{it:cont}--\ref{it:spr-Q}}, there exist $T_0\in \nn$ and $D'>0$, both depending on $C_P$ only, and $D'$ being a continuous function of $C_P$, such that for all $T\geq T_0$  and $n_0 < n \leq T$, we have
		\[
		\| F^P_{n} \dotsm F^P_{n_0}\| < D'.
		\] \label{lem:boundLPs}
	\end{lemma}

	\begin{proof}
		For each $k\in \{n_0,\dotsc,n\}$, write $F_k^P = \sum_{m=1}^z \lambda^m_k P^m_k$ using semisimplicity of peripheral eigenvalues. Recall that for each $m$, $\lambda^m(s)$ and $P^m(s)$ are $C^2$ in $[0,1]$.
		Then,
		\begin{align*}
		\prod_{k=n_0}^n F^P_k &= \prod_{k=n_0}^n \sum_{m=1}^z \lambda^m_k P^m_k \\
		&= \sum_{m=1}^z \left(\prod_{k=n_0}^n \lambda^m_k\right)P^m_n \dotsm P^m_{n_0} + \sum_{\substack{i_{n_0},\dotsc,i_n=1,\dotsc,z \\ \text{not all equal} }}\left(\prod_{k=n_0}^n \lambda^{i_k}_k\right) P^{i_n}_n\dotsm P^{i_{n_0}}_{n_0}.
		\end{align*}
		For each $m$ and all $k'\leq k$, Proposition \ref{adproj} gives for $T\geq T_0$
		\begin{equation}
		\|P^m_k \dotsm P^m_{k'}\| \leq C/T+ \sup_{1\geq s>s'\geq 0} \| P^m(s)W(s)W(s')\inv\|  \leq  C(1+1/T), \label{eq:bound_product_of_projectors}
		\end{equation}
		for some constant~$C$ which depends continuously on $C_P$. Again we can bound this $C(1+1/T)$ by a new constant~$C$ with the same properties as the original~$C$. Taking~$k=n$ and~$k'=n_0$ in \eqref{eq:bound_product_of_projectors} bounds the first sum by $z\,C$, since the eigenvalues are on the unit circle.
		For the second sum, we know that $P^{i}_k P^m_{k-1} \leq C_P/T$ if $i\neq m$,  as a consequence of the relation $P^i(s)P^m(s)=0$. Bounding the terms in the second sum is again done by a simple combinatorial argument: let $d$ be the number of transitions  $P^m_k P^{i}_{k-1}$ where $m\neq i$. In each term in the second sum, there is at least one such transition by design. If we have $n-n_0$ stars representing projectors, and $d$ bars representing transitions, then there are $n-n_0-1$ gaps between the stars, from which we need to choose $d$ to put a bar. So $\binom{n-n_0-1}{d}$ is the number of ways to divide the projectors into groupings. For each grouping, we have at most $z$ choices of which projector it should be. So in total, there are at most $\binom{n-n_0-1}d z^{d+1}$
		terms with $d$ transitions. Each such term has norm bounded by $C^{d+1}C_P^d/T^d$ via \eqref{eq:bound_product_of_projectors}, and using that each transition yields a factor $C_P/T$, and that all the eigenvalues have modulus one. Lastly, there cannot be more than $n-n_0-1$ transitions (in fact, fewer than $(n-n_0)/z$). So, in total, we may bound the second sum by
		\begin{align*}
		\sum_{d=1}^{n-n_0-1} zC{n-n_0-1 \choose d}\frac{(z C C_P)^d}{T^d} &\leq zC \left(1 + \frac{z C C_P}{T} \right)^{n-n_0-1} \\
		&\leq zC\left(1 + \frac{zC C_P}{T}\right)^T \leq zC\exp(zC C_P).
		\end{align*}
		In total then, we have
		\[
		\|F^P_n\dotsm F^P_{n_0}\| \leq zC(1 + \exp(z C C_P))=:D'. \qedhere
		\]
	\end{proof}

	Let us turn now to compositions of $\L^Q$'s.
		\begin{lemma} \label{lem:boundLQs}
		Under the assumptions  \textup{\ref{it:cont}--\ref{it:spr-Q}}, for all $\ell'\in (\ell,1)$, there exists $D>0$, $T_0\in\nn$ depending on $C_P$ only, and $D$  a continuous function of $C_P$, such that for any $T\geq T_0$ and $n_0 < n\leq T$, we have
		\[
		\|F^Q_n\dotsm F^Q_{n_0}\| \leq D \ell'^{n-n_0}.
		\]
	\end{lemma}
\begin{proof}
		Let $ \ell' \in (\ell,1)$. As shown in Lemma 4.3 of \cite{HJPR1}, the spectral hypothesis \ref{it:spr-Q} and the regularity assumptions \ref{it:cont}, \ref{it:P-C2} imply that given $\epsilon>0$, we may choose $m\in \N$ uniformly in $s$ so that
		\begin{equation}\label{defmq}
		\|F^Q(s)^m\| < (\ell+\epsilon)^m.
		\end{equation}
		Choose $\epsilon$ so that $\ell+\epsilon <  \ell'$.

		We now deduce by induction that for all $\N\ni m<T$, $\|F^Q_{k+m}\dotsm F_k^Q  - (F_k^Q)^{m+1}\| \leq \delta_m(T)$ for some map $\delta_m: \R_*^+\rightarrow \R_*^+$, independent of $k$, such that  $\lim_{T\rightarrow \infty}\delta_m(T)=0$. This is trivially true for $m=0$ and for $m\geq 1$, we have
      \begin{align*}
         	F^Q_{k+m}\dotsm F_k^Q-(F_k^Q)^{m+1}
               &=(F^Q_{k+m}-F_k^Q)(F^Q_{k+m-1}\dotsm F_k^Q) \\
               &\qquad {} + F_k^Q\big((F^Q_{k+m-1}\dotsm F_k^Q)-(F_k^Q)^{m-1}\big).
      \end{align*}
		With $F^Q_j=F^Q(s=\frac{j}{T})$,
		the first term is bounded above by
		\begin{equation}\label{eq:bidel}
      \begin{split}
		&\big\|F^Q\big(\tfrac{k+m}{T}\big)-F^Q\big(\tfrac{m}{T}\big)\big\|\big(\sup_{s\in [0,1]}\|F^Q(s)\|\big)^m \\
		&\qquad \qquad \qquad \leq \sup_{\substack{s\in [0,1]\\ s+\frac{m}{T}\in[0,1]}}\big\|F^Q\big(s+\tfrac{m}{T}\big)-F^Q\big(\tfrac{m}{T}\big)\big\|\big(\sup_{s\in [0,1]}\|F^Q(s)\|\big)^m.
      \end{split}
      \end{equation}
		This expression goes to zero as $T\rightarrow \infty$ by uniform continuity of $F^Q$ on $[0,1]$, and depends parametrically on $m$ only. The second term is bounded above by $\sup_{s\in [0,1]}\|F^Q(s)\|\delta_{m-1}(T)$, by induction hypothesis, hence the claim is proved with $\delta_m(T)=(\ref{eq:bidel}) + \sup_{s\in [0,1]}\|F^Q(s)\|\delta_{m-1}(T)$.

		Thus, given $\epsilon$ and $m$ chosen so that (\ref{defmq}) above holds, for $\tilde{\ell} = (\ell+\epsilon)^m +\delta_m(T)$, we have $\|F^Q_{k+m}\dotsm F_k^Q\| < \tilde{\ell}$,  where $\tilde{\ell}<1$ if $T>T(m,\epsilon)$, for some $T(m, \epsilon)$ large enough.

		Any integer $p\geq m$ can be partitioned  as
		\[
		p = \underbrace{m+ m+\dotsm +m }_{[p/m] \text{ times}}+ \underbrace{(p- m[p/m])}_{<m},
		\]
		so that for $c = \sup_{h=1,2,\dotsc,m-1} \left(\sup_{s \in [0,1]} \|F^Q (s)\|\right)^h$, the previous bound yields,
		\[
		\|F_{k+p}^Q \dotsm F_k^Q\| \leq c \tilde{\ell}^{[p/m]} \leq c \tilde{\ell}^{(p/m)-1} = \frac{c}{\tilde{\ell}} \left(\tilde{\ell}^{1/m}\right)^p.
		\]
		Note that here, in fact, we do not need $p\geq m$. If $p<m$, then $\| F^Q_{k+p}\dotsm F_k^Q\| \leq c$, $\frac{1}{\tilde{\ell}} \tilde{\ell}^{p/m} > 1$, and the bound still holds. Finally, set $\ell'' = (\tilde{\ell})^{1/m}$.
		Since $(\ell+ \epsilon)^m + \delta_m(T) \leq ((\ell+ \epsilon) + \delta_m(T)^{1/m})^m$, we have $\ell'' < \ell + \epsilon + \delta_m(T)^{1/m}$. For $T > T'(m,\epsilon)$ for some $T'(m, \epsilon)$ large enough, we have $\delta_m(T)^{1/m} < \ell' - (\ell+ \epsilon)$, and therefore $\ell'' < \ell'$. Thus, for $D= c/\tilde{\ell}$, we have
		\[
		\|F_{n}^Q \dotsm F_{n_0}^Q\| \leq D \ell''^{n-n_0} < D \ell'^{n-n_0}. \qedhere
		\]
	\end{proof}

	Our next ingredient is to approximate the composition $F^P_k\cdots F^P_0P(0)$, {i.e.}
	to show the equivalent of Proposition~4.6 in~\cite{HJPR1}.

	\begin{lemma}\label{lem:boundFps}
		Under the assumptions and notation of Proposition \ref{thm:P-Q-decomp}, there exist $T_0$ and $C > 0$ with the same properties as in that Proposition, such that
		\begin{equation}
		\sup_{k=1,\ldots,T} \big\| F^P(\tfrac{k}{T})\dotsb F^P(\tfrac{1}{T})P(0) - \sum_{m=1}^z \big( \prod_{n=1}^k \lambda_n^m\big) \, W(\tfrac kT) P^m(0)\big\| \leq C/T.
		\end{equation}

	\end{lemma}

	\begin{proof}
		We define for all $k\leq T$
		\begin{align}
		\cK_{k}&=W(\tfrac{k}{T})\,P(0), \ \ \ \cK^\dagger_{k}=P(0)\,W^{-1}(\tfrac{k}{T}), \\
		\Phi_k&=\sum_{j}\big(\prod_{n=1}^k\lambda^m_n\big)P^m(0), \ \ \ \Phi_k^\dagger=\sum_{j}\big(\prod_{n=1}^k\bar \lambda^m_n\big)P^m(0),\\
		A_k&=\cK_k\Phi_k, \ \ \ A_k^\dagger=\Phi_k^\dagger \cK^\dagger_{k}.
		\end{align}
		The above expression for $A_k$ gives in particular that
		\begin{equation} \label{eq:form_of_Ak}
		A_k = \sum_j \big( \prod_{n=1}^k \lambda_n^m\big) \, W(\tfrac kT) P^m(0).
		\end{equation}
		We have the following  identities which are consequences of the properties of the intertwining operators $W$:
		\begin{gather}
		\cK_k\cK_k^\dagger= P_k, \quad \cK_k^\dagger\cK_k=P(0), \quad \cK_k P^m(0)=P_k^m\cK_k, \quad \cK_k^\dagger P^m_k = P^m(0)\cK_k^\dagger,\\
		\Phi_k P^m(0)= P^m(0) \Phi_k, \quad \Phi_k^\dagger P^m(0)= P^m(0) \Phi_k^\dagger, \quad  \Phi_k  \Phi_k ^\dagger=P(0)= \Phi_k ^\dagger \Phi_k,
		\intertext{and}
		\cK_{k}=\cW_k\cK_{k-1}, \quad \cK_{k}^\dagger = \cK_{k-1}^\dagger \cW_k^\dagger,
		\intertext{where}
		\cW_k=W(\tfrac{k}{T})W^{-1}(\tfrac{k-1}{T})P(\tfrac{k-1}{T}),  \quad
		\cW_k^\dagger=W(\tfrac{k-1}{T})W^{-1}(\tfrac{k}{T})P(\tfrac{k}{T}).
		\end{gather}
		Hence $A_k$ is uniformly bounded in $T$ and $k\leq T$, since $\cK_k$ and $\Phi_k$ are, and satisfies relevant intertwining properties. The notation is chosen to be close to that used in the proof of Proposition~4.6 in~\cite{HJPR1}, and one gets that all steps of that of Proposition~4.6 in~\cite{HJPR1} go through, which ends the proof.
	\end{proof}

	This concludes the proof of Theorem \ref{thm:P-Q-decomp}.
\end{proof}
Combining Theorem \ref{thm:P-Q-decomp} with the proof of Corollary \ref{rankone} immediately implies the following:
\begin{corollary}\label{coro:P-Q-decomp}
	If the family $(F(s))_{s \in [0,1]}$ is admissible, and its peripheral eigenvalues are simple, with associated projectors $P^m(s)=\phi_m(s)\psi_m^*(s)$
	then for any $\ell' \in (\ell,1)$ there exist $C > 0$ and $T_0 \in \nn$, such that for $T\geq T_0$ and $k\leq T$,
   \begin{equation*}
      \big\|F(\tfrac{k}{T})\dotsb F(\tfrac{1}{T}) - \sum_{m=1}^z \big( \prod_{n=1}^k \lambda_n^m\big) \,\e^{-\int_0^{k/T}\psi_m^*(t)\big(\phi_m'(t)\big)\d t}\phi_m(1)\psi_m^*(0)
      - F^Q(\tfrac{k}{T})\dotsb F^Q(\tfrac{1}{T})Q(0)\big\|
   	\leq \frac{C}{T(1-\ell')}.
   \end{equation*}
	Moreover,
	$$
	\|F^Q(\tfrac{k}{T})\dotsb F^Q(\tfrac{1}{T})Q(0) \| \leq C {\ell'}^k
	$$
	and $C$ depends on $c_P$ defined in \eqref{eq_phipsi} only, is a continuous function of $c_P$, and $T_0$ depends on $c_P$ and $\ell'$ only.
\end{corollary}

\section{Proofs for Section~\ref{sec_propertiesfullstats}}\label{sec:proofs_for_FS}

\begin{proof}[Proof of Lemma~\ref{lem:balance_eq_traj}]
	By definition
	\begin{align*}
	\varsigma_{T}(\ai,\af,\vec{\imath},\vec{\jmath}) 
   &= \log \frac{\tr\big( U_{T}\dotsm U_1   ( \pi\init_{\ai} \otimes \Pi_{\vec \imath}) (\rho\init \otimes \Xi)( \pi\init_{\ai} \otimes \Pi_{\vec \imath}) U_1^* \dotsm U_T^*  (\pi\fin_{\af}\otimes \Pi_{\vec \jmath})\big)}{\tr\big( U_1^*\dotsm U_T^*(\pi\fin_{\af} \otimes \Pi_{\vec \jmath})(\rho\fin_T \otimes \Xi) (\pi\fin_{\af} \otimes \Pi_{\vec \jmath}) U_T \dotsm U_1 (\pi\init_{\ai} \otimes \Pi_{\vec \imath})\big)} \\
	\intertext{using assumptions {i.}, {ii.} and {iii.},}
	&\quad =  \log \frac{\frac{\tr (\sysstate\init \pi\init_{\ai})}{\dim \pi\init_{\ai}} \prod_{k=1}^T \frac{\tr(\envstate[k] \Pi_{i_k}^{(k)})}{\dim \Pi_{i_k}^{(k)}} \tr\big( U_{T}\dotsm U_1   ( \pi\init_{\ai} \otimes \Pi_{\vec \imath})  U_1^* \dotsm U_T^*  (\pi\fin_{\af}\otimes \Pi_{\vec \jmath})\big)}
	{\frac{\tr (\sysstate\fin \pi\fin_{\af})}{\dim \pi\fin_{\af}} \prod_{k=1}^T \frac{\tr(\envstate[k] \Pi_{j_k}^{(k)})}{\dim \Pi_{j_k}^{(k)}} \tr\big( U_1^*\dotsm U_T^*(\pi\fin_{\af} \otimes \Pi_{\vec \jmath})U_T \dotsm U_1 (\pi\init_{\ai} \otimes \Pi_{\vec \imath})\big)} \\
	&\quad =  \log\frac{\tr (\sysstate\init \pi\init_{\ai}) \dim \pi\fin_{\af}}{\tr (\sysstate\fin \pi\fin_{\af}) \dim \pi\init_{\ai}} + \log \frac{  \tr\big( U_{T}\dotsm U_1   ( \pi\init_{\ai} \otimes \Pi_{\vec \imath})  U_1^* \dotsm U_T^*  (\pi\fin_{\af}\otimes \Pi_{\vec \jmath})\big)}
	{ \tr\big( U_1^*\dotsm U_T^*(\pi\fin_{\af} \otimes \Pi_{\vec \jmath})U_T \dotsm U_1 (\pi\init_{\ai} \otimes \Pi_{\vec \imath})\big)} \\
	&\quad \qquad {} + \log \prod_{k=1}^T \frac{ \tr(\Exp{-\beta_k \henv[k]} \Pi_{i_k}^{(k)})}{Z_k \dim \Pi_{i_k}^{(k)}} - \log \prod_{k=1}^T \frac{ \tr(\Exp{-\beta_k \henv[k]} \Pi_{j_k}^{(k)})}{Z_k \dim \Pi_{j_k}^{(k)}}  \\
	\intertext{
      and because $\frac{\tr(\e^{-\beta_k h_{\env_k}} \Pi_{i_k}^{(k)})}{\dim \Pi_{i_k}^{(k)}}= \exp\big({-\beta_k \frac{\tr(\e^{-\beta_k h_{\env_k}\Pi_{i_k}^{(k)}})}{\dim \Pi_{i_k}^{(k)}}}\big)$ by assumption {iii.} again,
      }
	&\quad =  \log\frac{\tr (\sysstate\init \pi\init_{\ai}) \dim \pi\fin_{\af}}{\tr (\sysstate\fin \pi\fin_{\af}) \dim \pi\init_{\ai}} + \log \frac{  \tr\big( U_{T}\dotsm U_1   ( \pi\init_{\ai} \otimes \Pi_{\vec \imath})  U_1^* \dotsm U_T^*  (\pi\fin_{\af}\otimes \Pi_{\vec \jmath})\big)}
	{ \tr\big( U_1^*\dotsm U_T^*(\pi\fin_{\af} \otimes \Pi_{\vec \jmath})U_T \dotsm U_1 (\pi\init_{\ai} \otimes \Pi_{\vec \imath})\big)}  \\
	&\quad \qquad {} + \sum_{k=1}^T \frac{ \tr({-\beta_k \henv[k]} \Pi_{i_k}^{(k)})}{\dim \Pi_{i_k}^{(k)}} - \sum_{k=1}^T \frac{ \tr({-\beta_k \henv[k]} \Pi_{j_k}^{(k)})}{\dim \Pi_{j_k}^{(k)}}
	\end{align*}
	and the last term vanishes by cyclicity of the trace.
\end{proof}

\begin{proof}[Proof of Proposition~\ref{prop:averaged_balance_equation}]
	We start with the relation \eqref{eq:ssys-avg}. On one hand,
	\begin{align*}
	\E_T\Big( \log \frac{\tr(\pi\init_{\ai} \sysstate\init)}{\dim \pi\init_{\ai}} \Big)
   &= \sum_{\ai} \log \frac{\tr(\pi\init_{\ai} \sysstate\init)}{\dim \pi\init_{\ai}} \sum_{\af, \vec{\imath}, \vec{\jmath}} \bP^F_{T}(\ai,\af, \vec{\imath}, \vec{\jmath}) \\
	&\qquad = \sum_{\ai} \log \frac{\tr(\pi\init_{\ai} \sysstate\init)}{\dim \pi\init_{\ai}} \\
      &\qquad\qquad \sum_{\af, \vec{\imath}, \vec{\jmath}}  \tr\big( U_{T}\dotsm U_1   ( \pi\init_{\ai}\sysstate\init \pi\init_{\ai} \otimes \Pi_{\vec \imath}\,\Xi\, \Pi_{\vec \imath})  U_1^* \dotsm U_T^*  (\pi\fin_{\af}\otimes \Pi_{\vec \jmath})\big) \\
	&\qquad = \sum_{\ai} \log \frac{\tr(\pi\init_{\ai} \sysstate\init)}{\dim \pi\init_{\ai}} \\
      &\qquad\qquad \sum_{ \vec{\imath}}  \tr\big( U_{T}\dotsm U_1   ( \pi\init_{\ai}\sysstate\init \pi\init_{\ai} \otimes \Pi_{\vec \imath}\,\Xi \,\Pi_{\vec \imath})  U_1^* \dotsm U_T^*  (\one \otimes \one)\big).
	\end{align*}
	Using that each $\envstate[k]$ is a function of $Y_k$, we have $\sum_{\vec{\imath}} \Pi_{\vec{\imath}} \,\Xi\, \Pi_{\vec{\imath}} = \Xi$, and
	\[
	\E_T\Big( \log \frac{\tr(\pi\init_{\ai} \sysstate\init)}{\dim \pi\init_{\ai}} \Big) = \sum_{\ai} \log \frac{\tr(\pi\init_{\ai} \sysstate\init)}{\dim \pi\init_{\ai}}  \tr\big( U_{T}\dotsm U_1   ( \pi\init_{\ai}\sysstate\init \pi\init_{\ai} \otimes \Xi )  U_1^* \dotsm U_T^* \big).
	\]
	As $\sysstate\init$ is a function of $\Ai$, and using the cyclicity of the trace, we have
	\begin{align*}
\E_T\Big( \log \frac{\tr(\pi\init_{\ai} \sysstate\init)}{\dim \pi\init_{\ai}} \Big)&= \sum_{\ai} \log \frac{\tr(\pi\init_{\ai} \sysstate\init)}{\dim \pi\init_{\ai}}  \tr\big(\pi\init_{\ai} \otimes \Xi \big) \frac{\tr(\pi\init_{\ai} \sysstate\init)}{\dim \pi\init_{\ai}}, \\
	&= \sum_{\ai}  \tr(\pi\init_{\ai} \sysstate\init) \log \frac{\tr(\pi\init_{\ai} \sysstate\init)}{\dim \pi\init_{\ai}}.
	\end{align*}
	Using the commutation relation $[\pi\init_{\ai},\sysstate\init] = 0$, the right-hand side is precisely $-S(\rho\init)$. The term involving $\sysstate\fin$ is treated similarly.

	We turn to \eqref{eq:senv-avg}:
	\begin{align*}
	&\E_T\Big( \sum_{k=1}^T \beta_k (E_{j_k}^{(k)} - E_{i_k}^{(k)}) \Big) \\
   &\qquad = \sum_{\vec{\imath}, \vec{\jmath}} \sum_{k=1}^T  \beta_k (E_{j_k}^{(k)} - E_{i_k}^{(k)}) \sum_{\ai,\af} \bP^F_{T}(\ai,\af, \vec{\imath}, \vec{\jmath}) \\
	\intertext{using $[\sysstate\init,\Ai] = 0$,}
	&\qquad = \sum_{\vec{\imath}, \vec{\jmath}} \sum_{k=1}^T  \beta_k (E_{j_k}^{(k)} - E_{i_k}^{(k)}) \tr\big( U_{T}\dotsm U_1   ( \sysstate\init  \otimes \Pi_{\vec \imath}\Xi \Pi_{\vec \imath})  U_1^* \dotsm U_T^*  (\one \otimes \Pi_{\vec \jmath})\big) \\
	&\qquad = \sum_{\vec{\jmath}} \sum_{k=1}^T  \beta_k E_{j_k}^{(k)} \tr\big( U_{T}\dotsm U_1   ( \sysstate\init  \otimes \Xi )  U_1^* \dotsm U_T^*  (\one \otimes \Pi_{\vec \jmath})\big) \\
	&\qquad\qquad{} - \sum_{\vec{\imath}} \sum_{k=1}^T  \beta_k E_{i_k}^{(k)} \tr\big( U_{T}\dotsm U_1   ( \sysstate\init  \otimes \Pi_{\vec \imath}\Xi \Pi_{\vec \imath})  U_1^* \dotsm U_T^*  (\one \otimes \one)\big) \\
	&\qquad= \sum_{k=1}^T \sum_{j_k} \beta_k  \tr\big( U_{T}\dotsm U_1   ( \sysstate\init  \otimes \Xi )  U_1^* \dotsm U_T^*  (\one \otimes \Pi_{j_k}^{(k)} E_{j_k}^{(k)})\big)\\
	&\qquad\qquad{} - \sum_{k=1}^T \sum_{i_k} \beta_k \tr\big(  \sysstate\init  \otimes \Xi\, \Pi_{i_{k}}^{(k)} E_{i_k}^{(k)}  \big).
	\end{align*}
	Using that each $\envstate[k]$ is a function of $Y_k$, we have $\sum_{i_k} \Pi_{i_{k}}^{(k)} E_{i_k}^{(k)} = \henv[k]$ and thus
	\[
	\E_T\Big( \sum_{k=1}^T \beta_k (E_{j_k}^{(k)} - E_{i_k}^{(k)}) \Big)= \sum_{k=1}^T \beta_k \tr(\envstate[k]\fin \henv[k]) - \sum_{k=1}^T \beta_k \tr(\envstate[k]\init \henv[k]). \qedhere
	\]
\end{proof}

\begin{proof}[Proof of Proposition~\ref{prop:phiTalpha}]
	\,\hfill\\
	Assume that $[Y(s),\xi(s)]=0$ for all $s$. Then by definition and from expression \eqref{eq:def_P_F^T}
		\begin{align*}
		\ee\big(\e^{\alpha \rvY}\big)
		&=
		\sum_{\vec\imath,\vec\jmath}\sum_{\ai,\af} \exp \Big(\alpha \sum_{k=1}^Y (y_{j_k}^{(k)}- y_{i_k}^{(k)})\Big)\, \pp^F_{T}(\ai,\af,\vec\imath,\vec\jmath)\\
		&= \sum_{\vec\imath,\vec\jmath} \tr\Big( U_{T}\dotsm U_1   \Big(\sum_{\ai} \pi\init_{\ai} \rho\init\pi\init_{\ai}  \otimes \prod_{k=1}^T \e^{-\alpha y_{i_k}^{(k)}}\Pi_{i_k}\,\Xi\Big)
         \\ & \qquad\qquad\qquad
         U_1^* \dotsm U_T^*  \Big(\sum_b\pi\fin_{\af}\otimes \prod_{k=1}^T \e^{\alpha y_{j_k}^{(k)}}\Pi_{j_k}\Big)\Big)\\
		&= \tr\Big( U_{T}\dotsm U_1   \Big(\sum_{\ai} \pi\init_{\ai} \rho\init\pi\init_{\ai}  \otimes \prod_{k=1}^T\e^{-\alpha Y_k}\,\Xi\Big)
         \\ & \qquad\qquad\qquad
         U_1^* \dotsm U_T^*  \,\Big(\id\otimes  \prod_{k=1}^T\e^{\alpha Y_k}\Big)\Big)\\
		&= \tr \big(\L_Y\ealpha(\tfrac TT)\circ\ldots\circ \L_Y\ealpha(\tfrac1T)(\sum_{\ai} \pi\init_{\ai} \rho\init\pi\init_{\ai})\big).
	\end{align*}
	Assume in addition that  $[\Ai,\sysstate\init]=0$. Then similarly, for $\alpha_1,\alpha_2$ in $\rr$,
	\begin{align*}
		\ee\big(\e^{\alpha_1 \rvY+\alpha_2 \rvW}\big)
		&=
		\sum_{\vec\imath,\vec\jmath}\sum_{\ai,\af} \exp \Big(\alpha_1 \sum_{k=1}^Y (y_{j_k}^{(k)}- y_{i_k}^{(k)})\Big)\,\e^{\alpha_2(\ai-\af)}\, \pp^F_{T}(\ai,\af,\vec\imath,\vec\jmath)\\
		&\quad = \sum_{\vec\imath,\vec\jmath} \tr\Big( U_{T}\dotsm U_1   \big(\sum_{\ai} \e^{+\alpha_2 \ai}\pi\init_{\ai} \rho\init  \otimes \prod_{k=1}^T \e^{-\alpha_1 y_{i_k}^{(k)}}\Pi_{i_k}\,\Xi\big)
         \\ & \quad\qquad\qquad\qquad
         U_1^* \dotsm U_T^*  (\sum_b \e^{-\alpha_2 \af}\pi\fin_{\af}\otimes \prod_{k=1}^T \e^{\alpha_1 y_{j_k}^{(k)}}\Pi_{j_k})\Big)\\
		&\quad =  \tr\Big( U_{T}\dotsm U_1   \big(\e^{+\alpha_2 \Ai} \rho\init  \otimes \prod_{k=1}^T\e^{-\alpha_1 Y_k}\,\Xi\big) U_1^* \dotsm U_T^*  (\e^{-\alpha_2 \Af}\otimes \prod_{k=1}^T \e^{\alpha_1 Y_k})\Big)\\
		&\quad = \tr \big(\e^{-\alpha_2 \Af}\L_Y^{(\alpha_1)}(\tfrac TT)\circ\ldots\circ \L_Y^{(\alpha_1)}(\tfrac1T)( \e^{+\alpha_2 \Ai} \rho\init )\big).
	\end{align*}
\end{proof}
\begin{proof}[Proof of Remark~\ref{rem:reg}]
	The dependence in $s$ will be treated last, since it fixes the overall regularity, as we will see. We consider the applications between the different Banach spaces involved  in the definition of $\L_Y^{(\alpha)}=\tr_\env(e^{\alpha Y}U^{(\tau)}(\cdot\otimes \xi)e^{-\alpha Y}U^{(-\tau)})$. We will denote the trace norm by $\|\cdot \|_1$ when the underlying Hilbert space is determined by the context, and we use the shorthand $\H_\T:=\H_\sys\otimes \H_\env$ for the total Hilbert space.

	\begin{enumerate}
		\item The map $P$ such that $\xi \mapsto P_\xi=\cdot \otimes \xi$ mapping $\I_1(\H_\env)$ to $\B(\I_1(\H_\sys), \I_1(\H_T))$ is a linear isometry, hence a $C^\infty$ map. Indeed, linearity is immediate.  For all $\eta\in \I_1(\H_\sys)$, $\|P_\xi(\eta)\|_{\I_1(\H_T)}=\|\eta \otimes \xi\|_1=\|\eta\|_1\|\|\xi\|_1$, which shows $\|P_\xi\|_{\B(\I_1(\H_\sys), \I_1(\H_\sys\otimes \H_\env))}=\|\xi\|_1$.
		\item For any Hilbert spaces $\H$, the maps $(\alpha,A)\mapsto \alpha A : \C\times \B(\H)\mapsto \B(\H)$ and $(A,B)\mapsto AB: \B(\H)\times\B(\H)\rightarrow \B(\H)$ are bilinear, thus $C^\infty$, and the map $A\mapsto e^{A}: \B(\H)\mapsto \B(\H)$ is $C^\infty$ as well.
		\item The map from $\B(\H_T)\times \B\Big(\I_1(\H_\env), \B\big(\I_1(\H_\sys), \B(\H_T)\big)\Big)\times B(\H_T)\rightarrow \B\Big(\I_1(\H_\env), \B\big(\I_1(\H_\sys), \B(\H_T)\big)\Big)$
      such that $(A,P,B)\mapsto APB$ is well defined and trilinear, which makes it $C^\infty$.
      Indeed, for any $(\eta,\xi)\in \I_1(\H_\sys)\times \I_1(\H_\env)$, we have
		$APB(\xi)=AP_\xi B: \eta\mapsto AP_\xi(\eta)B$. The trace norm of the latter in $\H_\T$ is
		\begin{align*}
		\|AP_\xi(\eta)B\|_{1} &\leq \|A\|_{\B(\H_T)} \|P_\xi(\eta)\|_{\I_1(\H_T)} \|B\|_{\B(\H_T)} \\
          &= \|A\|_{\B(\H_T)} \|P_\xi(\cdot)\|_{B(\I_1(\H_\sys),\I_1(\H_T))} \|B\|_{\B(\H_T)}\|\eta\|_1.
		\end{align*}
		This yields
      \begin{align*}
         &\|APB\|_{\B(\I_1(\H_\env), \B(\I_1(\H_\sys), \I_1(\H_T)))} \\
         &\qquad\qquad \leq \|A\|_{\B(\H_T)} \|P_\cdot(\cdot)\|_{\B(\I_1(\H_\env),\B(\I_1(\H_\sys)),\I_1(\H_T))}\|B\|_{\B(\H_T)}
      \end{align*}
      and boundedness of the trilinear map.
		\item We saw that the map $\tr_\env : \I_1(\H_\sys\otimes \H_\env)\rightarrow \I_1(\H_\sys)$ is a linear contraction, hence it is $C^\infty$.
	\end{enumerate}

	Consequently, we get that $(\alpha, Y, U^{(\tau)}, \xi)\mapsto \tr_\env(e^{\alpha Y}U^{(\tau)}(\cdot\otimes \xi)e^{-\alpha Y}U^{(-\tau)})$ is a
	$C^{\infty}$  map from $\C\times \B(\H_\T)\times \B(\H_\T)\times \I_1(\H_\env)$ to $\B(\I(\H_\sys))$. The hypotheses made on the $s$-dependence of~$Y$,~$U^{(\tau)}$ and~$\xi$ yield the result.
\end{proof}

\bibliographystyle{amsalpha}
\bibliography{researchbib}

\end{document}